%% file: main.tex
\documentclass[10pt,journal]{IEEEtran}
\usepackage[a-2b]{pdfx}
\usepackage{balance}
\usepackage{graphicx}
\usepackage{orcidlink}
\usepackage{subcaption}
\usepackage{array}
\newcolumntype{M}[1]{>{\arraybackslash}m{#1}}
\usepackage{comment}
\usepackage{amsmath}
\usepackage{amsthm}
\usepackage{url}
\usepackage{epsfig}
\usepackage{savesym}
\usepackage{verbatim}
\usepackage{multirow}
\usepackage{epstopdf}
\usepackage{caption}
\usepackage{subcaption}
\usepackage{lipsum}
\usepackage{setspace}
\usepackage{etoolbox}
\usepackage{tikz}
\usepackage{diagbox}
\usepackage{amsfonts}
\usetikzlibrary{automata, arrows, calc}
\usepackage{tkz-graph}

\usepackage{enumitem}
\usepackage{mathrsfs}
\usepackage[show]{chato-notes}
\usepackage{pifont}
\usepackage{makecell}
\usepackage{hyperref}
\usepackage{titlesec}
\usepackage{newtxmath}
\usepackage[T1]{fontenc} 

\usepackage[linesnumbered,algoruled,boxed,lined,noend]{algorithm2e}
\usepackage{breqn}
 
\SetVlineSkip{0pt}

\SetCommentSty{mycommfont}

\usepackage{bm}

\newtheorem{theorem}{Theorem}[section]

\newtheorem{definition}{Definition}[section]
\newtheorem{lemma}[theorem]{Lemma}

\newtheorem{example}{Example}[section]

\newtheorem{property}{Property}

\usepackage{enumitem}
\begin{document}

\input{definition}


\title{Efficient Densest Flow Queries in Transaction \\ Flow Networks (Complete Version)}

\author{
Jiaxin Jiang\orcidlink{0000-0001-8748-3225}, Yunxiang Zhao\orcidlink{0009-0000-3435-0832}, Lyu Xu\orcidlink{0000-0002-8999-0623}, Byron Choi\orcidlink{0000-0002-8381-336X}, \textit{Fellow}, \textit{HKIE} \\ Bingsheng He\orcidlink{0000-0001-8618-4581}, \textit{Fellow}, \textit{IEEE}, Shixuan Sun\orcidlink{0000-0003-4060-9438}, Jia Chen\orcidlink{0009-0003-1174-0063}
\IEEEcompsocitemizethanks{\IEEEcompsocthanksitem Jiaxing Jiang and Bingsheng He are with National University of Singapore, Singapore. E-mail: \{jiangjx,hebs\}@comp.nus.edu.sg
\IEEEcompsocthanksitem Yunxiang Zhao and Byron Choi are with Hong Kong Baptist University\textbf{}, Hong Kong. E-mail: \{csyxzhao, bchoi\}@comp.hkbu.edu.hk
\IEEEcompsocthanksitem Lyu Xu is with Nanyang Technological University, Singapore. Email: lyu.xu@ntu.edu.sg. \protect
\IEEEcompsocthanksitem Shixuan Sun is with Shanghai Jiao Tong University, Shanghai, China. E-mail: sunshixuan@sjtu.edu.cn
\IEEEcompsocthanksitem Jia Chen is with Grab Holdings, Singapore. Email: jia.chen@grab.com \protect \\ 
Corresponding author: Yunxiang Zhao. 
}
}

\maketitle

\input{0-abstract}

\input{1-introduction-v1}

\input{2-preliminary}

\input{2.5-overview.tex}

\input{3-reduction}

\input{4-baseline}

\input{6-greedy.tex}

\input{7-experiments}

\input{8-relatedworks}

\input{9-conclusion}

\bibliographystyle{abbrv}
\bibliography{ref}

\input{bio}

\input{appendix}

\end{document}

%% file: definition.tex

\newcommand{\byronnotes}[1]{\textcolor{blue}{\noindent Note: #1}}
\newcommand{\choi}[1]{\textcolor{red}{#1}}
\newcommand{\jiaxin}[1]{\textcolor{black}{#1}}
\newcommand{\jx}[1]{\textcolor{red}{[[JX: \textbf{#1}]]}}
\newcommand{\revise}[1]{#1}
\newcommand{\tr}[1]{}
\newcommand{\lyu}[1]{\textcolor[rgb]{0.6, 0.4, 0.8}{#1}}
\newcommand{\yunxiang}[1]{#1}
\newcommand{\yxzhao}[1]{\textcolor[rgb]{0.7,0.3,0.9}{#1}}
\newcommand{\TODO}[1]{\textcolor{BurntOrange}{\textbf{TODO:} #1}}
\newcommand{\Remind}[1]{\textcolor{RubineRed}{\textbf{Reminder:}#1}}
\newcommand{\byronsuggestion}[1]{\textcolor{Green}{\textbf{Reminder:}#1}}
\newcommand{\red}[1]{{#1}}
\newcommand{\eat}[1]{}
\newcommand{\kw}[1]{{\ensuremath {\textsf{#1}}}\xspace}
\newenvironment{tbi}{\begin{itemize}
		\setlength{\topsep}{0.6ex}\setlength{\itemsep}{0ex}} 
	{\end{itemize}} 
\newcommand{\ei}{\end{itemize}}

\newcommand{\PRADS}{\textsf{\small PADS}}
\newcommand{\KPADS}{\textsf{\small KPADS}}
\newcommand{\BPADS}{\textsf{\small BPADS}}
\newcommand{\ADS}{\textsf{\small ADS}}
\newcommand{\PKD}{\textsf{\small PKD}}
\newcommand{\FRAMEWORK}{\kw{FRAMEWORK\_NAME}}
\newcommand{\ALGO}{\kw{ALGO\_NAME}}
\newcommand{\PEval}{\kw{PEval}}
\newcommand{\IncEval}{\kw{IncEval}}
\newcommand{\Assemble}{\kw{Assemble}}
\newcommand{\ARef}{\textsf{\small ARefine}}
\newcommand{\ACmpl}{\textsf{\small AComplete}}
\newcommand{\DataGraph}{$G$}
\newcommand{\Ghier}{\mathbb{G}}
\newcommand{\qhier}{\mathbb{Q}}
\newcommand{\dist}{\mathsf{dist}}
\newcommand{\Next}{\mathsf{next}}
\newcommand{\distance}{\mathsf{d}}
\newcommand{\reach}{\textsc{reach}}
\newcommand{\Specialization}{\mathsf{Spec}}
\newcommand{\desc}{\mathsf{desc}}
\newcommand{\support}{\mathsf{sup}}
\newcommand{\distort}{\mathsf{DT}}
\newcommand{\degree}{\mathsf{deg}}
\renewcommand{\equiv}{\mathsf{equiv}}
\newcommand{\equivv}[1]{[#1]_\mathsf{equiv}}
\newcommand{\irchi}[2]{\raisebox{\depth}{$#1\chi$}}
\newcommand{\Summarization}{\mathpalette\irchi\relax}
\newcommand{\Configuration}{C}
\newcommand{\radius}{d_{max}}
\newcommand{\Eval}{\mathsf{eval}}
\newcommand{\peval}{\mathsf{peval}}
\newcommand{\F}{{\cal F}}
\newcommand{\Score}{\mathsf{scr}}
\newcommand{\ie}{\emph{i.e.,}\xspace}
\newcommand{\eg}{\emph{e.g.,}\xspace}
\newcommand{\wrt}{\emph{w.r.t.}\xspace}
\newcommand{\aka}{\emph{a.k.a.}\xspace}
\newcommand{\resp}{\emph{resp.}\xspace}
\newcommand{\stlength}{$l$}
\newcommand{\equi}{\mathsf{equi}}
\newcommand{\sn}{\mathsf{sn}}
\newcommand{\METIS}{\mathsf{\small METIS}}
\newcommand{\Seq}{\mathsf{Seq}}

\newcommand{\maxsf}{\mathsf{max}}
\newcommand{\PIE}{\mathsf{PIE}}
\newcommand{\PINE}{\mathsf{PINE}}

\newcommand{\threshold}{\tau}

\newcommand{\cost}{\mathsf{cost}}
\newcommand{\costq}{\mathsf{cost}_\mathsf{q}}
\newcommand{\CR}{\textsc{VCP}}
\newcommand{\DT}{\mathsf{distort}}
\newcommand{\FP}{\mathsf{fp}}
\newcommand{\maxSAT}{\mathsf{maxSAT}}
\newcommand{\True}{\mathsf{T}}
\newcommand{\False}{\mathsf{F}}
\newcommand{\OptGen}{\mathsf{OptGen}}
\newcommand{\freq}{\mathsf{freq}}

\newcommand{\match}{\mathsf{match}}

\newcommand{\vpd}{V_{pd}}
\newcommand{\vtp}{V_{tbp}}
\newcommand{\BiGindex}{\mathsf{BiG\textnormal{-}index}}
\newcommand{\filter}{\mathsf{filter}}
\newcommand{\ans}{\mathsf{ans\_graph\_gen}}
\newcommand{\Azero}{\mathbb{A}}
\newcommand{\boost}{\mathsf{boost}}
\newcommand{\bkws}{\mathsf{bkws}}
\newcommand{\fkws}{\mathsf{fkws}}
\newcommand{\rkws}{\mathsf{rkws}}
\newcommand{\dkws}{\mathsf{dkws}}
\newcommand{\knk}{\mathsf{knk}}
\newcommand{\ksp}{\mathsf{ksp}}
\newcommand{\config}{\mathsf{config}}
\newcommand{\content}{\mathsf{isKey}}
\newcommand{\pcnt}{\mathsf{pcnt}}
\newcommand{\private}{\textsf{isPrivate}}
\newcommand{\Path}{{\mathcal{P}}}
\renewcommand{\P}{{\mathcal P}}

\newcommand{\C}{\widehat{C}}
\newcommand{\V}{\widehat{V}}
\newcommand{\E}{\widehat{E}}

\newcommand{\ppkws}{\textsf{\small PPKWS}\xspace}

\newcommand{\DKWS}{\textsf{\small DKWS}\xspace}
\newcommand{\kDKWS}{\textsf{\small $k$DKWS}\xspace}
\newcommand{\SKWS}{\textsf{\small BFKWS}\xspace}
\newcommand{\KWS}{\textsf{\small KWS}\xspace}
\newcommand{\VU}{\mathbb{V}}
\newcommand{\VI}{\mathcal{V}}
\newcommand{\VM}{\bar{\mathcal{V}}}
\newcommand{\vsf}{\mathsf{v}}
\newcommand{\usf}{\mathsf{u}}
\newcommand{\answerset}{\mathcal{A}}
\newcommand{\candanswerset}{\bar{\mathcal{A}}}
\newcommand{\prune}{S}
\newcommand{\Ud}{\hat{\dist}}
\newcommand{\Ld}{\check{\dist}}
\newcommand{\invert}{\mathscr{I}}
\newcommand{\MB}{u.b}
\newcommand{\MF}{\mathsf{f}}
\newcommand{\Queue}{\mathcal{P}}
\newcommand{\Visit}{\mathsf{Vis}}
\newcommand{\invertV}{V_{\invert}}
\newcommand{\invertE}{E_{\invert}}
\newcommand{\tnormal}[1]{\textnormal{#1}}

\newcommand{\DKWSBF}{\textsf{\small BF}\xspace}
\newcommand{\DKWSNP}{\textsf{\small BF+PADS+NP}\xspace}
\newcommand{\DKWSPADS}{\textsf{\small BF+PADS}\xspace}
\newcommand{\DKWSPINE}{\textsf{\small BF+ALL}\xspace}
\newcommand{\Notify}{\mathsf{Notify}}
\newcommand{\Parameters}{\mathscr{X}}
\newcommand{\grape}{\kw{GRAPE}}
\newcommand{\Buffer}{\mathbb{B}}
\newcommand{\SI}{\kw{SI}}

\newcommand{\SBGindex}{\mathsf{SBGIndex}}

\newcommand{\SGIndex}{\mathsf{SGIndex}} 

\newcommand{\Portal}{\mathbb{P}}
\newcommand{\rclique}{\mathsf{r\textnormal{-}clique}}
\newcommand{\Blinks}{\mathsf{Blinks}}
\newcommand{\Rclique}{\mathsf{Rclique}}

\newcommand{\pprclique}{\mathsf{PP\textnormal{-}r\textnormal{-}clique}}
\newcommand{\ppknk}{\mathsf{PP\textnormal{-}knk}}
\newcommand{\ppBlinks}{\mathsf{PP\textnormal{-}Blinks}}

\newcommand{\baselinerclique}{\mathsf{Baseline\textnormal{-}r\textnormal{-}clique}}
\newcommand{\baselineknk}{\mathsf{Baseline\textnormal{-}knk}}
\newcommand{\baselineBlinks}{\mathsf{Baseline\textnormal{-}Blinks}}

\newcommand{\stitle}[1]{\vspace{0.4ex}\noindent{\bf #1}}
\newcommand{\etitle}[1]{\vspace{0.8ex}\noindent{\underline{\em #1}}}
\newcommand{\eetitle}[1]{\vspace{0.6ex}\noindent{{\em #1}}}

%
\newcommand{\techreport}[2]{#2}
\newcommand{\SGFrame}{\mathsf{SGFrame}} 

\newcommand{\stab}{\rule{0pt}{8pt}\\[-2.0ex]}
\newcommand{\tab}{\hspace{4ex}}

\newcommand{\Q}{{\cal Q}}



\newcommand{\eop}{\hspace*{\fill}\mbox{\qed}}

\setlength{\floatsep}{0.1\baselineskip plus 0.1\baselineskip minus 0.1\baselineskip}
\setlength{\textfloatsep}{0.1\baselineskip plus 0.1\baselineskip minus 0.1\baselineskip}
\setlength{\intextsep}{0.1\baselineskip plus 0.1\baselineskip minus 0.1\baselineskip}
\setlength{\dbltextfloatsep}{0.1\baselineskip plus 0.05\baselineskip minus 0.05\baselineskip}
\setlength{\dblfloatsep}{0.1\baselineskip plus 0.1\baselineskip minus 0.1\baselineskip}


\titlespacing*{\subsection}{0pt}{0.3\baselineskip}{0.1\baselineskip}


\newcommand{\Src}{S}
\newcommand{\Dst}{T}
\newcommand{\T}{\mathcal{T}}
\newcommand{\SD}{At least $k$ $\mathsf{S}$-$\mathsf{T}$ maximum-flow}
\newcommand{\SDMF}{STDF}
\newcommand{\WCC}{$\mathsf{WCC}$}
\newcommand{\SDMFG}{$\mathsf{kSTMF}^g$}
\newcommand{\TEMSDMF}{\ensuremath{\mathcal{T}\text{-}\mathsf{STDF}}}
\newcommand{\STcore}{$\mathsf{ST}$-FCore}
\newcommand{\Core}{\mathsf{Core}}
\newcommand{\MFlow}{\textnormal{\textsc{MFlow}}}
\newcommand{\Def}{\textnormal{\textsc{Def}}}
\newcommand{\MFavg}{MF}
\newcommand{\DKS}{\mathsf{DkS}}
\newcommand{\CBB}{C^T}
\newcommand{\fT}{f^T}
\newcommand{\tsf}{\mathsf{t}}
\newcommand{\csf}{\mathsf{c}}
\newcommand{\fsf}{\mathsf{f}}
\newcommand{\Tsf}{\mathsf{T}_{max}}
\newcommand{\algo}{\mathsf{algo}}
\newcommand{\FnDense}{\mathsf{FnDense}}
\newcommand{\FnSparse}{\mathsf{FnSparse}}
\newcommand{\Transform}{\widehat{G}}
\newcommand{\TFNet}{$\mathsf{TFN}$}
\newcommand{\RTFNet}{$\mathsf{RTFN}$}
\newcommand{\Baseline}{\mathsf{Baseline}}
\newcommand{\BTF}{$\mathsf{DIN}$-$\mathsf{TF}$}
\newcommand{\BRTF}{$\mathsf{DIN}$-$\mathsf{RTF}$}
\newcommand{\Greedy}{$\mathsf{PEEL}$}
\newcommand{\FCT}{$\mathsf{FCT}$}
\newcommand{\GDYWCC}{$\mathsf{PEEL}$-$\mathsf{DC}$}
\newcommand{\BWCC}{$\mathsf{DC}$}
\newcommand{\Tran}{\mathsf{TR}}
\newcommand{\ts}{\tau}
\newcommand{\name}{\mathsf{Conan}} 
\newcommand{\gfg}{GFG}
\newcommand{\FF}{$\mathsf{Flow}$-$\mathsf{Force}$}
\newcommand{\VPFF}{$\mathsf{Vertex}$ $\mathsf{Pair}$ $\mathsf{Flow}$-$\mathsf{Force}$}
\newcommand{\RDV}{\mathsf{RDV}}
\newcommand{\dego}{\mathsf{deg}^{\mathsf{out}}}
\newcommand{\degi}{\mathsf{deg}^{\mathsf{in}}}
\newcommand{\PF}{\mathsf{PF}}
\newcommand{\LPF}{\mathsf{LPF}}
\newcommand{\DA}{\mathsf{DA}}
\newcommand{\PR}{\mathsf{PR}}
\newcommand{\Aux}{\mathsf{Aux}}
\newcommand{\DF}{\mathsf{DF}}
\newcommand{\DFP}{\mathsf{DF'}}
\newcommand{\DFArray}{\textbf{{DF}}}
\newcommand{\D}{\revise{\mathcal{D}}}
\newcommand{\DNF}{\mathcal{DNF}}
\newcommand{\Grab}{Grab}
\newcommand{\NFT}{NFT}

\newcommand{\Push}{\mathsf{Push}}
\newcommand{\Relabel}{\mathsf{Relabel}}
\newcommand{\h}{h}
\newcommand{\excess}{\mathsf{excess}}
\newcommand{\Spade}{$\mathsf{Spade}$}

\newtheorem{manualtheoreminner}{Theorem}
\newenvironment{manualtheorem}[1]{%
  \renewcommand\themanualtheoreminner{#1}%
  \manualtheoreminner
}{\endmanualtheoreminner}

\newtheorem{manuallemmainner}{Lemma}
\newenvironment{manuallemma}[1]{%
  \renewcommand\themanuallemmainner{#1}%
  \manuallemmainner
}{\endmanualtheoreminner}

\newtheorem{manualpropertyinner}{Property}
\newenvironment{manualproperty}[1]{%
  \renewcommand\themanualpropertyinner{#1}%
  \manualpropertyinner
}{\endmanualpropertyinner}



%% file: 0-abstract.tex
\begin{abstract}
Transaction flow networks are crucial in detecting illicit activities such as wash trading, credit card fraud, cashback arbitrage fraud, and money laundering. \revise{Our collaborator, Grab, a leader in digital payments in Southeast Asia, faces increasingly sophisticated fraud patterns in its transaction flow networks. In industry settings such as Grab’s fraud detection pipeline, identifying fraudulent activities heavily relies on detecting dense flows within transaction networks. Motivated by this practical foundation,} we propose the \emph{\(S\)-\(T\) densest flow} (\SDMF{}) query. Given a transaction flow network \( G \), a source set \( \Src \), a sink set \( \Dst \), and a size threshold \( k \), the query outputs subsets \( \Src' \subseteq \Src \) and \( \Dst' \subseteq \Dst \) such that the maximum flow from \( \Src' \) to \( \Dst' \) is densest, with \(|\Src' \cup \Dst'| \geq k\). Recognizing the NP-hardness of the \SDMF{} query, we develop an efficient divide-and-conquer algorithm, $\name$. \revise{Driven by industry needs for scalable and efficient solutions}, we introduce an approximate flow-peeling algorithm to optimize the performance of $\name$, enhancing its efficiency in processing large transaction networks. \revise{Our approach has been integrated into Grab's fraud detection scenario, resulting in significant improvements in identifying fraudulent activities.} Experiments show that $\name$ outperforms baseline methods by up to three orders of magnitude in runtime and more effectively identifies the densest flows. We showcase $\name$'s applications in fraud detection on transaction flow networks from our industry partner, Grab, and on non-fungible tokens (NFTs).
\end{abstract}

%% file: 1-introduction-v1.tex
\section{Introduction}\label{sec:intro}

\IEEEPARstart{T}{ransaction} flow networks, such as Bitcoin networks ~\cite{wu2020detecting} and Ethereum networks~\cite{wood2014ethereum,luo2025rich}, have become increasingly prevalent in various applications, including e-payment systems~\cite{cao2019titant} and cryptocurrency exchanges~\cite{kondor2018principal}. These networks are vulnerable to exploitation by fraudsters for illicit activities, including wash trading~\cite{cong2022crypto}, credit card fraud~\cite{delamaire2009credit}, and money laundering~\cite{weber2019anti,hu2019characterizing,moser2013inquiry,colladon2017using}. In industry, detecting such fraudulent activities is a critical and challenging task due to the scale and complexity of transaction networks. Our industry partner, Grab, has encountered sophisticated fraud patterns that traditional methods fail to detect efficiently. \jiaxin{Beyond the Grab setting, dense and temporally constrained fund-flow patterns also serve as a core analytic primitive in anti–money-laundering (AML) systems~\cite{shadrooh2024smotef,wu5076472shadows}, NFT wash-trading detection~\cite{song2023abnormal,tovsic2025beyond}, credit-card mule-ring identification~\cite{Modepalli2025FraudRings,li2020flowscope}, and cross-border layering analysis~\cite{huang2024graph}. Unlike machine-learning classifiers~\cite{di2024amatriciana,dou2020enhancing} that operate on individual nodes, these applications require reasoning about global multi-source--multi-sink flow concentration patterns, which motivates a principled graph-analytic formulation.} To address the industry challenges, we propose a novel densest flow query that is designed to enhance fraud detection capabilities in practical applications. We first demonstrate the query with the example of money laundering detection.

\stitle{Query for Money Laundering.} Money laundering, a financial crime, unfolds in three phases: placement, layering, and integration~\cite{villanyi2021money,soudijn2012removing,kute2021deep,li2020flowscope}. It commences with \textit{source accounts} receiving illicit funds, which are subsequently transferred to \textit{sink accounts} (\aka destination accounts) to obscure their illegal origins. A critical element in the layering phase is the extensive use of money mules—individuals often in financially vulnerable situations—who conduct numerous transactions to help conceal the funds' transfer~\cite{villanyi2021money}. Despite existing methods like~\cite{lyu25deltabflow} that monitor short-term, high-volume flows between specific source and sink pairs, significant gaps remain when compared to real-world industry scenarios. \revise{In industry, dense flow analysis has become a widely used approach for identifying fraudulent patterns, focusing on the concentration of transaction flow within a group rather than individual source-sink pairs. This approach shifts the focus from isolated transactions to detecting suspicious activity among interconnected accounts, enabling the identification and differentiation of fraudsters within groups of suspect accounts.} These intermediaries, while moving large sums of money, create an intricate network of transactions that significantly obscures the trail of the original illicit funds. Fraudulent activities often involve substantial transfers within a few accounts, resulting in a pattern of dense financial flows~\cite{li2020flowscope}. We refine the concept of density as detailed in~\cite{jiang2023spade,hooi2016fraudar,shin2017densealert} to define flow density. In particular, the density of flow from a set of sources to a set of sinks is determined by dividing the flow value by the number of sources and sinks. The occurrence of dense flow among specific sources and sinks is not limited to money laundering; it is also prevalent in other fraudulent activities, such as wash trading and credit card fraud, within transaction flow networks~\cite{serneels2023detecting,gao2020tracking}. \jiaxin{Such dense multi-hop fund movements cannot be captured by node-level or edge-level ML/GNN classifiers, which do not model the temporally ordered, multi-source--multi-sink propagation of value. Detecting the densest temporally valid flow therefore constitutes a structural, flow-topological problem rather than a predictive ML task, and motivates the formulation of the \SDMF{} query.}

\stitle{Fraud Detection Pipelines in Grab.} Detecting fraudulent actors in transaction flow networks is a critical challenge for both academia and industry. \revise{Densest flow queries have become essential tools for identifying and isolating fraudulent activity patterns~\cite{li2020flowscope,tariq2023topology,gao2020tracking,granados2022geometry}.} Grab's fraud detection pipeline, illustrated in Figure~\ref{fig:intro}, consists of the following steps: (a) \emph{Network Construction}—the transaction flow network \(G\) where vertices represent users, merchants, or cards, and edges represent transactions with capacities reflecting transaction amounts. \revise{Risk control specialists identify sets of accounts exhibiting suspicious behavior, such as abnormal transaction patterns or significant outgoing or incoming flows (labeled as $\Src$ and $\Dst$, respectively);} (b) \emph{Temporal Dependency}—analysis of temporal dependencies to capture evolving fraud patterns; (c) \emph{Densest Flow Detection}—identification of dense flow patterns indicative of fraudulent behavior using densest flow queries; and (d) \emph{Applications and Operational Requirements}—various fraud detection applications such as credit card fraud, wash trading, and cashback arbitrage fraud, along with operational needs like scalability and response time. Similar pipelines are employed by other organizations~\cite{li2020flowscope,ye2021gpu,jiang2023spade}.

    \begin{figure*}[tb]
        \includegraphics[width=1\linewidth]{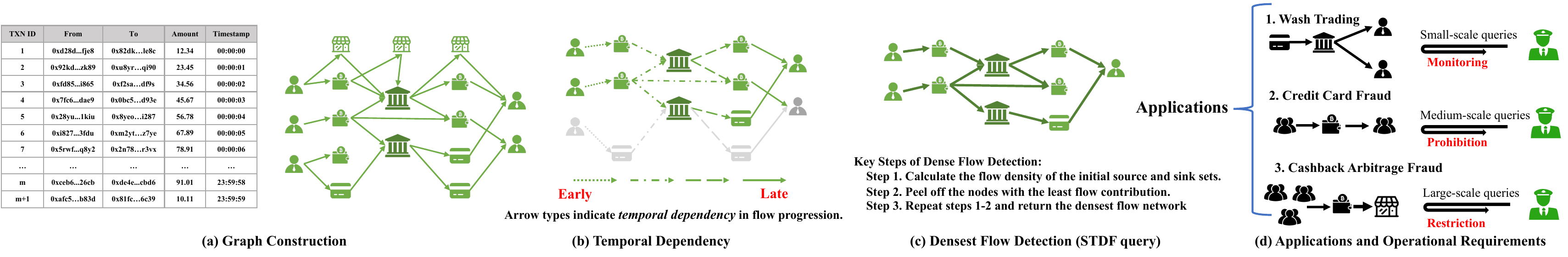}
        \vspace{-1.5em}
        \caption{Densest Flow Query on Transaction Flow Networks and Example Applications in Grab}\label{fig:intro}
    \end{figure*}

Motivated by this fraud detection pipeline, we propose the \textit{$S$-$T$ densest flow} (\SDMF{}) query, specifically designed to enhance fraud detection in transaction flow networks. This query addresses the unique challenges posed by such networks. Given disjoint sets $\Src$ and $\Dst$ representing sources and sinks, \SDMF{} seeks subsets $\Src'\subseteq \Src$ and $\Dst'\subseteq \Dst$ that maximize flow density while ensuring that the combined size of these subsets meets or exceeds a minimum threshold \( k \) (i.e., \( |\Src'| + |\Dst'| \geq k \)). This criterion is crucial, as it provides flexibility to detect both key players and peripheral members of fraud organizations, enhancing robustness against evasion tactics like creating excess accounts to dilute flows. Different applications vary in scale: for example, wash trading typically involves smaller groups, credit card fraud operations are medium-sized, and cashback arbitrage fraud schemes often involve a large number of accounts. By adjusting the minimum size threshold \( k \), \SDMF{} can detect fraudulent groups across these varying scales.

\stitle{Challenges.} Addressing \SDMF{} queries presents two major challenges. \textit{First}, these queries are NP-hard (Section~\ref{sec-pre}). A naïve approach would require enumerating all possible combinations of group sizes for $S$ and $T$ to find the optimal subsets of $\Src$ and $\Dst$, involving numerous maximum flow computations. However, maximum flow algorithms themselves have high computational complexities. On large-scale graphs like the Ethereum transaction network, which processed $11.7$ million transactions per month in 2021~\cite{wood2014ethereum}, frequently invoking such computationally intensive algorithms is impractical. \textit{Second}, transaction flow networks, especially those involving fraudulent activities, exhibit unique temporal dependencies. Fraudulent flows often occur in specific temporal patterns. Traditional maximum flow algorithms cannot capture these temporal dependencies. Therefore, computing maximum flows requires novel algorithms that account for temporal constraints, adding another layer of complexity to the problem.

This paper presents several key contributions to the field of fraud detection in transaction networks, including:
\begin{enumerate}[leftmargin=*] \item \textbf{Novel Query Formulation}: We introduce the \SDMF{} query, motivated by real-world fraud detection scenarios~\cite{li2020flowscope,tariq2023topology,ye2021gpu} and the needs of our industry partner to identify dense flows in transaction networks.
\item \textbf{Divide-and-Conquer Algorithm}: We develop a divide-and-\textbf{CON}quer \textbf{A}lgorithm for de\textbf{N}sest flow queries ($\name$), that efficiently enumerates subsets of $\Src$ and $\Dst$ to solve densest flow queries.
\item \textbf{Peeling Algorithm with Theoretical Guarantee}: We propose a peeling algorithm with a 3-approximation theoretical guarantee, significantly reducing the number of maximum flow computations to $(|\Src|+|\Dst|)^2$ instances.
\item \textbf{Pruning Techniques}: We introduce pruning methods to optimize the performance of the 3-approximation algorithm, further enhancing computational efficiency.
\item \revise{\textbf{Industrial Deployment and Impact}}: \revise{Our methods have been deployed in Grab's detection pipeline, leading to the discovery of significant fraudulent activities, including credit card fraud and cashback arbitrage fraud. Moreover, $\name{}$ has shown substantial impact in public applications; for example, it detects fraudulent behaviors such as wash trading and money laundering in NFT networks.}
\end{enumerate}

%% file: 2-preliminary.tex

\section{Background}\label{sec-pre}

\subsection{Preliminaries}\label{subsec:prelim}

A \textbf{flow network} $G=(V,E,C)$ is a directed graph, where (a)~$V$ is a set of vertices, with each vertex $v\in V$ representing an account; (b)~$E$ $\subseteq V\times V$ is a set of edges, with $(u,v)\in E$ representing a transaction\footnote{In the case where multiple edges exist between two vertices, exemplified by $\{e_1,\ldots, e_i\}$ between vertices $u$ and $v$, various strategies can be adopted to simplify the graph. One approach is the introduction of an intermediate node for each transaction, effectively transforming graphs with multiple edges into simplified versions without them. Thus, we assume the graphs considered in this study do not contain multiple edges.}; and (c)~$C$ is a non-negative capacity mapping function, such that $C(u,v)$ is the capacity on edge $(u,v)\in E$, which is the transaction volume. For simplicity, self-loops in the flow network are not considered.

\begin{definition}[Flow]\label{def:flow}
Given a flow network $G=(V,E,C)$ with a source $s \in V$ and a sink $t \in V$, a flow from $s$ to $t$ is a mapping function $f: V \times V \to \mathbb{R}^{+}$ satisfying the following:
\begin{enumerate}
	\item \stitle{Capacity Constraint:} \textnormal{$\forall (u,v) \in E$, the flow $f(u,v)$ does not exceed $C(u,v)$, \ie $f(u,v) \leq C(u,v)$.}
	\item \stitle{Flow Conservation:} \textnormal{$\forall u \in V \setminus \{s,t\}$, the flow into $u$ equals the flow out of $u$, \ie $\sum_{v\in V} f(v,u) = \sum_{v\in V} f(u,v)$.} 
\end{enumerate}	
\end{definition}

\eat{
\begin{definition}[Flow]\label{def:flow}
Given a flow network $G=(V,E,C)$ with a source $s\in V$ and a sink $t\in V$, a flow from $s$ to $t$ is a mapping function, $f: V\times V\to \mathbb{R}^{+}$, that satisfies the following two properties:
\begin{enumerate}
	\item\label{flow:cod1} \stitle{Capacity constraint:} $\forall u,v\in V$, \red{$(u,v)\in E$}, $f(u,v) \leq C(u,v)$;
	\item\label{flow:cod2} \stitle{Flow conservation:} $\forall u\in V-\{s,t\}$, $\sum\limits_{v\in V} f(v,u) = \sum\limits_{v\in V} f(u,v)$. 
\end{enumerate}	
\end{definition}
}

The {\em value of a flow} $f$ from $s$ to $t$ is denoted by $|f| = \sum_{v\in V}f(s,v)$. Given a flow network $G$, a source $s$ and a sink $t$, the maximum flow problem aims to find a flow $f$ such that $|f|$ is maximized and the value is denoted by $\MFlow(s,t)$. The maximum flow problem can be expanded to accommodate multiple sources, $\Src$, and multiple sinks, $\Dst$, with the flow value represented by $\MFlow(\Src, \Dst)$. To approach this problem, it is known that we can introduce a super source $s'$, connecting it to each source $s_i \in S$ with an edge $(s', s_i)$ of an infinite capacity $C(s', s_i) = +\infty$. Similarly, a super sink $t'$ is created and linked to each sink $t_i \in T$ via an edge $(t_i, t')$, with an infinite capacity $C(t_i, t') = +\infty$. This ensures that $\MFlow(s', t') = \MFlow(\Src, \Dst)$~\cite{ahuja1993network}.

\eat{
\stitle{Multi-source multi-sink maximum flow.} A maximum flow problem can be extended for multiple sources $\Src$ and multiple sinks $\Dst$, denoted by $\MFlow(\Src, \Dst)$. To solve this problem, we create a super source $s'$ and add an edge $(s',s_i)$ for each $s_i\in S$ with a capacity $C(s', s_i) = +\infty$. Additionally, we create a super sink $t'$ and add an edge $(t_i,t')$ for each $t_i\in T$ with a capacity $C(t_i, t')=+\infty$. It has been proven that $\MFlow(\Src,\Dst) = \MFlow(s',t')$~\cite{ahuja1993network}.
}

Existing studies~\cite{jiang2023spade,hooi2016fraudar,shin2017densealert} overlook temporal dependency, which limits the effectiveness. Here, we extend the definitions to capture temporal dependency.

\stitle{Timestamp.} A transaction, represented as an edge, is naturally associated with an occurrence timestamp, denoted by $\ts$. Here, $\ts$ indicates the chronological order of transactions, where larger values of $\ts$ indicate more recent transactions. Then, \jiaxin{\textbf{T}ransaction \textbf{F}low \textbf{N}etwork (\TFNet{})} is formalized as follows:

\begin{definition}[\TFNet{}]\label{def:tem-flow}
A \TFNet{} $G=(V,E,C,\T)$ is a directed graph where: (a)~$V$, $E$, and $C$ retain their definitions from the flow network, and (b) $\T{}: E \to \mathbb{Z}^+$ is a timestamp mapping function assigning a timestamp to each edge $e=(u,v) \in E$, indicating the moment the transaction occurred, denoted by $\T(e)$ or $\T(u,v)$.
\end{definition}

\noindent\jiaxin{\stitle{Temporal flow.} In a \TFNet{}, a temporal flow from $s$ to $t$ is a standard flow that additionally respects \emph{temporal dependency}: for every intermediate vertex $u$ and every time $\ts$, the total amount of money that has \emph{entered} $u$ up to time $\ts$ is at least the total amount that has \emph{left} $u$ by time $\ts$. Intuitively, an account cannot spend funds \emph{before} it receives them.}

\jiaxin{Formally, we use the same capacity and flow-conservation conditions as in the static case, and enforce the above cumulative in–out constraint at every vertex and time point. For brevity, we denote the maximum temporal flow value from $s$ to $t$ by $\MFlow(s,t)$. Figure~\ref{fig:tem-network} illustrates how the temporal constraint can drastically reduce the feasible flow compared to the static case.
}



\begin{figure}[tb]
	\begin{center}
	\includegraphics[width=\linewidth]{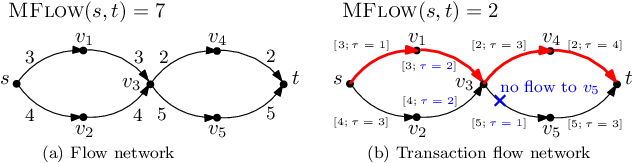}
		\end{center}
  \vspace{-1em}
	\caption{\jiaxin{A transaction flow network}}
	\label{fig:tem-network}
\end{figure}

\begin{example}
Consider the network in Figure~\ref{fig:tem-network}(a), the maximum flow value $\MFlow(s,t) = 7$. In Figure~\ref{fig:tem-network}(b), each edge is annotated with a timestamp. Referring to Definition~\ref{def:tem-flow}, observe that no flow is permissible from $v_1$ to $v_5$ via $v_3$ as the transaction from $v_1$ to $v_3$ occurs after the transaction from $v_3$ to $v_5$ ($\T(v_1, v_3) > \T(v_3,v_5)$), thereby breaching the temporal flow constraint at $\tau=1$. A similar temporal discrepancy prevents flow from $v_2$ to $v_5$ through $v_3$. The only viable temporal flow paths from $s$ to $t$ are highlighted in red, resulting in a maximum temporal flow value, \textnormal{\MFlow}$(s,t)=2$.
\end{example}

\subsection{Problem Statement}\label{subsec:background}

\begin{figure}[tb]
\centering
    \includegraphics[width=\linewidth]{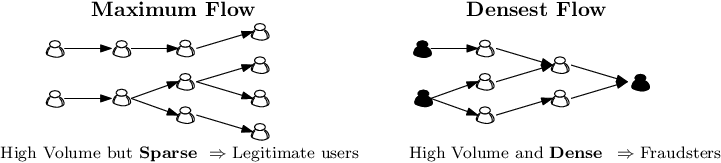}
    \caption{\jiaxin{Difference between Maximum Flow and Densest Flow}}\label{fig:densevsmax}
\end{figure}

\jiaxin{While many accounts with high-flow volumes, as depicted in Figure~\ref{fig:densevsmax}, may belong to legitimate institutions or active users, relying solely on maximum flow analysis is insufficient to distinguish fraudulent users from these legitimate entities. Previous studies have indicated that fraudulent activities often involve large-volume transfers within a small, concentrated group of participants, thereby creating a pattern of dense financial flows~\cite{li2020flowscope}. Such a pattern underscores the need for a more nuanced metric that captures not just the volume, but also the density of these flows.}

\stitle{Flow Density.} Given a set of sources $\Src$ and a set of sinks $\Dst$, the flow density from $\Src$ to $\Dst$ is defined as the maximum flow value divided by the sum of the sizes of $\Src$ and $\Dst$~\cite{jiang2023spade,hooi2016fraudar,shin2017densealert}, specifically \( g(\Src,\Dst) = \frac{\MFlow(\Src, \Dst)}{|\Src| + |\Dst|} \).

\stitle{$\mathsf{S}$-$\mathsf{T}$ Densest Flow ({\SDMF}).} Given a flow network $G$, a set $\Src$ of sources, a set $\Dst$ of sinks, and an integer $k$, \SDMF{} is to find a subset $\Src'\subseteq\Src$ and a subset $\Dst'\subseteq\Dst$, such that $|\Src'|+|\Dst'|\geq k$, and maximize the flow density, $g(\Src', \Dst')$. 

\stitle{The \SDMF{} Decision Problem.} Given a flow network $G$, a set $\Src$ of sources and a set $\Dst$ of sinks, and a parameter $c$, is there a subset $\Src'\subseteq\Src$ and a subset $\Dst'\subseteq\Dst$ such that $|\Src'|+|\Dst'| \geq k$ and $g(\Src', \Dst') \geq c$? The \SDMF{} decision problem is proven NP-complete in Lemma~\ref{lemma:nphard}, thus making the optimization problem of \SDMF{} NP-hard.  \eat{The proof can be found in Appendix~\ref{sec:proof} of \cite{techreport}, due to space limitations.}

\begin{lemma}\label{lemma:nphard}
	\SDMF{} decision problem is NP-complete.
\end{lemma}
\begin{proof}
	{\it sketch}. The proof is based on a reduction from the decision problem of the densest at least $k$ subgraph ($\DKS$) (cf. \cite{feige1997densest}). The proof is presented in~\cite{techreport}, due to space limitations.
\end{proof}

Building upon \SDMF{} and incorporating the temporal flow constraint, we introduce \TEMSDMF{}. Notably, \SDMF{} represents a particular case of \TEMSDMF{}, where the timestamp $\ts$ for all edges in the \TFNet{} is set to~$1$.

\eat{
We can expand the \SDMF{} concept to encompass the \TFNet{}, incorporating the temporal flow constraint to address temporal dependency. This extended version is denoted as \TEMSDMF{}. By assigning a value of $\ts_i=1$ to each edge in the \TFNet{}, we essentially transform it into a standard flow network, making \SDMF{} a specific case of \TEMSDMF{}. Hence, \TEMSDMF{} is also NP-hard. }

\tr{
\begin{lemma}
	\TEMSDMF{} is NP-hard. 
\end{lemma}
}
\tr{
\begin{proof}
\SDMF{} is an instance of \TEMSDMF{}. Since \SDMF{} is NP-hard, \TEMSDMF{} is NP-hard.
\end{proof}
}

\stitle{Problem Statement.} Given a \TFNet{} $G=(V,E,C,\T)$, and a \TEMSDMF{} query $Q=(\Src, \Dst, k)$, the answer of $Q$, denoted as $Q$($G$) = $(f, (\Src', \Dst'))$ of $Q$, where $\Src'\subseteq\Src$ and $\Dst'\subseteq\Dst$ is a pair of proper subsets that maximizes the density $g(\Src',\Dst')$ on $G$, and $f$ is the maximum flow from $\Src'$ to $\Dst'$. This paper aims to determine the answer of a given query, $Q$($G$).

\subsection{\jiaxin{Modeling Rationale}}\label{sec:rational}

\jiaxin{\stitle{Flow-based density.} Unlike classical edge-density~\cite{jiang2023spade,hooi2016fraudar,shin2017densealert} measures that rely solely on the sum of edge weights, transaction networks naturally obey flow capacity, flow conservation, and temporal ordering constraints. Ignoring these constraints may treat non-executable or duplicated transfers as dense patterns. Using $\MFlow(S,T)$ ensures that the detected density reflects the \emph{actual realizable amount of money} transferable from $S$ to $T$, which is essential for modeling money laundering chains, wash trading loops, and multi-hop smurfing structures observed in practice.}

\noindent \jiaxin{\textbf{Normalization choices.}
A geometric-mean normalization such as 
$\MFlow(S,T)/\sqrt{|S||T|}$ is mathematically dominated by the smaller of $|S|$ and $|T|$, which artificially inflates the density of highly unbalanced groups—a common pattern in fraud (e.g., few sources but many mule accounts). In contrast, the arithmetic term $|S|+|T|$ penalizes the two sides proportionally, yielding a more stable and size-consistent measure that reflects the collective participation of all accounts. This linear scaling also aligns with prior fraud-detection formulations~\cite{jiang2023spade,hooi2016fraudar,shin2017densealert}, making it appropriate for transaction-flow networks.
}

\noindent \jiaxin{\noindent\textbf{Why a transaction-flow network (\TFNet{}) and compatibility with multilayer models.} Although real payment ecosystems may involve multiple operational layers (e.g., cards, wallets, merchants), industrial fraud-monitoring systems (including Grab’s) standardize all records into a unified account-to-account transaction log.  The \TFNet{} abstraction therefore preserves exactly the information required for computing temporal maximum flows. Importantly, our formulation is fully \emph{compatible} with multilayer representations: a multilayer network can be flattened to a \TFNet{} by treating cross-layer edges as ordinary directed edges, or conversely, a \TFNet{} can be lifted to a multilayer model without affecting the temporal-flow constraints or the definition of the \SDMF{} objective. Thus, using a \TFNet{} does not restrict generality, while avoiding the additional semantic and computational complexity introduced by explicit multilayer models.
}

\noindent \jiaxin{ \stitle{Why $S$ and $T$ are not partitions.} In fraud scenarios, $S$ and $T$ represent accounts with net outgoing and net incoming behavior, respectively, rather than a graph bipartition. This directional interpretation is consistent with real fraudulent workflows (source $\rightarrow$ mule layers $\rightarrow$ sink), making it more appropriate than using $(S,T)$ as a generic cut.}

\noindent \jiaxin{ \stitle{Size constraints.} In fraud scenarios, the ratio between sources and sinks is highly unpredictable: a credit-card fraud case may have only a few sources but many receiving mule accounts, whereas cashback-arbitrage schemes may involve large and fluctuating groups on both sides. Because such proportions are unknown \emph{a priori}, setting two independent thresholds $(k_1,k_2)$ is operationally impractical—there is no reliable way for investigators to determine appropriate values. In contrast, a single size constraint $|S'|+|T'|\ge k$ reflects the requirement that a fraud group must contain at least $k$ participants while allowing the algorithm to automatically choose the most meaningful split between $S'$ and $T'$. This makes the formulation practical across heterogeneous fraud types.
}

%% file: 2.5-overview.tex
\vspace{-1em}
\section{Solution Overview of $\name$}\label{sec:overview}

To evaluate the \TEMSDMF{} query, we propose a two-stage solution called $\name$ (see Figure~\ref{fig:overview}). \revise{Our solution is driven by practical industry insights from fraud detection at Grab.}

\etitle{Observation 1.} Fraudulent activities exhibit temporal dependencies that differ from normal user behavior patterns. For example, in credit card fraud, fraudsters transfer funds following specific time sequences. Considering temporal dependencies increases detection accuracy from 64.93\% to 82.39\%.

\stitle{Stage 1: Network Transformation.}  $\name$ comprises a lightweight network transformation technique for evaluating \TEMSDMF{} on the transformed networks, called \RTFNet{} (detailed in Section~\ref{sec:preprocess}). This technique ensures the correctness of the maximum temporal flow and enables transforming the \TEMSDMF{} query into a \SDMF{} query (Figure~\ref{fig:overview}). It is proven that any maximum flow algorithms can be applied (without modification) to the transformed network to obtain maximum temporal flows. The answers to the \TEMSDMF{} query can be simply derived from the \RTFNet{} and the answer to the corresponding \SDMF{} query.

\etitle{Observation 2.} In scenarios like cashback arbitrage fraud, transaction flows within certain groups are extremely dense. At Grab, executing a dense flow query can take several hours, delaying fraud detection and allowing fraudulent activities to continue undetected. Efficient algorithms are essential to process these dense flows promptly.

\stitle{Stage 2: Query Evaluation.} Given a \SDMF{} query, $Q=(\Src, \Dst, k)$, $\name$ evaluates the query in a divide-and-conquer algorithm. Firstly, $\name$ divides $\Src$ (resp. $\Dst$) into subsets $\Src_i$ (resp. $\Dst_i$) \eat{and searches for the dense flow in $\Src$ and $\Dst$.} \eat{computes dense flows in each $\Src_i$ and $\Dst_i$ }(Section~\ref{sec:dec}). Then, $\name$ peels the vertices recursively in each $\Src_i$ and $\Dst_i$ (Section~\ref{sec:appr}). After finding the 3-approximate densest flows from $\Src_i$ to $\Dst_i$, $\name$ combines them to obtain the global 3-approximate densest flow.

\begin{figure}[tb]
	\begin{center}
	\includegraphics[width=0.35\textwidth]{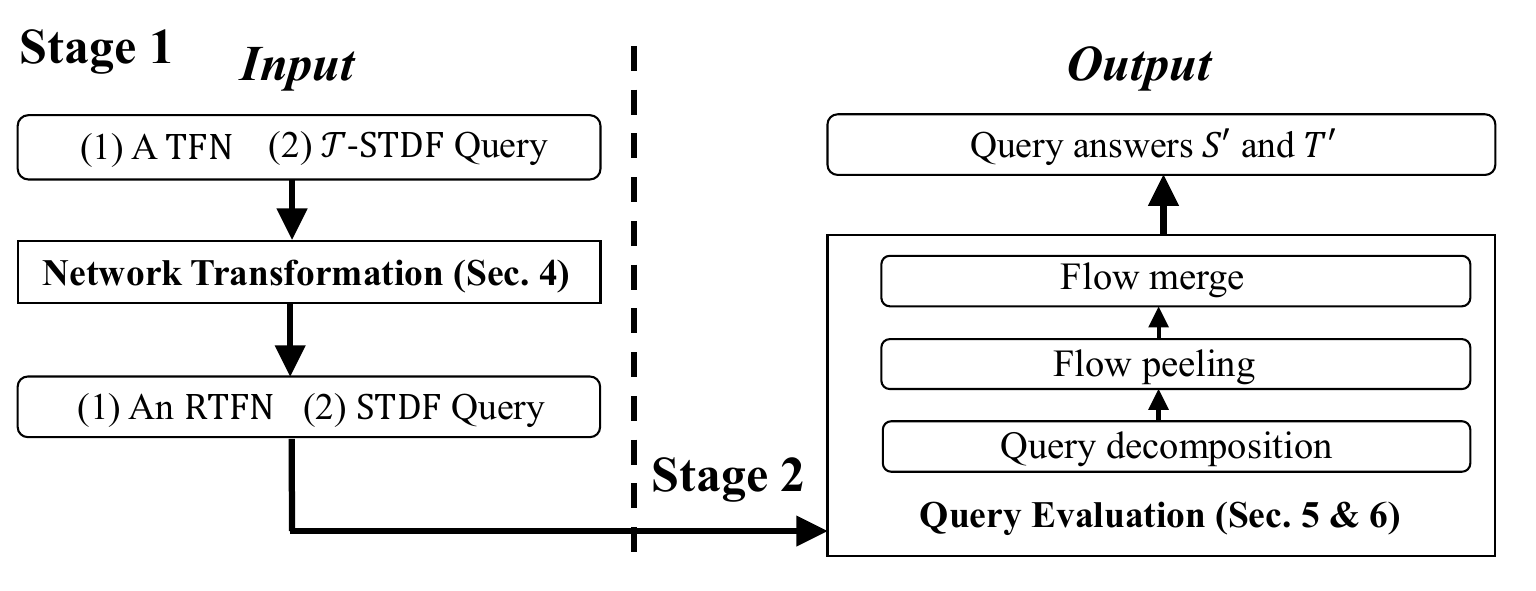}
		\end{center}
  \vspace{-1em}
	\caption{A two-stage solution overview of $\name$}
	\label{fig:overview}
\end{figure}

%% file: 3-reduction.tex
\section{Baseline for Maximum Temporal Flow}\label{sec:preprocess}

\jiaxin{Existing work on maximum flow falls into two categories, neither of which applies
to the \SDMF{} setting.}

\noindent \jiaxin{\textit{(1) Classical static max-flow algorithms.} Algorithms such as augmenting-path methods~\cite{ford1956maximal}, push--relabel~\cite{cherkassky1997implementing}, and  pseudoflow~\cite{hochbaum2008pseudoflow} operate on static networks and rely on unrestricted flow reversal in the residual graph. Such reversal violates the temporal-dependency, and therefore these algorithms cannot compute correct temporal flows (cf. Fig.~\ref{fig:augmenting_paths}).}

\noindent  \jiaxin{\textit{(2) Temporal-flow formulations.} Prior temporal-flow studies are based on problem definitions fundamentally different from ours. Earliest-arrival and transit-time models~\cite{gale1959transient,wilkinson1971algorithm,schmidt2014earliest,skutella2009introduction} optimize arrival time under traversal delays and do not enforce the “funds-must-arrive-before-they-can-leave’’ causality constraint. Similarly, models with time-varying or ephemeral edge availability~\cite{hamacher2003earliest,akrida2019temporal} assume intermittent edge activation but do not require cumulative temporal feasibility at intermediate vertices. Our temporal flow, in contrast, requires monotone non-decreasing timestamps on any feasible augmenting path and regret-enabled flow reversal to preserve such causality—constraints absent from all prior formulations. Because these methods also operate only in the single-source–single-sink, single-evaluation setting, they cannot be directly applied to \SDMF{}.}

\jiaxin{These limitations motivate our regret-enabled transformation, which converts a
\TFNet{} into a static network that preserves temporal consistency and allows
any high-performance max-flow algorithm to be used correctly.}



\stitle{Residual Network.} Given a \TFNet{} $G=(V,E,C,\T)$ and a flow $f$, the residual network is denoted as $G_f=(V_f,E_f,C_f, \T_f)$. The residual capacities $C_f$ and timestamps $\T_f$ for $(u,v) \in V \times V$ are defined as:
\begin{footnotesize}
\begin{equation}
    \begin{aligned}
        \left\langle C_f(u,v), \T_f(u,v) \right\rangle &= 
        \begin{cases}
            \left\langle C(u,v) - f(u,v), \T(u,v) \right\rangle & \text{if } (u,v)\in E \\
            \left\langle f(v,u), \T(v,u) \right\rangle & \text{if } (v,u)\in E \\
            \left\langle 0, 0 \right\rangle & \text{otherwise}
        \end{cases}
    \end{aligned}
\end{equation}
\end{footnotesize}
where $V_f = V$ and $E_f=\{(u,v)~|~C_f(u,v) > 0\}$.

Given a \TFNet{} $G=(V,E,C,\T)$, a source $s$ and a sink $t$, a path $P=(v_1,\ldots, v_n)$ is a \textit{temporal augmenting path} if a) $v_1=s$; b) $v_n=t$; and c) $\forall i\in [2,n-1]$, $\T(v_{i-1},v_i)\leq \T(v_i,v_{i+1})$ .

\begin{definition}[Regret-disabling vertex ($\RDV$)] In a \TFNet{}, $G=(V,E,C,\T)$, a vertex $u$ is designated as an $\RDV$ if $\exists \ts_1,\ts_2\in \T_u^{\mathsf{in}}, \ts_3,\ts_4\in \T_u^{\mathsf{out}}$ such that $\ts_1 < \ts_3 < \ts_2 < \ts_4$, where
    (i) $\T_u^{\mathsf{in}}=\{\T(v, u)~|~(v,u) \in E\}$, and (ii) $\T_u^{\mathsf{out}}=\{\T(u, v)~|~(u,v)\in E\}$.
\end{definition}

The presence of $\RDV$s in a \TFNet{} poses challenges for classic maximum flow algorithms. \eat{These vertices \eat{can disrupt the expected behavior of algorithms such as those based on}\yxzhao{disables} augmenting-path-based algorithms (\eg~\cite{ford1956maximal}) \yxzhao{to decrease existing flow by flowing from the opposite direction}.} To illustrate this issue, we provide a brief example below.


\begin{figure*}[tb]
	\begin{center}
	\includegraphics[width=\textwidth]{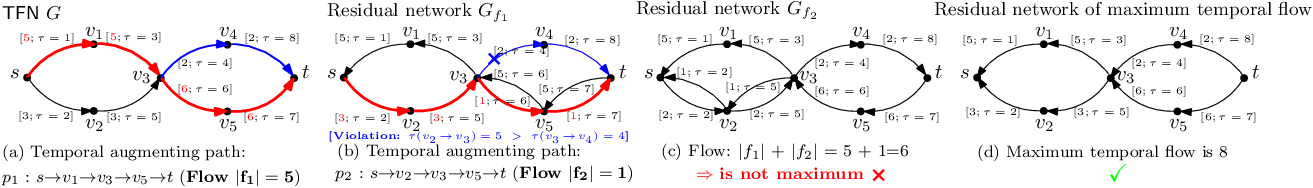}
		\end{center}
  \captionsetup{width=\linewidth}
	\caption{\jiaxin{(a)-(c) classic algorithms cannot return maximum temporal flow, and (d) Residual network of maximum temporal flow.}}
	\label{fig:augmenting_paths}
\end{figure*}

\begin{example}
In a \TFNet{} as depicted in Figure~\ref{fig:augmenting_paths}, $\T(v_1,v_3) = 3$, $\T(v_2,v_3)=5$, $\T(v_3,v_4)=4$, and $\T(v_3,v_5)=6$. $v_3$ is an $\RDV$ since $\T(v_1,v_3) < \T(v_3,v_4) < \T(v_2,v_3) < \T(v_3,v_5)$. Suppose that the first found temporal augmenting path is $p_1=(s,v_1,v_3,v_5,t)$ with a temporal flow value $|f_1|~= 5$. The residual network after finding $p_1$ results in a second temporal augmenting path, $p_2$, and the temporal flow value along $p_2$ is $|f_2|~=1$. Thus, the value of temporal flow from $s$ to $t$ found by classic algorithms is only $|f_1|~+~|f_2|~= 6$ as shown in Figure~\ref{fig:augmenting_paths}(c). Due to the temporal flow constraint, it is not possible to flow from $v_2$ to $v_4$ through $v_3$ (as shown in Figure~\ref{fig:augmenting_paths}(b)) unless part of $f_1$ is reversed. An alternate augmenting path, $p_1=(s,v_1, v_3, v_4, t)$ must be selected at the beginning (as shown in Figure~\ref{fig:augmenting_paths}(a)). The maximum temporal flow from $s$ to $t$ of this \TFNet{} is $\MFlow(s,t)=8$ if three augmenting paths are found in the order of $(s,v_1, v_3, v_4, t)$, $(s,v_1,v_3,v_5,t)$, and $(s,v_2,v_3,v_5,t)$. The residual network is shown in Figure~\ref{fig:augmenting_paths}(d).
\end{example}

To resolve this issue, we present our \TFNet{} transformation that enables augmenting-path-based algorithms to compute the correct maximum flow value.

\begin{figure}[tb]
	\begin{center}
	\includegraphics[width=0.48\textwidth]{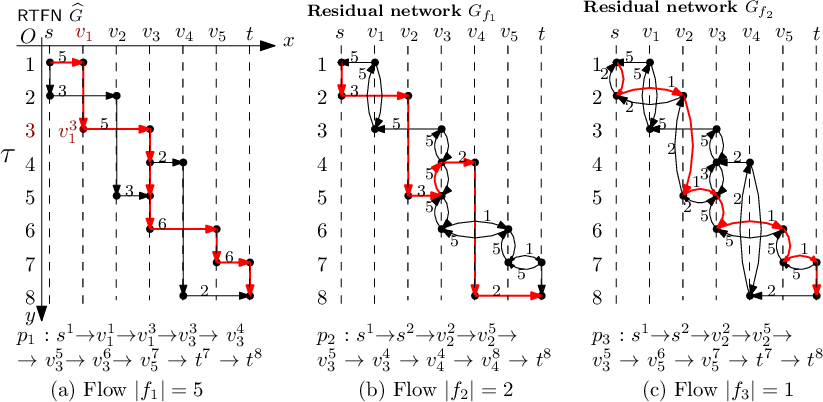}
		\end{center}
	\caption{(a) \RTFNet{} $\Transform$, (a-c) three augmenting paths $p_1$, $p_2$ and $p_3$, and (b-c) two residual networks $G_{f_1}$ and $G_{f_2}$. ({\em Remarks: The capacity on vertical edges is $+\infty$ unless specified.})}
	\label{fig:transformed}
\end{figure}

\noindent\textbf{Network Transformation Algorithm.}  Given a \TFNet{}, $G=(V,E,C,\T)$, $\name$ transforms $G$ into a \textbf{R}egret-enabled \textbf{T}emporal \textbf{F}low \textbf{N}etwork (\RTFNet) $\Transform = (\widehat{V}, \widehat{E}, \widehat{C})$ which is \textit{regret-enabled} for classic maximum flow algorithms to reverse the flow identified in previous iterations. To accomplish this, $\name$ builds a {\em virtual} Cartesian coordinate system $xOy$ with the vertices on the $x$-axis and the timestamps on the $y$-axis. The key steps, that run in $O(|V|+|E|)$, are as follows.

\begin{enumerate}[leftmargin=*]
    \item \textbf{Initialization.} $\Transform$ is initialized as an empty graph.
    \item \textbf{Iterative steps.} For each edge $e=(u,v)$ to be transformed, where $\ts = \T(e)$, a copy $u^{\ts}$ of $u$ (resp. $v^{\ts}$ of $v$) is created at the coordinate $\langle u,\ts\rangle$ (resp. $\langle v, \ts \rangle$) in the $xOy$ Cartesian coordinate system. The edge $e$ is transformed into $\hat{e}=(u^{\ts},v^{\ts})$ with capacity $\C(\hat{e}) = C(e)$. For all copies $\{u^{\ts_1}, \ldots, u^{\ts_k}\}$ of vertex $u$, additional edges $e' = (u^{\ts_i}, u^{\ts_{i+1}})$ for $i\in [1,k-1]$ are added to $\Transform$ with capacity $\C(e') = +\infty$. To ease the discussion, $\hat{e}$ is referred to as the "horizontal edge", while $e'$ is referred to as the "vertical edge".
    \item \textbf{Termination.} The transformation terminates when all vertices and edges have been transformed.
\end{enumerate}

\begin{example}
Consider the \TFNet{} in Figure~\ref{fig:augmenting_paths}(a). The transformation is shown in Figure~\ref{fig:transformed}(a). $\Transform$ is initialized as an empty graph (Step (1)). The edge $(s,v_1)$ of $G$ is transformed into $(s^1,v^1_1)$ of $\Transform$. Consider the edge $(v_1,v_3)$ in $G$, the vertices $v^3_1$ and $v^3_3$, and the horizontal edge $(v^3_1, v^3_3)$ are created in $\Transform$. The existence of $v^1_1$ also results in the addition of a vertical edge $(v^1_1,v^3_1)$ to $\Transform$ (Step (2)). The \RTFNet{} shown in Figure~\ref{fig:transformed}(a) (Step (3)) is the network after all edges are transformed.
\end{example}

We establish in Theorem~\ref{theorem:identical} that the maximum flow values in \RTFNet{} and \TFNet{} are identical. To avoid disrupting the presentation flow, we present its proof in~\cite{techreport}).

\begin{theorem}\label{theorem:identical}
    Given a \TFNet{} $G$ with a source $s$ and a sink $t$, the value of \MFlow$(s,t)$$=$\MFlow$(s^{\ts_1},t^{\ts_{\max}})$ in the \RTFNet{} $\Transform$, where the earliest copy of a source $s$ and the latest copy of a sink $t$ as $s^{\ts_1}$ and $t^{\ts_{\max}}$, respectively.
\end{theorem}
\tr{
\begin{proof}
    We prove this theorem by contradiction. Assume that the maximum flow from $s^{\ts_1}$ to $t^{\ts_{\max}}$ on \RTFNet{} is $\widehat{f}$ and the maximum flow from $s$ to $t$ on \TFNet{} is $f$. 1) If $|\widehat{f}|~<|f|$, there exists a flow $\widehat{f'}$ on \RTFNet{}, such that $|\widehat{f'}|~=|f|$ by Lemma~\ref{lemma:tfnet}. $|\widehat{f'}|~>|\widehat{f}|$ contradicts the assumption that $\widehat{f}$ is the maximum flow on \RTFNet{}. 2) If $|\widehat{f}|~>|f|$, $f$ is not the maximum flow on \TFNet{} using Lemma~\ref{lemma:rtfnet} which contradicts the assumption. Hence, we conclude that $|\widehat{f}|~=|f|$.
\end{proof}
}

We next illustrate that classic algorithms are applied on \RTFNet{} to compute the maximum temporal flow value.

\begin{example}
     Consider \RTFNet{} in Figure~\ref{fig:transformed} (a). Assume that the first augmenting path found is $p_1$ with a temporal flow of $|f_1|~=5$. The second augmenting path found is $p_2$ in the residual network $G_{f_1}$ resulting in $|f_2|~=2$. $p_2$ is found as there is a reverse edge $(v_3^5, v_3^4)$ in $G_{f_1}$ that enables the flow to be reversed. When the last augmenting path $p_3$ with a temporal flow of $|f_3|~=1$ is found, the maximum flow from $s$ to $t$ is computed by \textnormal{\MFlow}$(s,t) = |f_1| + |f_2| + |f_3| = 8$.
\end{example}

\eat{
The size of \RTFNet{} is bounded as follows.
\begin{lemma}\label{lemma:space}
    Given a \TFNet{} $G=(V,E,C,\T)$ and its \RTFNet{} $\Transform=(\V, \E, \C)$, $|\V|$ is bounded by $2|E|$ and $|\E|$ is bounded by $3|E| - |V|$.
\end{lemma}
}

\tr{
\stitle{Time and space complexity.} Given a \TFNet{}, an $O(|V|+|E|)$ traversal algorithm can compute its \RTFNet{} network $\Transform=(\V, \E, \C)$. $|\V|$ is bounded by $2|E|$ and $|\E|$ is bounded by $3|E| - |V|$.
}

\stitle{Remarks.} With Theorem~\ref{theorem:identical}, we assume that $\name{}$ transforms any source $s_i\in \Src$ (resp. any sink $t_i\in \Dst$) in the \TEMSDMF{} query to $s_i^{\tau_1}$ (resp. $t_i^{\tau_{\max}}$) in the \RTFNet{}. We will focus on the \SDMF{} query on the \RTFNet{} for simplicity.

%% file: 4-baseline.tex
\section{Divide-and-conquer Approach}\label{sec:dc}

\begin{figure*}[tb]
	\begin{center}
	\includegraphics[width=0.85\textwidth]{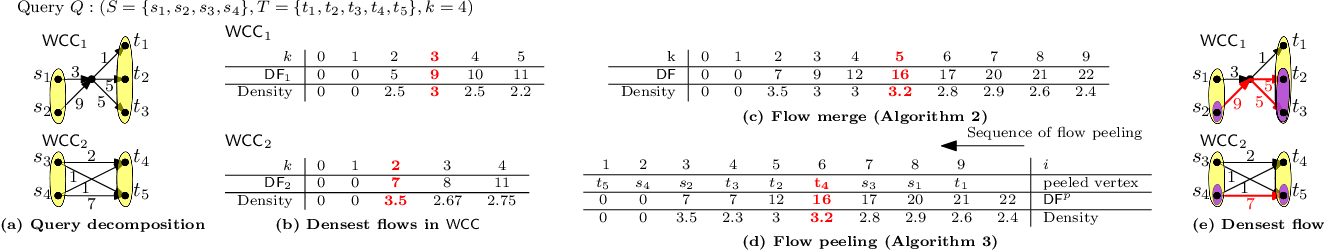}
		\end{center}
  \vspace{-3mm}
	\caption{\SDMF{} evaluation (A running example)}
	\label{fig:runningexample}
\end{figure*}

As presented in Section~\ref{subsec:background}, the density metric $g(\Src,\Dst)$ is defined as $\frac{\MFlow(\Src,\Dst)}{|\Src|+|\Dst|}$. It is evident that the maximum flow is the densest when $|\Src|+|\Dst|$ is fixed. However, it is time-consuming to enumerate all possible subsets of $\Src$ and $\Dst$ and compute the maximum flows to answer the \SDMF{} queries $Q=(\Src, \Dst, k)$ since this approach would require computing the maximum flow $2^{|\Src| + |\Dst|}$ times. Hence, we propose an efficient divide-and-conquer approach to reduce the number of maximum flow computations. Specifically, given a query $Q=(\Src, \Dst, k)$, $\name{}$ decomposes $S$ and $T$ into smaller subsets (Section~\ref{sec:dec}), finds densest flows for each fixed $|S'|+|T'|$ for each pair of subsets (Section~\ref{sec:comp-dense-flow}), and merges the densest flows to answer the query that $|S'|+|T'|\ge k$ (Section~\ref{sec:merge}).

\subsection{Query Decomposition}\label{sec:dec}

In this subsection, we propose to decompose the query subsets $\Src$ and $\Dst$ into smaller, non-overlapping subsets $\Src_i$ and $\Dst_i$. We first present how to determine these subsets. 

\stitle{Reachability.} Given a source $s$ and a sink $t$, if there is a path from $s$ to $t$, we say $s$ can reach $t$ and denote the reachability by $\reach(s,t) = \mathsf{True}$. Otherwise, $\reach(s,t) = \mathsf{False}$.

\stitle{Overlap of Flows.} Given two pairs of a source and a sink $(s_1,t_1)$ and $(s_2,t_2)$, if $\reach(s_1,t_2) = \mathsf{False}$ and $\reach(s_2,t_1) = \mathsf{False}$, the flows $f_1$ and $f_2$ are \textit{overlap-free}, where $f_1$ (resp. $f_2$) is the flow from $s_1$ (resp. $s_2$) to $t_1$ (resp. $t_2$). Otherwise, $f_1$ and $f_2$ overlap. We have the following property if two flows are overlap-free.

\begin{property}\label{property:reachability}
    {\em Given two pairs $(s_1,t_1)$ and $(s_2,t_2)$, and their maximum flow values $|f_1|~=~$\textnormal{\MFlow}$(s_1, t_1)$ and \\ $|f_2|~=~$\textnormal{\MFlow}$(s_2, t_2)$. If the flows $f_1$ and $f_2$ are overlap-free, $|f_1|~+~|f_2|~=~$\textnormal{\MFlow}$(\{s_1,s_2\}, \{t_1, t_2\})$.}
\end{property}

\eat{
\stitle{Auxiliary bipartite graph.} Inspired by property~\ref{property:reachability}, we build an auxiliary bipartite graph $G^{\Aux}=(\Src, \Dst, E^{\Aux})$ \wrt the query $Q=(\Src, \Dst, k)$. $(s,t)\in E^{\Aux}$ $\mathsf{iff}$ $\reach(s,t)=\mathsf{True}$ in $G$. To construct the bipartite graph $G^{\Aux}=(\Src, \Dst, E^{\Aux})$, any existing efficient reachability indexes can be adopted, such as ~\cite{sigmod2012,pathtree,cheng2013tf,su2016reachability}. With the state-of-the-art reachability indexing techniques, the graph is built in $O(|\Src||\Dst|\log |V|)$. 
}

\stitle{Vertex Set Decomposition.} With Property~\ref{property:reachability}, $\name$ divides the queries based on the reachability between the sources $\Src$, and the sinks $\Dst$, into subset pairs, denoted by $\mathsf{WCC}_i = (\Src_i, \Dst_i)$ ($i\in [1,\mathsf{nw}]$), where $\mathsf{nw}$ is the number of subset pairs. The sources in $\Src_i$ and the sinks in $\Dst_i$ are located within the {\em same} weakly connected component ($\mathsf{WCC}$). We remark that the sources (resp. sinks) outside the $\mathsf{WCC}$ are not reachable to $\Dst_i$ (resp. from $\Src_i$). The decomposition satisfies the following:
\begin{enumerate}
    \item $\Src = \bigcup_{i \in [1,~\mathsf{nw}]} \Src_i$ and $\Dst = \bigcup_{i \in [1,~\mathsf{nw}]} \Dst_i$; and
    \item $\forall i,j \in [1,~\mathsf{nw}]$, $\Src_i \cap \Src_j = \emptyset$, and  $\Dst_i \cap \Dst_j = \emptyset$.
\end{enumerate}

If $(\Src_1, \Dst_1)$ and $(\Src_2, \Dst_2)$ are two distinct \WCC{}s, the maximum flows $f_1$ and $f_2$ must be overlap-free, where $f_1$ (resp. $f_2$) is the maximum flow from $\Src_1$ (resp. $\Src_2$) to $\Dst_1$ (resp. $\Dst_2$). Otherwise, $f_1$ and $f_2$ overlap. 

\begin{property}\label{property:reachability2}
    {\em Given two maximum flows, $f_1$ from $\Src_1$ to $\Dst_1$ and $f_2$ from $\Src_2$ to $\Dst_2$, if $f_1$ and $f_2$ are overlap-free, \textnormal{\MFlow}$(\{\Src_1,\Src_2\}, \{\Dst_1, \Dst_2\}) = |f_1| + |f_2|$.}
\end{property}

\eat{
The \SDMF{} query, $Q=(\Src,\Dst,k)$, is impractical to compute as it requires evaluating the maximum flow for $2^{|\Src|+|\Dst|}$ different combinations of subsets from $\Src$ and $\Dst$. Query decomposition is a solution to this issue, as it decomposes the \SDMF{} problem into a set of smaller and manageable \SDMF{} sub-problems.
}

\begin{example}\label{eg:wcc}
Given a query $Q=(\Src,\Dst, 4)$ shown in Figure~\ref{fig:runningexample}, where $\Src=\{s_1,s_2,s_3,s_4\}$ and $\Dst=\{t_1, t_2, t_3, t_4, t_5\}$, there are $2^{|\Src| + |\Dst|}= 512$ combinations, \ie $512$ times maximum flow calculations. By checking the reachability, we obtain two \WCC{}s: $\mathsf{WCC}_1=(\Src_1,\Dst_1)$, where $\Src_1 = \{s_1,s_2\}$ and $\Dst_1=\{t_1,t_2,t_3\}$, and $\mathsf{WCC}_2=(\Src_2,\Dst_2)$, where $\Src_2 = \{s_3,s_4\}$ and $\Dst_1=\{t_4,t_5\}$, as shown in Figure~\ref{fig:runningexample}(a). There are $2^{|\Src_1| + |\Dst_1|}= 32$ (resp. $2^{|\Src_2| + |\Dst_2|}= 16$) combinations on $\mathsf{WCC}_1$ (resp. $\mathsf{WCC}_2$). Thus, a total of $32+16=48$ maximum flow calculations are needed, which is only $9.4\%$ of the calculations required by directly enumerating subsets of $\Src$ and $\Dst$.
\end{example}

\subsection{Computing Intermediate Densest Flows} \label{sec:comp-dense-flow}
We compute and store intermediate densest-flow values in a \jiaxin{\emph{densest-flow (DF) array}, denoted $\DF_i$}. For each $\mathsf{WCC}_i=(\Src_i,\Dst_i)$, $i\in[1,\mathsf{nw}]$,  $\DF_i[k]=\MFlow(\Src',\Dst')$ which is the densest flow value from $\Src'\subseteq \Src_i$ to $\Dst'\subseteq  \Dst_i$, such that $|\Src'|~+~|\Dst'|~=k$ and $g(\Src', \Dst')$ is maximized. The length of $\DF_i$, denoted as $len(\DF_i)$, is $|\Src_i|~+~|\Dst_i|~+~1$ since $0\leq k \leq |\Src_i|~+~|\Dst_i|$. All the arrays, $\DF_i$s, are stored in a sequence $\DFArray$.

\begin{example}
Consider $\mathsf{WCC}_1 = (\Src_1, \Dst_1)$ in Example~\ref{eg:wcc}. $\DF_1$ is initialized with a length of $5$. For $k=3$, the flow from $\Src'=\{s_2\}$ to $\Dst'=\{t_2,t_3\}$ is the densest with the flow value $\DF_1[3]=$ \textnormal{\MFlow}$(\Src', \Dst')=9$. Similarly, $\DF_1[2]=5$ with $\Src'=\{s_2\}$ and $\Dst'=\{t_2\}$, $\DF_1[4]=10$ with $\Src'=\{s_1,s_2\}$ and $\Dst'=\{t_2,t_3\}$, and $\DF_1[5]=11$ with $\Src'=\{s_1,s_2\}$ and $\Dst'=\{t_1,t_2,t_3\}$. $\DF_2$ for $\mathsf{WCC}_2$ is shown in Figure~\ref{fig:runningexample}(b).
\end{example}

\stitle{Time Complexity.} The state-of-the-art indexes \eg~\cite{sigmod2012,pathtree,cheng2013tf,su2016reachability} can be used to determine the reachability between the sources $\Src$ and the sinks $\Dst$. \WCC{}s can be found in $O(|\Src||\Dst|\log |V|)$ by using these indexes. The time complexity to compute the densest flow arrays, $\DF_i$ for $i\in [1,\mathsf{nw}]$, is $O(\sum_{i=1}^{\mathsf{nw}}2^{|\Src_i| + |\Dst_i|}M)$, where $M$ is the complexity of any maximum flow algorithms.

\subsection{Merging of Intermediate Densest Flows}
\label{sec:merge}

We next propose an efficient method for merging the densest flow of each \WCC{} to yield the global densest flow array for solving \SDMF{}. We prove in Lemma~\ref{lemma:subset} that the densest flow of multiple \WCC{}s can be constructed from the densest flows of individual \WCC{}s.

\begin{lemma}\label{lemma:subset}
    Given $\mathsf{WCC}_i~=~(\Src_i,~\Dst_i)$~and~$\mathsf{WCC}_j~=~(\Src_j, \Dst_j)$, if~$\Src_i'\subseteq \Src_i$, $\Dst_i'\subseteq \Dst_i$, $\Src_j'\subseteq \Src_j$, and $\Dst_j'\subseteq \Dst_j$, then 
    \begin{footnotesize}
    \begin{equation}
    \textnormal{\MFlow}(\Src_i'\cup \Src_j', \Dst_i'\cup \Dst_j') = \textnormal{\MFlow}(\Src_i', \Dst_i') + \textnormal{\MFlow}(\Src_j', \Dst_j').
    \end{equation}
    \end{footnotesize}
\end{lemma}

\input{algo-flow-merge}

Algorithm~\ref{algo:flowmerge} of $\name{}$ merges the densest flow arrays, $\DF$s, in a recursive manner (Line~\ref{algo:flowmerge:recursive} and Lines~\ref{algo:flowmerge:dfm:start}-\ref{algo:flowmerge:dfm:end}). If there is only one array to be merged, it is returned directly (Line~\ref{algo:flowmerge:dfm:one}). If there are two arrays to be merged (Line~\ref{algo:flowmerge:dfm:two}), $\name$ first initializes an array $\DFP$ to store the merged densest flow (Line~\ref{algo:flowmerge:init}). $\name$ compares $\DF_1[k_1] + \DF_2[k_2]$ with $\DF[k_1+k_2]$ ($k_1\in [0,len(\DF_1)]$, $k_2\in [0,len(\DF_2)]$) by iterating through $\DF_1$ and $\DF_2$. If $\DF_1[k_1] + \DF_2[k_2]$ is greater than $\DF[k_1+k_2]$, a denser flow is found (Line~\ref{algo:flowmerge:denser}). If there are more than two densest flow arrays, they are divided into two parts and merged recursively (Line~\ref{algo:flowmerge:dfm:end}). Once all $\DF$s have been merged, a global densest flow array $\DF$ is obtained and returned (Line~\ref{algo:flowmerge:recursive}).

\stitle{Computing the Answer of \SDMF{} from $\DF$.} For a \SDMF{} query, $Q=(\Src,\Dst, k)$, the maximal value of $\frac{\DF'[k']}{k'}$ is returned, where $k'\geq k$. The subsets $\Src'$ and $\Dst'$ can be obtained using an inverted map. The implementation details are omitted due to space limitations.

\begin{example}
For \WCC{}s in Example~\ref{eg:wcc}, their densest flow arrays are $\DF_1=[0,5,9,10,11]$ and $\DF_2=[0,7,8,11]$ (Figure~\ref{fig:runningexample}(b)). Consider the merged densest flow $\DF[5]$. There are $5$ pairs of ~$k_1$ and ~$k_2$ such that $k_1~+~k_2~=~5$. When $k_1=3$ and $k_2=2$, $\DF[5]~=~\DF_1[3]~+~\DF_2[2]~=~16$ is maximized. The rest elements in $\DF$ are computed similarly (Figure~\ref{fig:runningexample}(c)). Consider the query $Q=(\Src,\Dst, 4)$ in Example~\ref{eg:wcc}, $\frac{\DF'[5]}{5}$ is maximal. Therefore, the subsets $\Src'=\{s_2,s_4\}$ and $\Dst'=\{t_2,t_3,t_5\}$ have the densest flow which consists of two sub-flows: 1) from $\{s_2\}$ to $\{t_2, t_3\}$; and 2) from $\{s_4\}$ to $\{t_5\}$ (Figure~\ref{fig:runningexample}(e)). The answer to the query, $Q$, is $\Src'$ and $\Dst'$ with a flow density of $3.2$.
\end{example}

\stitle{Time Complexity.} Algorithm~\ref{algo:flowmerge} has the time complexity of $O((|\Src| + |\Dst|)^2)$. Specifically, the cost is \begin{footnotesize} $$\sum_{i=2}^{\mathsf{nw}}(len(\DF_i) \sum_{j=1}^{i}len(\DF_j)) < \sum_{i=1}^{\mathsf{nw}}(len(\DF_i)\sum_{j=1}^{\mathsf{nw}}len(\DF_j))$$ \end{footnotesize} Since $\sum_{j=1}^{\mathsf{nw}}len(\DF_j))$ $= |\Src| + |\Dst|$, the time complexity is bounded by $O((|\Src| + |\Dst|)^2)$.

%% file: algo-flow-merge.tex
\begin{algorithm}[tb]
    \caption{Densest Flow Merge}\label{algo:flowmerge}
    \SetKwProg{Fn}{Function}{}{}
    \footnotesize
    \KwIn{A sequence of densest flow arrays $\DFArray=[\DF_1,\dots,\DF_{\mathsf{nw}}]$}
    \KwOut{A merged densest flow array $\DFP$}
    
    \Return $\mathsf{DFM}(\DFArray, 1, \mathsf{nw})$ \label{algo:flowmerge:recursive}  \\
    
    \Fn{$\mathsf{DFM}$$($\textnormal{\DFArray}, $l$, $r)$}{ \label{algo:flowmerge:dfm:start}
        \If(\tcp*[h]{only one densest flow array}){$l = r$}{ 
            \Return $\DF_l$ \label{algo:flowmerge:dfm:one}\\
        }
        $m \gets \lfloor \frac{l+r}{2} \rfloor$ \\
        \Return $\textsc{ArrMrg} (\mathsf{DFM}(\DFArray, l, m), \mathsf{DFM}(\DFArray, m + 1, r))$ \eat{\tcp*[h]{more than two $\DF$s}}\label{algo:flowmerge:dfm:end}\\ 
    }
    \Fn{$\textnormal{\textsc{ArrMrg}}(\DF_1, \DF_2)$}{\label{algo:flowmerge:dfm:two}
    init an empty densest flow array $\DFP$  \label{algo:flowmerge:init} \\ 
        \ForEach(){$k\in [0, \ldots, len(\DF_1) + len(\DF_2)-1]$}{
            $\DFP[k] \gets 0$
        }
    
        \ForEach(){$k_1\in [0,\ldots, len(\DF_1)-1]$}{
            \ForEach(){$k_2\in [0,\ldots, len(\DF_2)-1]$}{
                \If(){$\DFP[k_1+k_2] < \DF_1[k_1]+\DF_2[k_2]$}{
                $\DFP[k_1+k_2] \gets \DF_1[k_1]+\DF_2[k_2]$ \label{algo:flowmerge:denser}\\  
            } 
        }
    }
    \Return $\DFP$
    }

\end{algorithm}

%% file: 6-greedy.tex
\section{3-Approximation algorithm for \SDMF}\label{sec:appr}

\jiaxin{The divide-and-conquer approach (presented in Section~\ref{sec:dc}) solves the \SDMF{} query by dividing $\Src$ and $\Dst$ into smaller subsets. However, when the subsets remain large, computing their densest-flow arrays is still expensive. Inspired by $\DKS$~\cite{feige1997densest}, in Section~\ref{sec:flowpeeling} we propose an approximate \emph{flow peeling} algorithm with a one-third approximation guarantee to substantially reduce this cost. Section~\ref{sec:pruning} further introduces a pruning technique that reduces the number of maximum-flow computations.}

\subsection{Flow Peeling Algorithm for \SDMF{}}\label{sec:flowpeeling}

We use the term $\PF(u,\Src,\Dst)$ to denote the reduction in the maximum flow value $\MFlow(\Src, \Dst)$ when vertex $u$ is removed from the set $\Src \cup \Dst$. This represents the \textit{peeling flow} of vertex~$u$.

\begin{definition}[Peeling Flow ($\PF$)] \label{def:pf}
     Given a flow network $G=(V,E,C)$, a set $\Src$ of sources and a set $\Dst$ of sinks, $\mathsf{PF}(u, \Src, \Dst)$ is:
     \begin{footnotesize}
    \begin{equation}\label{forceeq}
        \mathsf{PF}(u, \Src, \Dst) = \begin{cases}
            \textnormal{\MFlow}(\Src, \Dst) - \textnormal{\MFlow}(\Src\setminus\{u\}, \Dst), & \text{if} \; u\in \Src \\
            \textnormal{\MFlow}(\Src, \Dst) - \textnormal{\MFlow}(\Src, \Dst\setminus\{u\}), & \text{if} \; u\in \Dst  \\
            \end{cases}
\end{equation}
\end{footnotesize}
\end{definition}

\input{algo-argmax-baseline}

\stitle{Flow Peeling (Algorithm~\ref{algo:greedy}, Figure~\ref{fig:runningexample}(d)).} At the beginning of the flow peeling algorithm, $\Src_n$ and $\Dst_n$ are set to  $\Src$ and $\Dst$, respectively, and $n$ is set to $|\Src| + |\Dst|$ (Lines~\ref{algo:greedy:inite}). We use $(\Src_i,\Dst_i)$ to denote the vertex set pair after the $i$-th peeling step.  The algorithm iteratively peels a vertex $u_i$ either from $\Src_{i}$ or $\Dst_{i}$ such that $\PF(u_i)$ is minimized (Lines~\ref{algo:greedy:minizedb}-\ref{algo:greedy:minizede}). The process is repeated until all vertices have been peeled, resulting in a series of sets, denoted by $(\Src_n,\Dst_n), \ldots, (\Src_0,\Dst_0)$ of sizes $n, \ldots, 0$. Then, $\Src_i$ and $\Dst_i$ ($i\in [k, n]$), that maximize the density metric $g(\Src_i, \Dst_i)$, are returned. For simplicity, we denote the minimum peeling flow value in each step as $\delta_i = \min\{\PF(u_i,\Src_i,\Dst_i) | u_i\in \Src_i\cup \Dst_i\}$.

\begin{example}
Consider the query $Q$ in Example~\ref{eg:wcc}. Figure~\ref{fig:runningexample}(d) shows the full sequence of the flow peeling. Figure~\ref{fig:greedy} shows the details of the flow peeling process when $i=6$ and $i=5$. When $i>6$, vertices $t_1$, $s_1$ and $s_3$ have been peeled. Therefore, when $i=6$, $\Src_6=\{s_2,s_4\}$ and $\Dst_6=\{t_2,t_3,t_4,t_5\}$. Since $t_4$ has the smallest peeling flow value $\PF(t_4,\Src_6,\Dst_6) = 1$, it is peeled. This leaves $\Src_5=\Src_6=\{s_2,s_4\}$ and $\Dst_5=\Dst_6\setminus \{t_4\} = \{t_2,t_3,t_5\}$. When $i=5$, $t_2$ is peeled since it has the smallest peeling flow value. After all vertices are peeled, the returned sets ($\Src_5$, $\Dst_5$) have the densest flow among all $(\Src_i,\Dst_i)$ ($i\in [0,9]$).
\end{example}

\stitle{Analysis of Approximation Ratio.} Due to space limitations, we only present the main ideas of the proofs. Detailed proofs of all lemmas and theorem are provided in~\cite{techreport}.

To gain insights into the structural characteristics of networks and determine the essential vertices, we first present a definition of the $F\Core{}$.

\begin{definition}[$F\Core{}$] Given a set $\Src$ of sources and a set $\Dst$ of sinks, $\Src^F\subseteq \Src$ and $\Dst^F\subseteq \Dst$, is an $F\Core$ if ~$~\forall u\in \Src^F\cup \Dst^F$, $\PF(u,\Src^F,\Dst^F) \geq F$, denoted by $(\Src^F,\Dst^F)=F\Core(\Src, \Dst)$.
\end{definition}

\begin{figure}[tb]
	\begin{center}
	\includegraphics[width=0.45\textwidth]{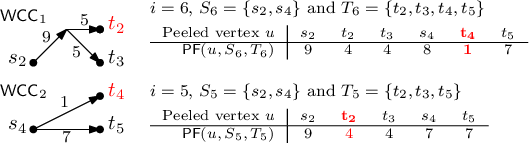}
	\end{center}
 \vspace{-3mm}
	\caption{$i=6$ and $i=5$ of Algorithm~\ref{algo:greedy}: $\Src_6=\{s_2,s_4\}$ and $\Dst_6=\{t_2,t_3,t_4,t_5\}$, and peeling flows of $t_4$ and then $t_2$}\label{fig:greedy}
\end{figure}

Intuitively, vertices in $F\Core{}(\Src,\Dst)$ have significant $\PF$s, \wrt $\Src$ and $\Dst$. We can prove that $F\Core$ exists for $0 \leq F\leq g(\Src, \Dst)$, where $g(\Src, \Dst)$ represents the density of the maximum flow from $\Src$ to $\Dst$.

\begin{lemma}\label{lemma:existence}
    \eat{$\forall F \in [0, g(\Src, \Dst)]$, $\exists i \in [0,n]$, $(\Src_i, \Dst_i) = F\Core(\Src,\Dst)$.}
    For any $F \in [0, g(\Src, \Dst)]$, there always exists $i \in [0,n]$ that $(\Src_i, \Dst_i) = F\Core(\Src,\Dst)$.
\end{lemma}

For $F \in [0,g(\Src,\Dst)]$, there may exist multiple $F\Core$s. In the following, we focus on the $F\Core$ with the highest index $i$ for $\Src_i$ and $\Dst_i$, \ie $\delta_j < F$ for $j\in (i,n]$.

\begin{lemma}\label{lemma:apprcore}
  $\forall \alpha\in [0,1]$ and $F = \alpha \cdot g(\Src, \Dst)$, $(1-\alpha)$~\textnormal{\MFlow}$(\Src, \Dst)~\le$~\textnormal{\MFlow}$(\Src^F,\Dst^F)$.
\end{lemma}


\eat{
\yxzhao{Denote the exact solution to \KSTDF{} query by $(S',T')$ and assume $\alpha=\frac{2}{3}$. Then, $F=\frac{2g(S',T')}{3}$. There always exists $i \in [0,n]$ that $(\Src_i, \Dst_i) = F\Core(\Src,\Dst)$ (Lemma~\ref{lemma:existence}). If $i\ge k$, then $g(\Src_i,\Dst_i)\ge \frac{g(S',T')}{3}$ and hence is a 3-approximation of the answer. If $i<k$, then $\frac{\MFlow(S_i,T_i)}{3}\revise{\ge}\MFlow(S,T)$ (Lemma~\ref{lemma:apprcore}). Therefore, we have the following theorem.}
}

Denote the exact answer to \SDMF{} query by $(S',T')$, consider $\alpha=\frac{2}{3}$ and $F= \alpha \cdot g(\Src, \Dst)$. Lemma~\ref{lemma:existence} guarantees that there always exists $(S_i,T_i)$ that is a $F$Core, and Lemma~\ref{lemma:apprcore} guarantees that either $(S_i,T_i)$ or $(S_k,T_k)$ is a 3-approximation of the exact answer. Therefore, we have the following theorem.
\begin{theorem}\label{theorem:appr}
    Algorithm~\ref{algo:greedy} is a $3$-approximation for \SDMF{}.
\end{theorem}

\stitle{Time Complexity.} Algorithm~\ref{algo:greedy} requires $|\Src|+|\Dst|$ times the cost of the maximum flow calculation to obtain the smallest peeling flow for each $i$. Since there are at most $|\Src|+|\Dst|$ peelings, the time complexity is bounded by $O((|\Src|+|\Dst|)^2M)$, where $M$ represents the complexity of any maximum flow algorithm.

\noindent \jiaxin{\stitle{Remarks on the relation to $\DKS{}$.} Although our flow-peeling procedure is conceptually inspired by the intuition behind $\DKS{}$—removing low-contribution vertices iteratively—the process of \SDMF{} is fundamentally different.  In $\DKS{}$, the marginal gain of a vertex is a \emph{local}, additive quantity derived from edge weights, which enables classical degree-based peeling.  In contrast, the marginal gain in \SDMF{} depends on a \emph{global} multi-source–multi-sink \emph{temporal} maximum flow: removing one vertex may reroute, diminish, or eliminate entire flow paths, and its contribution cannot be inferred from local connectivity. Moreover, $\DKS{}$ is a special case of \SDMF{} obtained when all timestamps equal $1$ and the flow value $\MFlow(S,T)$ degenerates into the total edge weight between $S$ and $T$. Because temporal flow interactions are non-decomposable, classical $\DKS{}$ peeling rules cannot be directly applied; thus, we introduce a new flow-based marginal definition and a flow-aware $F\Core$ ordering for \SDMF{}.}

\subsection{Pruning in the Flow Peeling Algorithm}\label{sec:pruning}

\input{algo-argmax-pruning.tex} 

\jiaxin{Algorithm~\ref{algo:greedy} enumerates all vertices $s \in \Src_i$ (resp.\ $t \in \Dst_i$) and computes $\MFlow(\Src_i\setminus\{s\},\Dst_i)$ (resp.\ $\MFlow(\Src_i,\Dst_i\setminus\{t\})$) in order to find the vertex $u_i$ that minimizes the peeling flow $\PF(u,\Src_{i-1},\Dst_{i-1})$ (Line~\ref{algo:greedy:minizede}). However, this full enumeration is computationally expensive because each  candidate requires a new maximum-flow evaluation.}

\jiaxin{To make this step more efficient, we exploit that peeling-flow values at later iterations are \emph{algebraically constrained} by those computed on larger source/sink sets: once $\PF(s,\Src,\Dst)$ and $\PF(t,\Src,\Dst)$ are known, the peeling flows after removing $s$ or $t$ cannot become arbitrarily small. This enables us to maintain, for every vertex $u\in \Src\cup\Dst$, a  \emph{lower bound} $\LPF(u)$ on its future peeling-flow value without performing any max-flow computation. Intuitively, if the smallest already-computed peeling flow is strictly smaller than the minimum possible $\LPF$ among all remaining vertices, then no unexamined vertex can become the best peeling candidate, and the enumeration can be safely terminated. Based on this idea, we derive a set of properties that yield valid lower bounds for pruning unnecessary candidates.}

\begin{property}\label{property:peelnotneg}
    $\forall u\in \Src\cup \Dst$, $\mathsf{PF}(u,\Src,\Dst)\ge 0$.
\end{property}


\begin{property}\label{property:otherside}
$\forall s\in \Src, t\in \Dst$, 
\begin{align}
    1)\ & \mathsf{PF}(t, \Src\setminus\{s\}, \Dst) \ge \mathsf{PF}(t, \Src, \Dst) - \mathsf{PF}(s, \Src, \Dst); \text{ and} \\
    2)\ & \mathsf{PF}(s, \Src, \Dst\setminus\{t\}) \ge \mathsf{PF}(s, \Src, \Dst) - \mathsf{PF}(t, \Src, \Dst).
\end{align}
\end{property}

\tr{
\begin{proof}
With the definition of peeling flow, we have
\begin{equation}\label{eq:tst-fp}
\begin{split}
\PF(t,\Src,\Dst) & =  \MFlow(\Src,\Dst)-\MFlow(\Src,\Dst\setminus\{t\})\\
\end{split}
\end{equation}
\begin{equation}\label{eq:sst-fp}
\begin{split}
\PF(s,\Src,\Dst) & =  \MFlow(\Src,\Dst)-\MFlow(\Src\setminus\{s\},\Dst)\\
\end{split}
\end{equation}
Combining Equation~\ref{eq:tst-fp} and Equation~\ref{eq:sst-fp}, we have the following.
\begin{equation}\label{eq:combine-fp}
    \PF(t,\Src,\Dst) - \PF(s,\Src,\Dst) =\MFlow(\Src\setminus\{s\},\Dst) -  \MFlow(\Src,\Dst\setminus\{t\})
\end{equation}
Due to property~\ref{property:peelnotneg}, we have
\begin{equation}
    \PF(s,\Src,\Dst\setminus\{t\}) = \MFlow(\Src,\Dst\setminus\{t\}) - \MFlow(\Src\setminus\{s\},\Dst\setminus\{t\}) \ge 0
\end{equation}
Therefore, $\MFlow(\Src,\Dst\setminus\{t\}) \geq \MFlow(\Src\setminus\{s\},\Dst\setminus\{t\})$. Moreover, with Equation~\ref{eq:combine-fp}, we have the following.
\begin{equation}
\begin{split}
\PF(t,\Src\setminus\{s\},\Dst) & =  \MFlow(\Src\setminus\{s\},\Dst)-\MFlow(\Src\setminus\{s\},\Dst\setminus\{t\})\\
& \geq \MFlow(\Src\setminus\{s\},\Dst)-\MFlow(\Src,\Dst\setminus\{t\})\\
& = \PF(t,\Src,\Dst) - \PF(s,\Src,\Dst)
\end{split}
\end{equation}
Similarly, we have $\mathsf{PF}(s,\Src,\Dst\setminus\{t\})\ge \mathsf{PF}(s,\Src,\Dst) - \mathsf{PF}(t,\Src,\Dst)$.
\end{proof}
}

\stitle{Flow Peeling Using a Lower Bound (Algorithm~\ref{algo:greedyprune}).} By using Properties~\ref{property:peelnotneg}-\ref{property:otherside}, we maintain a \textit{lower bound of the peeling flow} for each vertex in $\Src \cup \Dst$, denoted by $\LPF(u)$ (Line~\ref{algo:greedyprune:initLPF}). To determine $u_i$ that minimizes $\PF(u,\Src_{i-1},\Dst_{i-1})$, we calculate $\PF(u)$, where $u\in \Src_i\cup \Dst_i$, in ascending order of $\LPF(u)$. \textit{If the current minimum $\PF(u,\Src_i,\Dst_i)$ is smaller than the minimum $\LPF$ of all vertices that have not been enumerated, the enumeration is terminated (Line~\ref{algo:greedyprune:break}), and $u$ is chosen for the peeling.} $\LPF(u)$ is refined if the following conditions are satisfied.
\begin{enumerate}[leftmargin=*]
    \item When $\PF(u,\Src_i,\Dst_i)$ is calculated, $\LPF(u)$ is refined by $\LPF(u)=\PF(u,\Src_i,\Dst_i)$ (Line~\ref{algo:greedyprune:LPF1}).
    \item When a vertex $u$ is peeled (Lines~\ref{algo:greedyprune:peelstart}-\ref{algo:greedyprune:peelend}), if $u\in \Src_i$, $\LPF(t)$ is refined by $\LPF(t)=\LPF(t)-\PF(u,\Src_i,\Dst_i)$, where $t \in \Dst_i$; otherwise, $\LPF(s)$ is refined by $\LPF(s)=\LPF(s)-\PF(u,\Src_i,\Dst_i)$, where $t \in \Src_i$.
\end{enumerate}

\stitle{Time Complexity.} The time complexity of Algorithm~\ref{algo:greedyprune} is $O((|\Src| + |\Dst|)^2M)$, where $M$ is the complexity of any maximum flow algorithm. Sorting $\LPF$s adds an extra $O((|\Src|+|\Dst|)^2\cdot\log(|\Src|+|\Dst|))$ cost. Since $|\Src|~\ll~|V|$ and $|\Dst|~\ll~|V|$, the sorting is efficient in practice. The pruning technique reduces the elapsed time of the peeling algorithm by avoiding enumerations and provides an approximation guarantee.

\subsection{Divide-and-conquer Peeling Algorithm}\label{subsec:dcpeel}

\begin{figure}[tb]
	\begin{center}
	\includegraphics[width=0.45\textwidth]{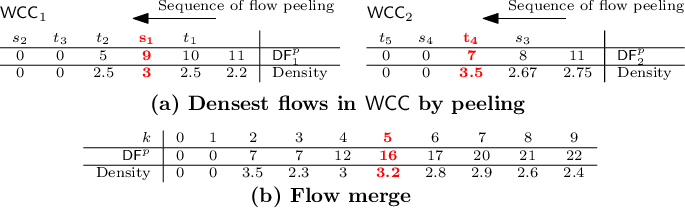}
	\end{center}
 \vspace{-3mm}
	\caption{Examples of the divide-and-conquer peeling algorithm}\label{fig:peelingdc}
\end{figure}

The flow peeling algorithm in Sections~\ref{sec:flowpeeling}-\ref{sec:pruning} is orthogonal to the divide-and-conquer approach (Section~\ref{sec:dc}). We propose to integrate them to obtain a divide-and-conquer peeling algorithm as follows.

\begin{enumerate}[leftmargin=*]
\itemsep0em 
    \item \stitle{Query Decomposition.} $\name{}$ decomposes $\Src$ and $\Dst$ into a set of \WCC{}s as introduced in Section~\ref{sec:dec}.
    \item \stitle{Flow Peeling.} $\name{}$ performs flow peeling on each $\mathsf{WCC}_i$ to produce the densest flow array, denoted by $\DF_i^p$.
    \item \stitle{Densest Flow Merge.} $\name{}$ merges $\DF_i^p$ to obtain the global answer of \SDMF{} by using Algorithm~\ref{algo:flowmerge}.
\end{enumerate}

\begin{example}
    In Figure~\ref{fig:runningexample}, the sources and sinks are decomposed into two \WCC{}s as depicted in Figure~\ref{fig:runningexample}(a). Upon applying the peeling algorithms on each of these $\mathsf{WCC}$s, two densest flow arrays ($\DF_1^{p}$ and $\DF_2^{p}$) are obtained as shown in Figure~\ref{fig:peelingdc}(a). Finally, these two densest arrays are merged to produce the global densest flow array $\DF^p$ (Figure~\ref{fig:peelingdc}(b)).
\end{example}

\stitle{Time Complexity.} The overall time complexity is \\ $O(\sum_{i=1}^{\mathsf{nw}}(|\Src_i| + |\Dst_i|)^2M)$\eat{ which is }, bounded by $O((|\Src| + |\Dst|)^2M)$, where $M$ represents the complexity of any maximum flow algorithm.

\stitle{Space Complexity.} The sizes of both the original graph and the transformed graph are bounded by $O(|V|+|E|)$. \eat{The maximum size of the flow array is limited by $O(|\Src| + |\Dst|)$.} The size of the flow array is bounded by $O(|\Src| + |\Dst|)$.

%% file: algo-argmax-baseline.tex
\begin{algorithm}[tb]
    \caption{Flow Peeling Algorithm for \SDMF{}}\label{algo:greedy}
    \SetKwProg{Fn}{Function}{}{}
    \footnotesize
    \KwIn{$G=(V, E, C)$, $Q=(\Src, \Dst, k)$}
    \KwOut{$Q(G)$}
    $n \gets |\Src|~+~|\Dst|$, $\Src_n \gets \Src$, $\Dst_n \gets \Dst$ \label{algo:greedy:inite} \\
    \ForEach(){$i \in [n-1,\ldots, 0]$}{
        $\Src_i\gets \emptyset, \Dst_i\gets \emptyset$
    }
    \ForEach(){$i \in [n,\ldots, 1]$}{
        $\delta_i \gets +\infty$ \label{algo:greedy:minizedb} \tcp*[h]{init the minimum peeling flow value}\\ 
        \ForEach(\tcp*[h]{search for $u$ minimizing $\PF(u,\Src_i,\Dst_i)$}){$u\in \Src_i\cup \Dst_i$}{
            \If{$\PF(u, \Src_i, \Dst_i) < \delta_i$}{ \label{algo:greedy:callpf}
                $\delta_i \gets \PF(u, \Src_i, \Dst_i)$ \\
                $u_i \gets u$ \label{algo:greedy:minizede}\\
            }
        }
        \eIf(\tcp*[h]{peel $u_i$ from $\Src_i$}){$u_i\in \Src_i$}{ $\Src_{i-1}\gets \Src_i\setminus\{u_i\}$, $\Dst_{i-1}\gets\Dst_i$\\
               }(\tcp*[h]{peel $u_i$ from $\Dst_i$}){ $\Src_{i-1}\gets\Src_i$, $\Dst_{i-1}\gets\Dst_i\setminus\{u_i\}$\\
               }
    }
    \Return $\Src_i$ and $\Dst_i$ maximizing density $g(\Src_i, \Dst_i)$, where $i\geq k$
\end{algorithm}

%% file: algo-argmax-pruning.tex
\begin{algorithm}[tb]
    \caption{Peeling Algorithms with Pruning}\label{algo:greedyprune}
    \SetKwProg{Fn}{Function}{}{}
    \footnotesize
    \KwIn{$G=(V, E)$, $Q=(\Src, \Dst, k)$}
    \KwOut{$Q(G)$}
    $n \gets |\Src| + |\Dst|$, $\Src_n \gets \Src$, $\Dst_n \gets \Dst$ \\
    \ForEach(){$i \in [n-1,\ldots, 0]$}{
        $\Src_i\gets \emptyset, \Dst_i\gets \emptyset$
    }
    \ForEach(){$u \in \Src\cup \Dst$}{
        $\LPF(u) = 0$ \label{algo:greedyprune:initLPF} \tcp*[h]{Property \ref{property:peelnotneg}}
    }
    \ForEach(){$i \in [n,\ldots, 1]$}{
        $\delta_i \gets +\infty$ \\
        Sort $u\in \Src_i \bigcup \Dst_i$ in the ascending order of $\LPF(u)$ \\
        \ForEach(){$u\in \Src_i \bigcup \Dst_i$}{
            \If(\tcp*[h]{early pruning}){$\delta_i\le \LPF(u)$}{
                break \label{algo:greedyprune:break}
            }
            \If{$\PF(u,\Src_i,\Dst_i)<\delta_i$}{  \label{algo:greedyprune:callpf}
                $\LPF(u)\gets\PF(u,\Src_i,\Dst_i)$ \label{algo:greedyprune:LPF1}\\
                $\delta_i \gets \PF(u,\Src_i,\Dst_i)$ \\
                $u_i \gets u$ \\
            }
        }
        \If{$u_i\in \Src_i$}{ \label{algo:greedyprune:peelstart}
            \ForEach(\tcp*[h]{Property~\ref{property:otherside}}){$t\in \Dst_i$}{$\LPF(t)\gets\LPF(t)-\PF(u,\Src_i,\Dst_i)$}
            $\Src_{i-1}\gets\Src_i\setminus\{u_i\}$
        }
        \Else{
            \ForEach(\tcp*[h]{Property~\ref{property:otherside}}){$s\in \Src_i$}{$\LPF(s)\gets\LPF(s)-\PF(u,\Src_i,\Dst_i)$}
            $\Dst_{i-1}\gets\Dst_i\setminus\{u_i\}$ \label{algo:greedyprune:peelend} 
        }
        $\delta_i \gets \MFlow(\Src_i, \Dst_i) - \MFlow(\Src_{i-1}, \Dst_{i-1})$ \\
    }
    
    \Return $\Src_i$ and $\Dst_i$ maximizing density $g(\Src_i, \Dst_i)$, where $i\geq k$ \\
\end{algorithm}

%% file: 7-experiments.tex
\section{EXPERIMENTAL STUDY}\label{sec:exp}

\eat{
We begin by outlining the experimental settings in Section~\ref{Exp-Settings}. We first evaluate the performance of network transformation (Stage 1 of $\name$) Section~\ref{Exp-Effectiveness-Transformation}. We then investigate into both the efficiency and the effectiveness of $\name$ (Stage 2) in Section~\ref{Exp-Efficiency} and Section~\ref{Exp-Effectiveness}, respectively. Lastly, we provide case studies in Section~\ref{Exp-CaseStudy} to further demonstrate its capabilities to investigate anomalies in real-world transaction flow networks.
}

\subsection{Experimental Setup}\label{Exp-Settings}

\stitle{Software and Hardware.} Our experiments are conducted on a machine with a Xeon Gold 6330 CPU, and $64$GB memory. The algorithms are implemented in C++ and the implementation is made memory-resident. All codes are compiled by GCC-8.5.0 with -$O3$. For computing the maximum flow, we implemented the algorithm from~\cite{dinic1970algorithm} as a built-in component. It should be remarked that any maximum flow algorithm could be employed to replace the one from~\cite{dinic1970algorithm}.

\input{exp-datasets}

\begin{table*}[!tb]
\caption{\jiaxin{\FCT{} vs. \BTF{} vs. \BRTF{} on the whole network ("Trans." represents the time taken for network transformation, and "\% of Trans." indicates the percentage of network transformation time relative to the total query time of \BRTF{}).}} \label{table:preprocessing-effect}
\centering
\resizebox{\linewidth}{!}{%
\begin{footnotesize}
\jiaxin{
\begin{tabular}{|c||c|c|c|c|c|c|c|c|c|c|c|c|}
  \hline
  \multirow{3}{*}{\textbf{Datasets}} 
  & \multicolumn{5}{c|}{\textbf{Query Efficiency without / with Transformation}} 
  & \multicolumn{3}{c|}{\textbf{Performance of Max. Flow Value}} 
  & \multirow{3}{*}{\textbf{Distance}}
  & \multicolumn{3}{c|}{\textbf{Peak Memory (MB)}} \\
  \cline{2-9}\cline{11-13}
   & \multirow{2}{*}{\FCT{} (ms)} 
   & \multirow{2}{*}{\BTF{} (ms)} 
   & \multicolumn{3}{c|}{\BRTF{}} 
   & \multirow{2}{*}{\FCT{}} 
   & \multirow{2}{*}{\textbf{\BTF{}}} 
   & \multirow{2}{*}{\textbf{\BRTF{}}}
   & & \multirow{2}{*}{\FCT{}} 
   & \multirow{2}{*}{\BTF{}} 
   & \multirow{2}{*}{\BRTF{}} \\
  \cline{4-6}
   & & & Trans. (ms) & Total Query Time (ms) & \% of Trans. & & & & & & & \\
  \hline
  Btc2011 & 6 & 0.9 & 1.3 & 174 & 0.74 & 0.0017 & 0.060 & \textbf{2.99} & 165 & 271 & 89 & 89 \\
  Btc2012 & 6,134 & 13 & 28 & 4,472 & 0.63 & 0.0002 & 0.021 & \textbf{1.59} & 43 & 1,259 & 431 & 432 \\
  Btc2013 & 7,952 & 47 & 103 & 16,780 & 0.62 & 0.0005 & 0.023 & \textbf{0.70} & 23 & 2,935 & 926 & 948 \\
  Eth2016 & 6 & 4 & 15 & 1,231 & 1.24 & 7.34 & 65.659 & \textbf{214.66} & 5,074 & 316 & 237 & 240 \\
  Eth2021 & 970 & 941 & 941 & 55,371 & 1.70 & 0.04 & 0.422 & \textbf{1.46} & 108 & 19,359 & 6,686 & 6,695 \\
  IBM & 52,680 & 12 & 39 & 401 & 8.83 & 14K & 98K & \textbf{98K} & 5 & 3,201 & 3,046 & 3,047 \\
  \hline
\end{tabular}%
}
\end{footnotesize}
}
\end{table*}

\stitle{Datasets.} Our experiments are conducted on six datasets, including five real-world transaction-flow datasets and one large-scale synthetic benchmark.
\begin{enumerate}[leftmargin=*]
    \item Three real-world datasets are extracted from the Bitcoin transaction network in $2011$, $2012$, and $2013$~\cite{btc2013}. In addition, we collect transactions from the Ethereum network in $2016$ and $2021$~\cite{wood2014ethereum}. We extract datasets by year because both the graph structure and transaction amount distributions vary significantly across time. In the early years, both BTC and ETH exhibit larger transaction values and denser activity.
    \item \jiaxin{To complement real-world data with a controllable benchmark, we additionally include the IBM synthetic transaction dataset released by the Watson Research Lab~\cite{altman2023realistic}. This dataset is designed to simulate realistic financial transaction flows at scale.} 
\end{enumerate}

\stitle{Timespan.} We provide an interface that allows investigators to specify a detection timespan $\Delta=[\tau_s,\tau_e]$. \jiaxin{For the real-world Bitcoin and Ethereum datasets, we evaluate $\name$ using three timespan lengths—one day, one week, and one month—denoted by $\Delta_d$, $\Delta_w$, and $\Delta_m$, respectively, which reflect common investigation practices on recent transactions. To generate queries, we randomly select 20 timespans from each type of time window for evaluation.} For the IBM synthetic transaction dataset, we evaluate $\name$ on the full graph, since its timestamps are synthetically generated and do not correspond to operational investigation windows; this setting is sufficient for assessing scalability.


\stitle{Queries.} For each selected timespan, we randomly generate $50$ pairs of sets $\Src$ and $\Dst$, as queries,~which satisfy the following conditions: 1) $\forall s_i \in \Src, \dego(s_i) \geq 1$ and $\forall t_i\in \Dst, \degi(t_i)\geq 1$, and 2) $|\Src|=|\Dst|=\frac{n}{2}$, where $n$ is the parameter called the size of~the~queries. \jiaxin{For the IBM synthetic transaction dataset, queries are generated on the entire graph, following the same query-generation procedure.} Table~\ref{table:Statistics} summarizes some characteristics of the datasets and queries. When investigating efficiency, we may discuss the term maximum flow as it helps to analyze the runtime, and it is different from the densest flow by a factor of $|S'|+|T'|$.


\stitle{Parameters.} For the query $Q = (\Src, \Dst, k)$ of \TEMSDMF{}, the number of sources and sinks, denoted as $n=|\Src| + |\Dst|$, affects scalability, while $k$ influences the detection of the flow density. To demonstrate the efficiency with respect to the query size, we vary $n$ to be 16, 32, and 64. To illustrate how $k$ affects the flow density, we vary $k$ among 6, 8, and 10. 
The default values for $k$ and $n$ are set to 6 and 32, respectively.

\stitle{Algorithms.} We investigate two major performance factors of $\name$, each investigation includes the respective related competitors. We first compare the efficiency and effectiveness of our graph transformation technique (Section~\ref{sec:preprocess}). The competitors are listed as follows:
\begin{enumerate}[wide, labelwidth=!, labelindent=0pt]
\itemsep0em 
\item \etitle{\FCT{}.} The network transformation of ~\cite{kosyfaki2021flow} represents the most recent work closely related to temporal flow detection. For a fair comparison, we adopt their default settings, including all optimization features. \eat{This configuration in our experiments is referred to as \FCT{}.}
\item \etitle{\BTF{} $\&$ \BRTF{}.} We implement the Dinic's algorithm in~\cite{dinic1970algorithm} on the \TFNet{} and the \RTFNet{}, respectively.
\end{enumerate}

We use $1{,}000$ random queries with a single source $s$ and a single sink $t$. We set a time threshold of $1{,}000$ seconds. We denote the queries that exceed this threshold as "Did Not Finish" ("DNF"). To investigate into the efficiency of the query evaluation of $\name$, we compare four variants as follows:
\begin{enumerate}[wide, labelwidth=!, labelindent=0pt]
	\itemsep0em 
	\item \etitle{\BWCC.} This baseline constitutes the dense flow detection pipeline at \Grab{} by enumerating all combinations of the subsets of $\Src$ and $\Dst$. We enhance this enumeration with the query decomposition and the densest flow merge techniques. Note that \BWCC{} returns the exact answers for \SDMF{} queries.
	\item \etitle{\Greedy.} Our flow peeling algorithm (Algorithm~\ref{algo:greedy}).
    \item \etitle{\Greedy*.} \Greedy{} with the pruning (Algorithm \ref{algo:greedyprune}).
	\item \etitle{\GDYWCC.} We combine \Greedy{} and \BWCC{} (Section~\ref{subsec:dcpeel}).
    \item \etitle{\GDYWCC*.} \GDYWCC{} with the pruning (Section~\ref{sec:pruning}).
\end{enumerate}



\subsection{Effectiveness of Network Transformation (Stage 1)}\label{Exp-Effectiveness-Transformation}

\tr{by analyzing the network processing time and overall efficiency in the default setting. We then examine the impact of varying the parameter $n$ on the runtime of the compared algorithms.}

\tr{
\begin{table}[tb]
\caption{\FCT{} vs. \BTF{} vs. \BRTF{} \wrt{} $\Delta_m$ ("Trans." represents the time taken for network transformation, and "\% of Trans." indicates the percentage of network transformation time relative to the total \BRTF{} time).}\label{table:preprocessing-effect}
\vspace{-0.5em}
\centering
\resizebox{\columnwidth}{!}{%
\begin{tabular}{|c|c|c|c|c|c|c|c|c|c|}
  \hline
  \multirow{2}{*}{\textbf{Datasets}} & \multicolumn{6}{c|}{\textbf{Elapsed Time (ms)}} & \multicolumn{3}{c|}{\textbf{Maximum Flow}} \\ 
  \cline{2-10}
   & \multicolumn{2}{c|}{\FCT{}~\cite{kosyfaki2021flow}} & \multirow{2}{*}{\BTF{}} & \multicolumn{3}{c|}{\BRTF{}} & \multirow{2}{*}{\FCT{}} & \multirow{2}{*}{\textbf{\BTF{}}} & \multirow{2}{*}{\textbf{\BRTF{}}} \\ 
  \cline{2-3} \cline{5-7}
   & $\mathsf{Path}$ & $\mathsf{RunTime}$ & & Trans. & Total & \% of Trans. &  &  &  \\ 
  \hline
  Btc2011 & 5.58 & 0.541 & 4 & 0.8 & 942 & 0.085 & 0.0017 & 0.082 & \textbf{3.331} \\
  Btc2012 & 6134 & 0.14  & 48 & 36.9 & 8383 & 0.440 &  0.0002 & 0.021 & \textbf{1.167} \\
  Btc2013 & 7952 & 0.32 & 201 & 161.4 & 32391 & 0.498 &  0.0005 & 0.026 & \textbf{0.562} \\
  Eth2016 & 5.59 & 0.33 & 15 & 9.1 & 2750 & 0.331 & 7.34 & 39.030 & \textbf{103.073} \\
  Eth2021 & 970 & 0.21 & 1288 & 577.9 & 77164 & 0.749 & 0.04 & 0.279 & \textbf{0.475} \\
  \hline
\end{tabular}%
}
\end{table}
}

\tr{
\begin{table}[tb]
\caption{Network processing time for queries \wrt $\Delta_m$ }\label{table:preprocessing}
\centering
\begin{scriptsize}
\begin{tabular}{|c|m{1.5cm}<{\centering}|m{1.9cm}<{\centering}|m{1.7cm}<{\centering}|}
  \hline
  {\bf Datasets} & {Network reduction (ms)}  & {Network transformation (ms)} & {Network compression (ms)}  \\ 
  \hline
  Btc2011 & 22.9 & 0.8 & 0.5 \\ \hline
  Btc2012 & 301.3 & 36.9 & 25.7 \\ \hline
  Btc2013 & 776.6 & 161.4 & 111.1 \\ \hline
  Eth2016 & 132.0 & 9.1 & 5.4 \\ \hline
  Eth2021 & 4685.0 & 577.9 & 323.1 \\ \hline
\end{tabular}
\end{scriptsize}
\end{table}
}

\tr{
\stitle{Impact of network processing.} For a \TFNet{}, we obtain its corresponding \RTFNet{} by using the network reduction (Section~\ref{sec:reduction}), the network transformation (Section~\ref{sec:transform}), and the network compression (Section~\ref{sec:compression}). For the set of $1{,}000$ random queries of $\Delta_m$, we present the average runtimes for each step across our tested datasets in Table~\ref{table:preprocessing}. Notably, the overall processing time for \RTFNet{} per query remains under $1$ second for all but the Eth2021 dataset. Further, Table~\ref{table:preprocessing-effect} details the average runtimes of \BTF{} and \BRTF{} for $1{,}000$ random queries of $\Delta_m$, each with a single source $s$ and a single sink $t$. As Tables~\ref{table:preprocessing-effect} and \ref{table:preprocessing} illustrate, the network processing time for these queries constitutes less than $7.5\%$ of the total query duration. Additionally, our experiments corroborate that the time required for network transformation is \textit{linear} to the network size. In Table~\ref{table:preprocessing-effect}, we present the results of the maximum flow values computed by~\cite{dinic1970algorithm} on \TFNet{} and \RTFNet{}, respectively. On Btc2011 (resp. Btc2012, Btc2013, Eth2016 and Eth2021), the maximum flow value of \BTF{} is only $2.47\%$ (resp. $1.80\%$, $4.63\%$, $37.87\%$ and $58.74\%$) as much as that of \BRTF{}. This \yxzhao{also} demonstrates the effectiveness of our network processing techniques in reducing the error of the maximum flow value\eat{, enabling more precise detection}.
}

\input{exp-query-performance.tex}

\eat{
\stitle{Impact of Network Transformation\eat{\footnote{$\name{}$ incorporates optional techniques for reducing the size of temporal networks, such as network reduction and network compression. Further details on these optimizations are elaborated in the appendix of our technical report, specifically in Section~\ref{sec:ntxprocess}~\cite{techreport}. Readers interested in the underlying principles and broader context of temporal networks are encouraged to consult this section for more information.}}.} 
}

\stitle{\jiaxin{Impact on Query Time.}} For a \TFNet{}, we obtain its corresponding \RTFNet{} by using the network transformation (Section~\ref{sec:preprocess}). Notably, the network transformation time for \RTFNet{} remains under $1$ second for all datasets. Further, Table~\ref{table:preprocessing-effect} reports the average runtimes of \BTF{} and \BRTF{} for $1{,}000$ random queries on respective datasets, each with a single source $s$ and a single sink $t$. As Table~\ref{table:preprocessing-effect} illustrates, the network transformation time for these queries constitutes less than $1.7\%$ of the total query duration. Additionally, our experiments corroborate that the time required for network transformation is \textit{linear} to the network size.

\stitle{\jiaxin{Impact on Maximum Flow Value.}} In Table~\ref{table:preprocessing-effect}, we also present the results of the maximum flow values computed by~\cite{dinic1970algorithm} on \TFNet{} and \RTFNet{}, respectively. On Btc2011 (resp. Btc2012, Btc2013, Eth2016 and Eth2021), the maximum flow value of \BTF{} is only $2.0\%$ (resp. $1.3\%$ $3.3\%$ $30.6\%$ and $28.8\%$) as much as that of \BRTF{}. This verified that our network processing techniques successfully addressed the problems caused by temporal flow constraint in \TFNet{} and enabled determining much larger flow values. \jiaxin{The implementation of \FCT{}~\cite{kosyfaki2021flow} consists of two primary steps: path enumeration and flow computation. Compared with both \BTF{} and \BRTF{}, \FCT{} consistently yields even smaller flow values. This is primarily due to its hop-bounded path enumeration strategy: \FCT{} enumerates only temporal paths up to a fixed length, which is often shorter than the actual flow distances observed in our queries (with average flow distance exceeding $20$ on both BTC and ETH datasets).}

\tr{
\input{exp-sizes}
}

\tr{
\stitle{Effectiveness of Network Reduction and Network Compression.} The result of network reduction and network compression mentioned in Section~\ref{sec:optimize} is shown in Table~\ref{table:Sizes}. On Btc2011 (resp. Btc2012, Btc2013, Eth2016, and Eth2021), network reduction reduces $94.7\%$ (resp. $87.9\%$, $83.1\%$, $93.3\%$, and $92.4\%$) vertices and $94.9\%$ (resp. $88.4\%$, $83.6\%$, $93.5\%$, and $92.8\%$) edges in the original \RTFNet{}. Furthermore, network compression reduces $77.5\%$ (resp. $77.5\%$, $78.8\%$, $81.3\%$, and $79.1\%$) vertices and $66.7\%$ (resp. $66.5\%$, $66.4\%$, $56.6\%$, and $58.9\%$) edges on the result of network reduction. The reduced network size results in a $72.7\times$ (resp. $15.1\times$, $14.7\times$, and $46.3\times$ and infinite) speed up in maximum flow queries. On Eth2021, maximum flow queries cannot be processed on original \RTFNet{} within the memory limit and cannot be compared. We can conclude that our network reduction and network compression techniques significantly reduce the graph size, thereby accelerating the computation of maximum flow queries.
}

\stitle{\FCT{} vs. \BTF{} vs. \BRTF{}.}  \jiaxin{During our evaluations, we found that increasing the default path length significantly hindered \FCT{}'s ability to complete path enumeration on most of our datasets. Consequently, we adhered to \FCT{}'s default settings with all optimizations. Next, we highlight the results we obtained: \BTF{} detected up to 105 times and at least $7$ times more flow than \FCT{}. Moreover, \BRTF{} identified up to 7{,}945 times more flow than \FCT{}. These discrepancies are caused by \FCT{}'s path enumeration strategy, which only considers paths within a certain length, potentially leading to the loss of flow. Notably, \BRTF{} consistently returned significantly larger temporal flows. In terms of efficiency, \BTF{} always outperformed \FCT{} on all datasets, particularly on BTC2012 where it required as little as 0.21\% of \FCT{}'s time. \emph{In addition, \BTF{} and \BRTF{} also exhibit substantially lower peak memory usage than \FCT{}, since they avoid explicit path materialization and operate on compact transformed networks (Table~\ref{table:preprocessing-effect}).}}

\subsection{Efficiency of Query Evaluation (Stage 2)}\label{Exp-Efficiency}

\stitle{Overall Efficiency (Figure~\ref{fig:overall-performance}).} We show the efficiency of the benchmarked techniques in various datasets using the default settings. It was found that \GDYWCC*{} was the most efficient in all datasets. Specifically, a) on average, \GDYWCC*{} was $4.9$ (resp. $6.5$, $4.1$, $3.9$ and $3.8$) times more efficient than \GDYWCC{} in Btc2011 (resp. Btc2012, Btc2013, Eth2016, and Eth2021) due to the pruning technique, which avoids an enumeration during each peeling step. Additionally, b) when compared to \Greedy*{}, \GDYWCC*{} had a speedup of $33$ (resp. $14$, $15$, $37$ and $33$) times on Btc2011 (resp. Btc2012, Btc2013, Eth2016, and Eth2021) on average. The divide-and-conquer approach reduces the number of peeling iterations due to smaller query sets. Furthermore, c) \GDYWCC*{} had a speedup of $253$ (resp. $97$, $110$, $313$ and $270$) times when compared to \Greedy{} on Btc2011 (resp. Btc2012, Btc2013, Eth2016 and Eth2021) on average. This is because \Greedy{} requires $n^2$ ($n=32$ in our default setting) times the maximum flow computation, which is still time-consuming. Lastly, d) it is not surprising that \GDYWCC*{} was $2{,}745$ (resp. $3{,}087$, $2{,}205$, $3{,}111$ and $349$) times more efficient than \BWCC{} in Btc2011 (resp. Btc2012, Btc2013, Eth2016, and Eth2021) on average since \BWCC{} enumerates each combination and computes the corresponding maximum flow as discussed in Section~\ref{sec:appr}. The improvement is more obvious in smaller datasets since $\Src$ and $\Dst$ has a higher chance of dividing into smaller subsets that the divide-and-conquer approach is more efficient.

\begin{figure}[tb]
    \begin{minipage}{\linewidth}
        \center
        \includegraphics[width=1\textwidth]{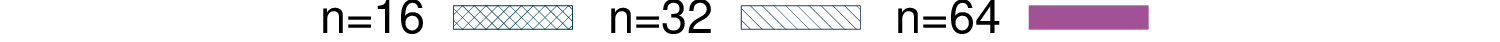}
    \end{minipage}
	\centering    
	\begin{subfigure}[b]{0.22\textwidth}
		\centering
            \includegraphics[width=\textwidth]{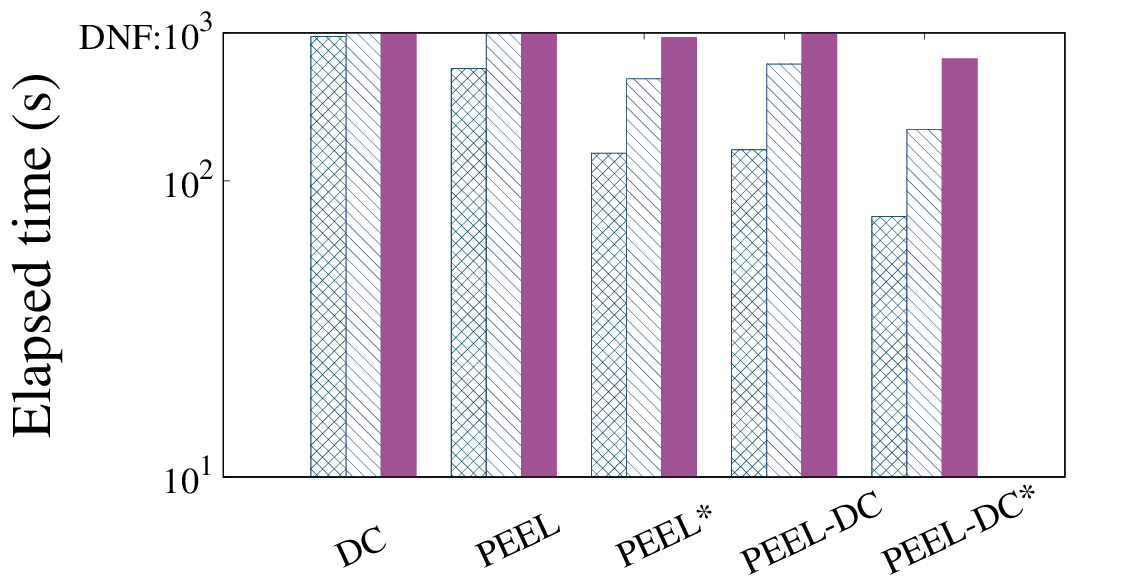}
		\caption{Btc2013}
		\label{fig:btc2013-l}
	\end{subfigure}
	\begin{subfigure}[b]{0.22\textwidth}
		\centering
            \includegraphics[width=\textwidth]{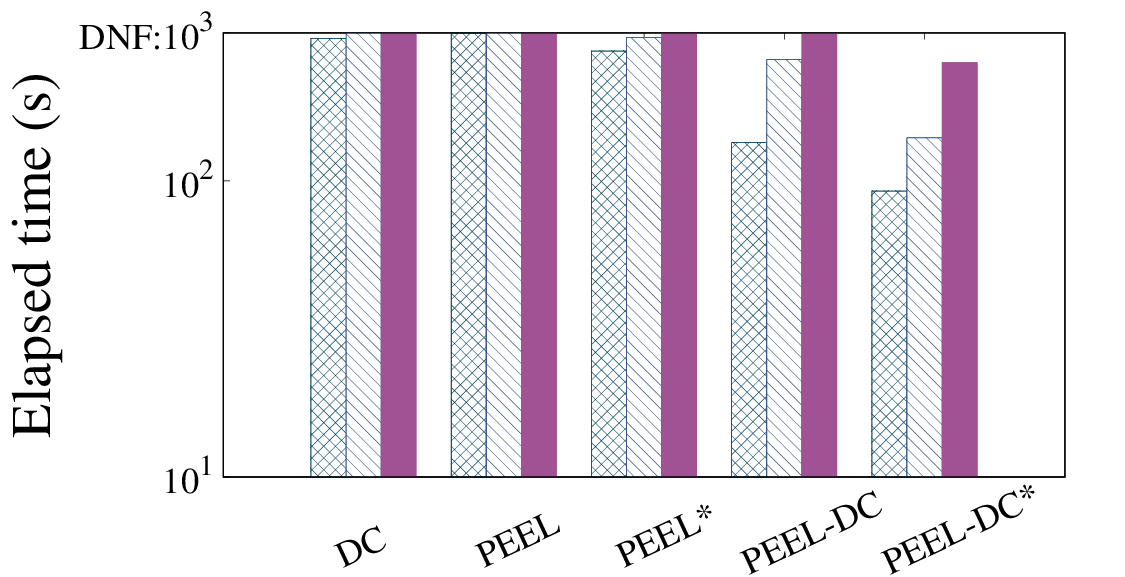}
		\caption{Eth2021}
		\label{fig:eth2021-n}
	\end{subfigure}
 \captionsetup{width=1\linewidth}
 \vspace{-0.5em}
	\caption{Elapsed time by varying $n$ (with $k=6$)}
	\label{fig:l}
\end{figure}

\stitle{Improvement of \BWCC{}.} We further investigate the impact of the divide-and-conquer technique. Our experiments, as shown in Figure~\ref{fig:overall-performance}, demonstrate that \GDYWCC{} (Section~\ref{subsec:dcpeel}) is up to $219$ (resp. $129$ and $23$) times faster than \Greedy{} \wrt $\Delta_d$ (resp. $\Delta_w$ and $\Delta_m$). On average, \GDYWCC{} is $117$ (resp. $40$ and $10$) times faster than \Greedy{} \wrt $\Delta_d$ (resp. $\Delta_w$ and $\Delta_m$). Additionally, we compare the elapsed time of \Greedy*{} and \GDYWCC*{}. The findings demonstrate that \GDYWCC{} achieves speed improvements of up to $81$ times (and $55$ times and $24$ times) compared to \Greedy{} for the different time spans $\Delta_d$, $\Delta_w$, and $\Delta_m$, respectively. On average, \GDYWCC*{} is $46$ (resp. $24$ and $9$) times faster than \Greedy*{} \wrt $\Delta_d$ (resp. $\Delta_w$ and $\Delta_m$). This verifies that it is efficient to divide sets $\Src$ and $\Dst$ into smaller subsets $\Src_i$ and $\Dst_i$, resulting in significantly fewer combinations of $\Src_i$~and~$\Dst_i$ in query evaluation.

\begin{figure}[tb]
\begin{minipage}{\linewidth}
        \center
        \includegraphics[width=1\textwidth]{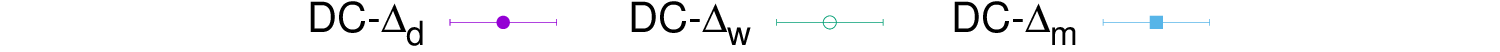}
    \end{minipage}
	\centering
	\begin{subfigure}[b]{0.2\textwidth}
		\centering
		\captionsetup{width=1\linewidth}
            \includegraphics[width=\textwidth]{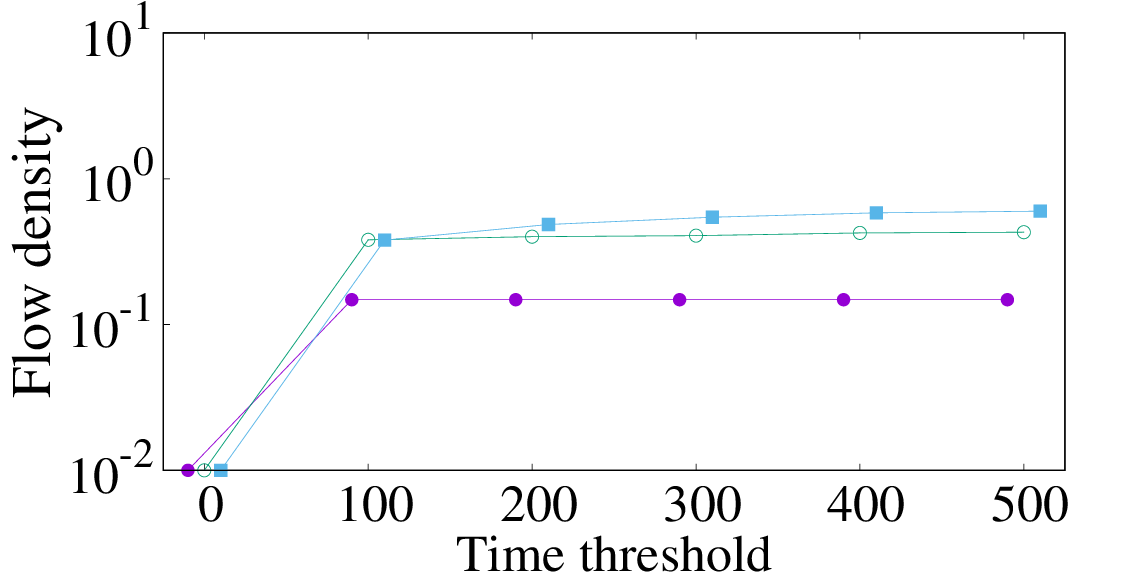}
		\caption{Btc2013}
		\label{fig:btc2013-truth}
	\end{subfigure}
	\begin{subfigure}[b]{0.2\textwidth}
		\centering
            \includegraphics[width=\textwidth]{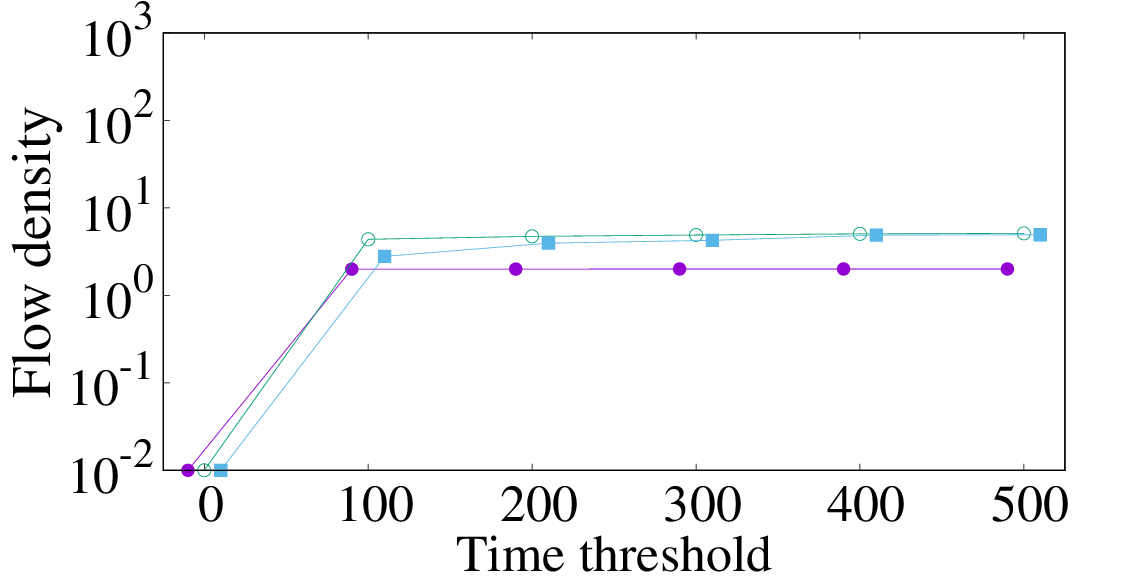}
		\caption{Eth2021}
		\label{fig:eth2021-truth}
	\end{subfigure}
    \captionsetup{width=1.05\linewidth}
    \vspace{-0.5em}
	\caption{Flow density of \BWCC{} for queries of different timespans (denoted by the subscript) under different time thresholds}
	\label{fig:effect-truth}
\end{figure}

\stitle{Improvement of Pruning.} We next evaluate the impact of the pruning technique (Line~\ref{algo:greedyprune:break} of Algorithm~\ref{algo:greedyprune}). As shown in Figure~\ref{fig:overall-performance}, \Greedy*{} is up to $17$ (resp. $278$ and $934$) times faster than \Greedy{} \wrt $\Delta_d$ (resp. $\Delta_w$ and $\Delta_m$). On average, \Greedy*{} is $4$ (resp. $72$ and $321$) times faster than \Greedy{} \wrt $\Delta_d$ (resp. $\Delta_w$ and $\Delta_m$). We also compare the performance improvement of the pruning technique in \GDYWCC*{}. \GDYWCC*{} is up to $6$ (resp. $9$ and $7$) times faster than \GDYWCC{} \wrt $\Delta_d$ (resp. $\Delta_w$ and $\Delta_m$). On average, \GDYWCC*{} is $4.0$ (resp. $5.4$ and $4.4$) times faster than \GDYWCC{} \wrt $\Delta_d$ (resp. $\Delta_w$ and $\Delta_m$). The reason for such an improvement is that the lower bounds of the peeling flow prune some unnecessary enumerations.

\stitle{Impact of the Timespans.} We evaluate the impact of the timespans by varying the length of the timespans on the performance of the compared algorithms. a) Not surprisingly, as the length of the timespans increases, all algorithms take longer to finish. \tr{When the duration of the timespans increases from $\Delta_d$ to $\Delta_w$, the elapsed time of \BWCC{} (resp. \Greedy{}, \Greedy*{}, \GDYWCC{}, and \GDYWCC*{}) increases by $576$ (resp. $48$, $57$, $249$, and $202$) times. When the length of the timespans increases from $\Delta_w$ to $\Delta_m$, the elapsed time of \BWCC{} (resp. \Greedy{}, \Greedy*{}, \GDYWCC{}, and \GDYWCC*{}) increases by $3$ (resp. $10$, $13$, $33$ and $31$) times.} This is consistent with the time complexity of the maximum flow $O(M)$, where $M=|V|^2|E|$, as $M$ is the factor of the complexities of the algorithms for \SDMF{}. b) It is important to note that \BWCC{}, \Greedy{} and \Greedy*{} fail to return the densest flow over long timespans, whereas \GDYWCC{} and \GDYWCC*{} can detect the densest flow over all timespans. As indicated in Section~\ref{sec:appr}, the main bottleneck of \SDMF{} is the time to calculate the maximum flow, which can be well reduced by \GDYWCC{} and \GDYWCC*{}.

\eat{
\begin{figure}[tb]
	\centering
	\begin{subfigure}[b]{0.24\textwidth}
		\centering
		\captionsetup{width=1.3\linewidth}
		\includegraphics[width=\textwidth]{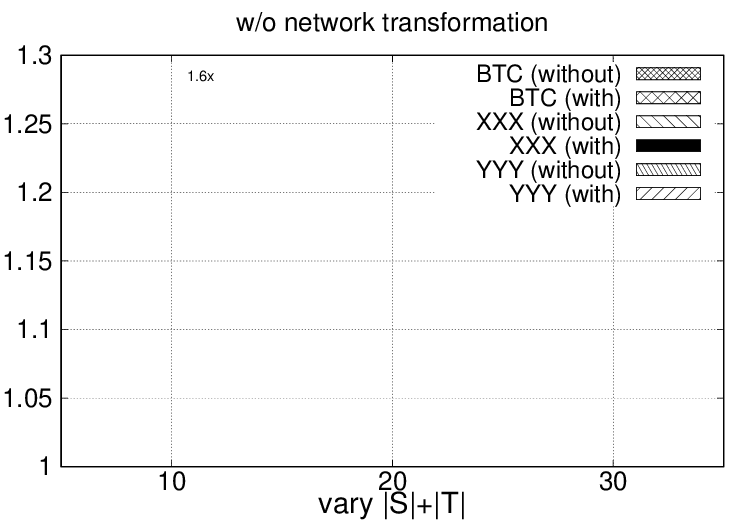}
		\caption{Btc2013}
		\label{fig:btc2013-k}
	\end{subfigure}
	\hfill
	\begin{subfigure}[b]{0.24\textwidth}
		\centering
		\includegraphics[width=\textwidth]{./figures/STMF-tran.eps}
		\caption{xxx}
		\label{fig:xxx-k}
	\end{subfigure}
	\caption{Average runtimes of algorithms by varying $k$}
	\label{fig:k}
\end{figure}
}

\begin{figure*}[tb]
    \centering
    \includegraphics[width=\linewidth]{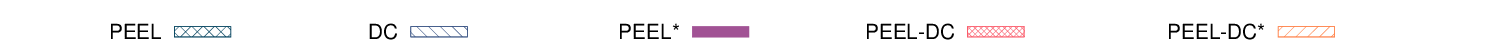}
     \setlength\tabcolsep{-3pt}
    \begin{tabular}{ccccc}
        \begin{minipage}{0.2\linewidth}
            \centering
            \includegraphics[width=\textwidth]{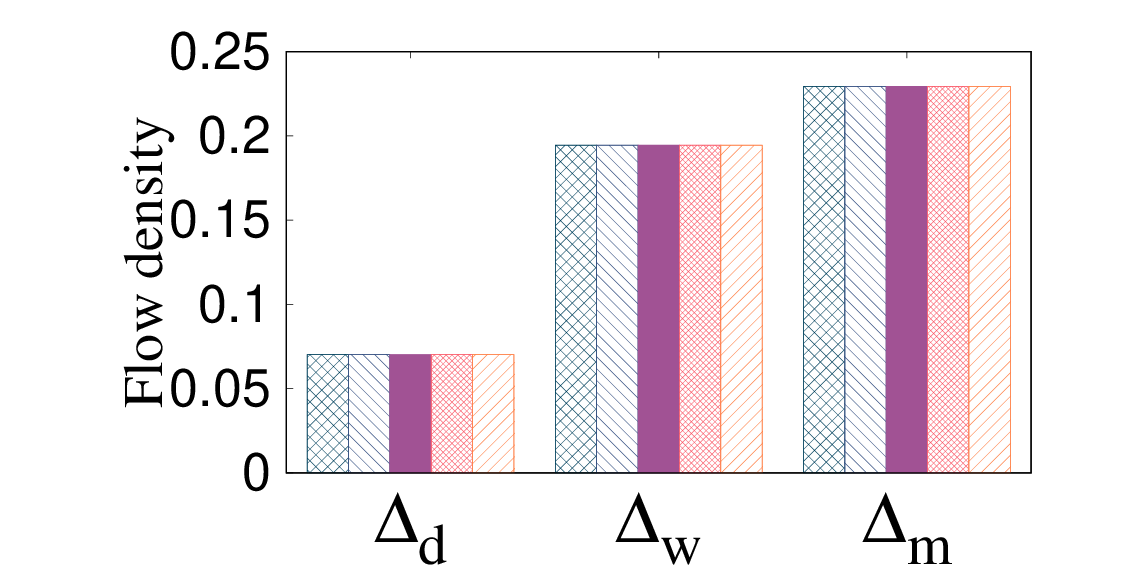}
            \subcaption{Btc2011}
            \label{fig:effect-2011-d}
        \end{minipage}
        & \begin{minipage}{0.2\linewidth}
            \centering
            \includegraphics[width=\textwidth]{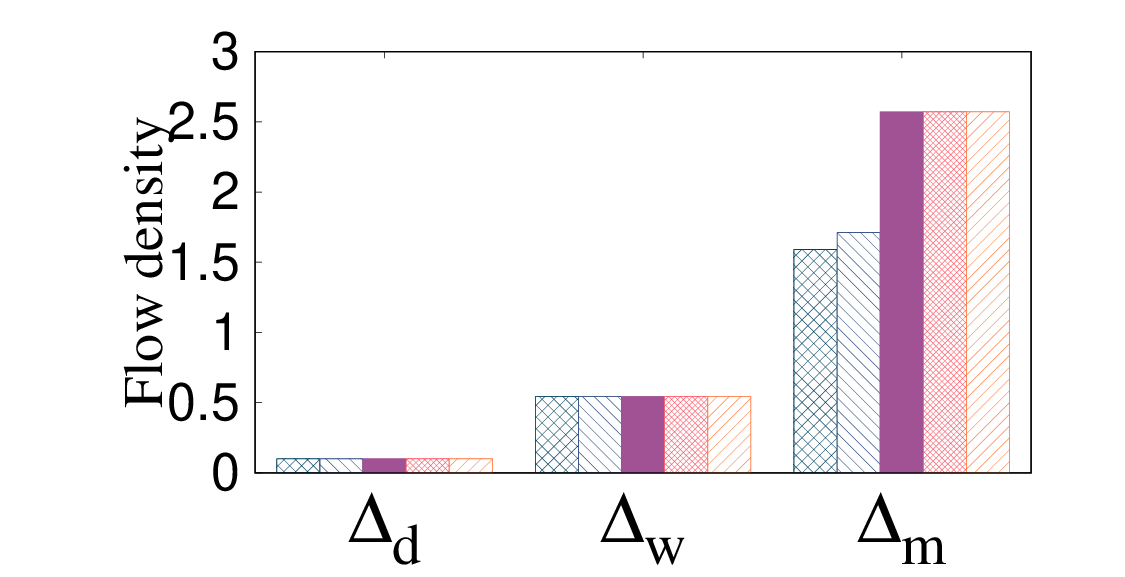}
            \subcaption{Btc2012}
            \label{fig:effect-2012-d}
        \end{minipage}
        & \begin{minipage}{0.2\linewidth}
            \centering
            \includegraphics[width=\textwidth]{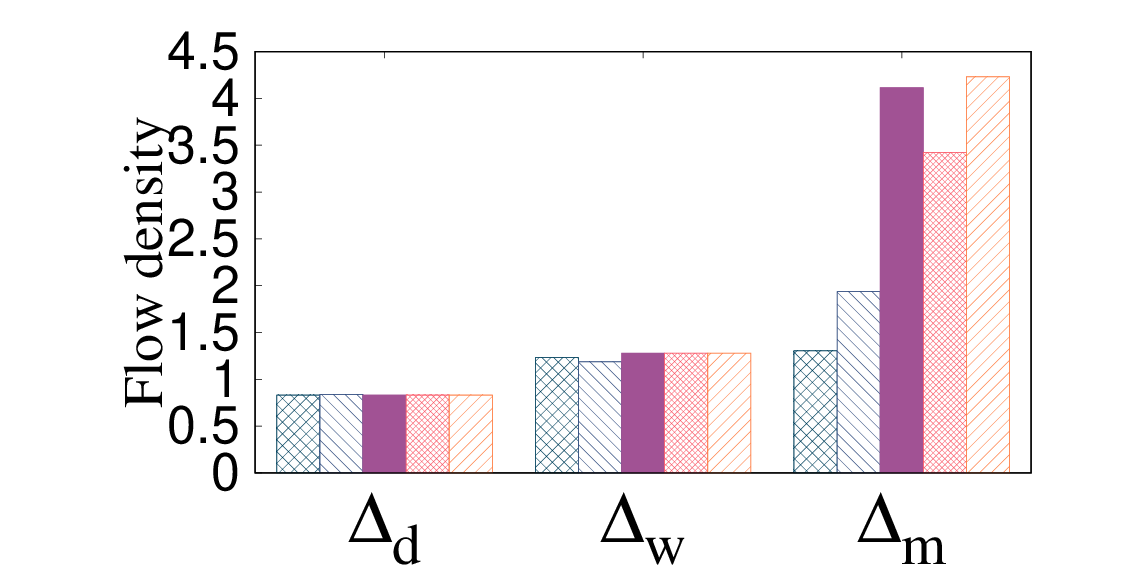}
            \subcaption{Btc2013}
            \label{fig:effect-2013-d}
        \end{minipage}
        & \begin{minipage}{0.2\linewidth}
            \centering
            \includegraphics[width=\textwidth]{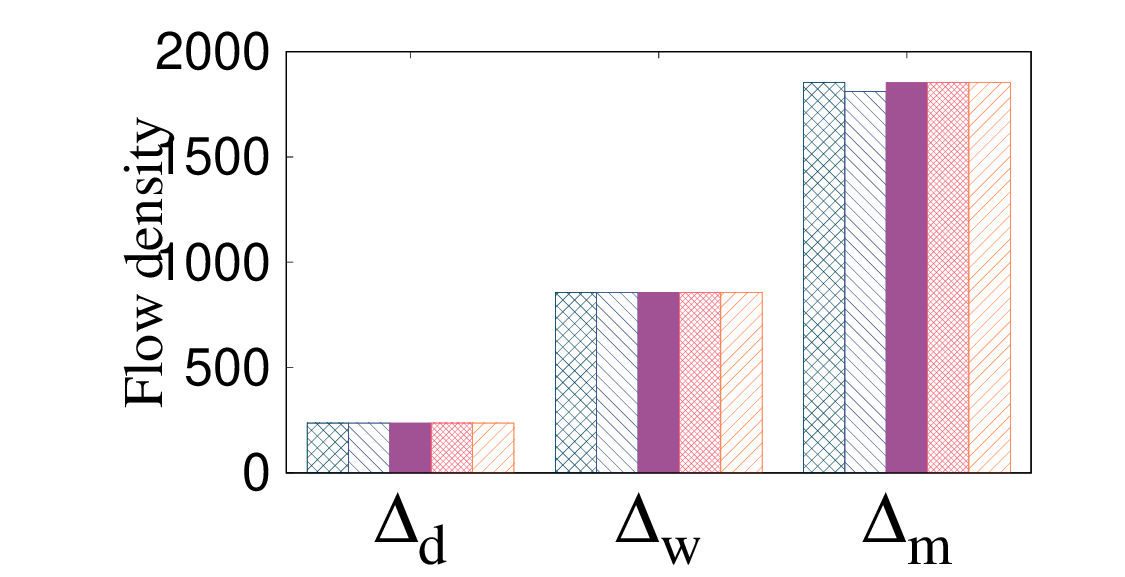}
            \subcaption{Eth2016}
            \label{fig:effect-xxx-d}
        \end{minipage}
        & \begin{minipage}{0.2\linewidth}
            \centering
    \includegraphics[width=\textwidth]{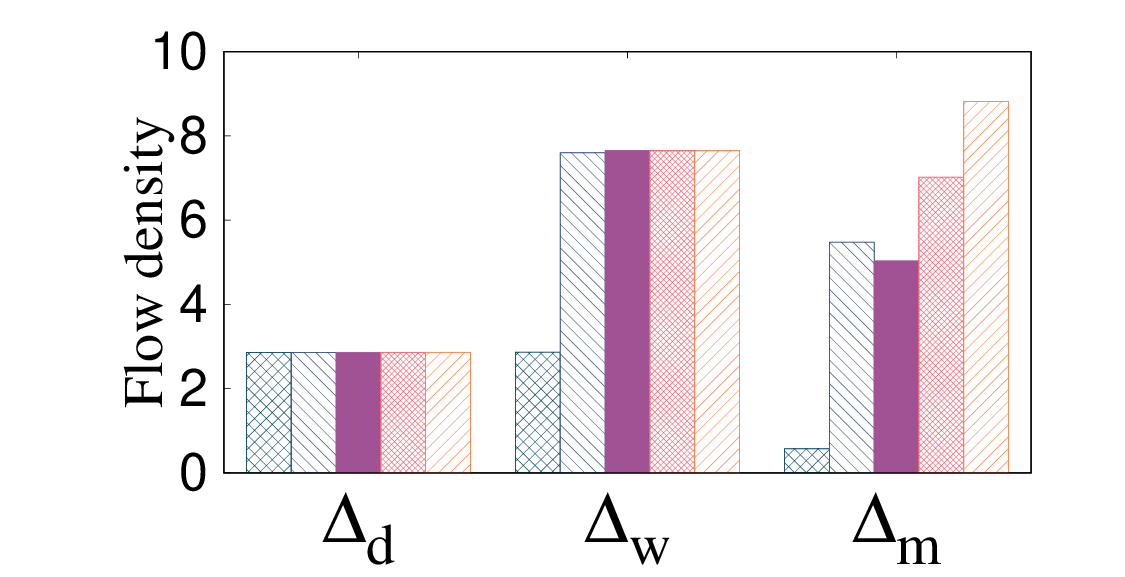}
            \subcaption{Eth2021}
            \label{fig:effect-yyy-d}
        \end{minipage}
    \end{tabular}
    \vspace{-0.5em}
    \caption{Flow density of algorithms on the five real-world datasets under default settings}
    \label{fig:overall-effectiveness}
\end{figure*}

\stitle{Impact of $n$.} The time complexity of the compared algorithms using \Greedy{} is relevant to the parameter $n$, \ie the size of the query vertex sets. In Figure~\ref{fig:l}, we present the runtimes of the compared algorithms on two larger datasets, Btc2013 and Eth2021, for different values of $n$. As expected, the runtimes of the compared algorithms increase as $n$ increases. 

\tr{This is due to the fact that a larger $n$ results in more combinations of subsets of $S$ and $T$. Additionally, similar to the results shown in Figure~\ref{fig:overall-performance}, \GDYWCC* is the most efficient among the compared algorithms for different settings.}

\begin{table}[!tb]
\caption{\jiaxin{Peak memory (MB) of algorithms on five real-world datasets under the default settings. Numbers in parentheses are normalized to \BWCC{} (lower is better).}}
\label{tab:peak-mem}
\centering
\resizebox{\linewidth}{!}{%
\jiaxin{\begin{tabular}{|c||c|c|c|c|c|}
\hline
\textbf{Dataset} 
& \textbf{\BWCC{}} 
& \textbf{\Greedy{}} 
& \textbf{\Greedy*{}} 
& \textbf{\GDYWCC{}} 
& \textbf{\GDYWCC*{}} \\
\hline
Btc2011 
& 1,129 (100.00\%) 
& 11 (0.97\%) 
& 11 (0.97\%) 
& 6 (0.53\%) 
& 6 (0.53\%) \\
Btc2012 
& 1,176 (100.00\%) 
& 65 (5.53\%) 
& 65 (5.53\%) 
& 25 (2.13\%) 
& 25 (2.13\%) \\
Btc2013 
& 1,160 (100.00\%) 
& 115 (9.91\%) 
& 115 (9.91\%) 
& 47 (4.05\%) 
& 47 (4.05\%) \\
Eth2016 
& 1,116 (100.00\%) 
& 42 (3.76\%) 
& 42 (3.76\%) 
& 15 (1.34\%) 
& 15 (1.34\%) \\
Eth2021 
& 1,641 (100.00\%) 
& 1,200 (73.12\%) 
& 1,200 (73.12\%) 
& 478 (29.13\%) 
& 478 (29.13\%) \\
IBM & 27,050 (100.00\%) & 27,048 (99.99\%) & 26,929 (99.55\%) & 8,532 (31.54\%) & 8,532 (31.54\%) \\
\hline
\end{tabular}}
}
\end{table}

\noindent \jiaxin{\stitle{Peak Memory Usage.}
To further assess space efficiency, we profiled the peak memory of each method on the same platform and query set. Table~\ref{tab:peak-mem} reports the peak memory consumption in megabytes (MB) on all datasets, with values in parentheses normalized to \BWCC{}. Overall, the reduction in memory footprint mainly comes from \emph{the search space reduction} and \emph{avoiding maintaining a large number of max-flow instances simultaneously}, rather than from pruning itself. Compared to \BWCC{}, the \Greedy{}-based variants use substantially less memory, because \Greedy{} limits the amount of intermediate max-flow state that needs to be maintained. As a result, \Greedy*{} (resp. \GDYWCC*{}) consumes only around $32.14\%$ (resp. $11.46\%$) of \BWCC{}'s peak memory on average.}

\subsection{Effectiveness of $\name{}$}\label{Exp-Effectiveness}

\tr{
Only \BWCC{} computes the exact densest flow of the \SDMF{} problem, while the remaining four algorithms provide approximate answers. Therefore, we examine the effectiveness of the five algorithms. 
}

\stitle{Flow Density of \BWCC.} In Figure~\ref{fig:effect-truth}, we present the average maximum flow densities of \BWCC{} for $1{,}000$ random queries under different thresholds of elapsed time for the three types of timespans in Btc2013 and Eth2021 datasets. In a nutshell, the \textit{denser} the detected flow, the \textit{more effective} the method. The computed maximum flow density increases rapidly in the first $100$ seconds. However, after $500$ seconds, the flow density value increases slowly in all scenarios. Based on these observations, we set the elapsed time threshold to $1{,}000$ seconds and \textit{\BWCC{} is configured to return the densest flow found within this time threshold.}

\stitle{Approximation Quality.} \jiaxin{We further compare the solution quality of our peeling-based methods against the exact baseline \BWCC{} on instances where \BWCC{} is able to complete. The results show that all peeling variants produce solutions extremely close to the optimum. Specifically, \Greedy{} and \Greedy*{} differ from \BWCC{} by only $0.03\%$ on average, while the hybrid variants \GDYWCC{} and \GDYWCC*{} further reduce the average difference to $0.01\%$. These results indicate that, in practice, the proposed peeling-based algorithms incur negligible loss in solution quality while achieving performance improvements.}

\stitle{Overall Effectiveness (Figure~\ref{fig:overall-effectiveness}).} We present the effectiveness of the five techniques in terms of flow density. a) In most cases, \GDYWCC*{} found the densest flow among the five techniques. In our experiments, \GDYWCC*{} is capable of returning the results \textit{before} the threshold. On average, \GDYWCC*{} detects $1.15$ (resp. $1.26$, $1.03$, and $1.03$) times denser flow than \BWCC{} (resp. \Greedy{}, \Greedy*{} and \GDYWCC{}). b) Compared to the exact algorithm \BWCC{}, we observe that although \BWCC{} is able to return answers for all queries on smaller datasets, such as Btc2011, \GDYWCC*{} is still very competitive. We denote the relative error by $\epsilon=\frac{g-\hat{g}}{g}$ to measure the densest quality, where $\hat{g}$ and $g$ are the flow densities of \BWCC{} and \GDYWCC*{}, respectively. Our experiment shows that \GDYWCC*{} introduces an error of only $0.2\%$, \wrt $\Delta_d$. It is noteworthy that, despite \BWCC{}'s capability to return exact answers, \eat{its reliance on exponential many combinations hinders its ability to return precise results within a designated timeframe. In particular, }\BWCC{} cannot finish all possible subset combinations, resulting in a lower flow density at the termination due to the time threshold.


\stitle{Impact of $k$.} In Figure~\ref{fig:effect-effect-k}, we examine the impact of the parameter $k$ on the flow densities. By varying $k$ to $6$, $8$, and $10$, we observe the following: a) \GDYWCC*{} consistently outperforms the other algorithms in finding the densest flows. b) As $k$ increases, the flow density decreases.\eat{ We observe that some fraudsters form smaller groups, which results in a higher volume of money transfers. }

\begin{figure}[tb] 
        \includegraphics[width=0.9\linewidth]{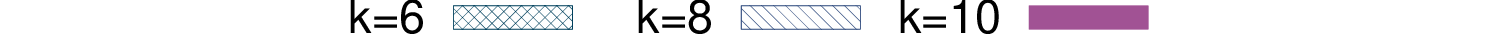} 
	\centering
	\begin{subfigure}[b]{0.23\textwidth}
		\centering
		\includegraphics[width=\textwidth]{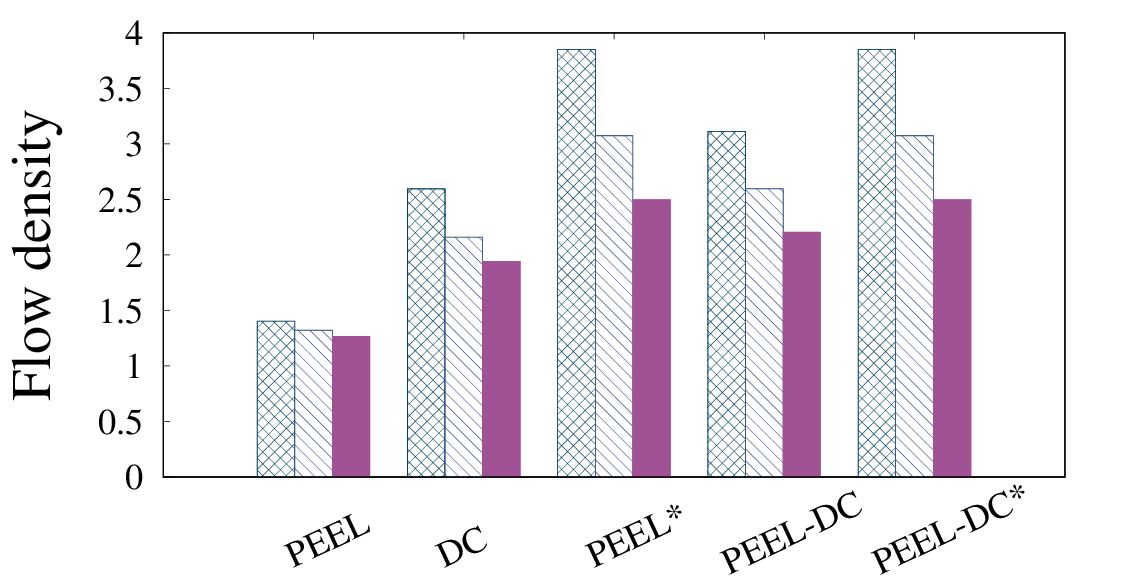}
		\caption{Btc2013}
		\label{fig:btc2013-k}
	\end{subfigure}
	\begin{subfigure}[b]{0.23\textwidth}
		\centering
		\includegraphics[width=\textwidth]{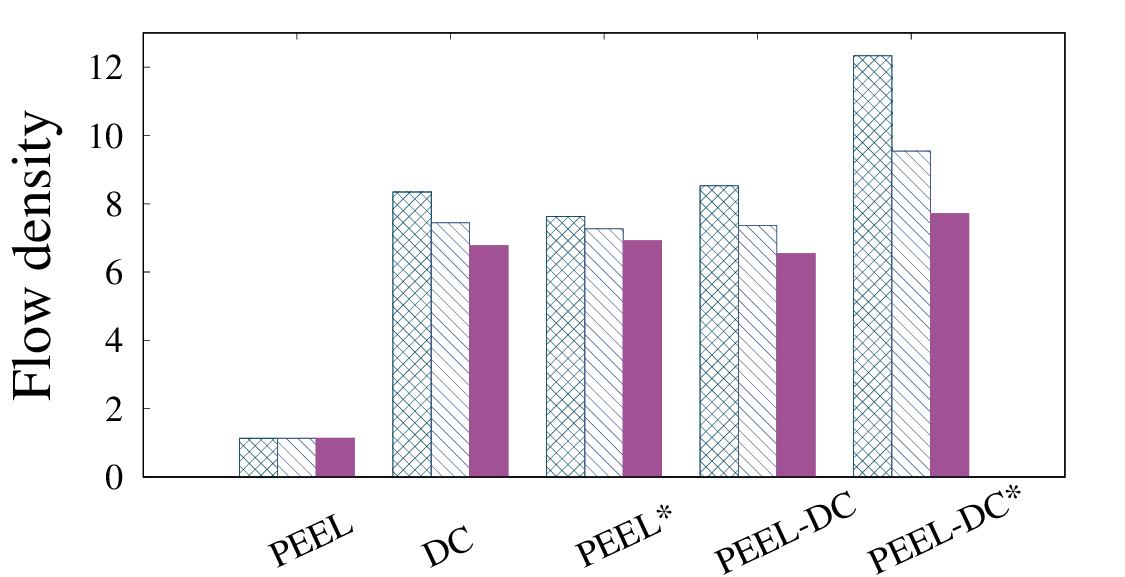}
		\caption{Eth2021}
		\label{fig:eth2021-effect-k}
	\end{subfigure}
    \captionsetup{width=1.1\linewidth}
    \vspace{-0.5em}
	\caption{Flow density of algorithms by varying $k$ (with $n=32$)}
	\label{fig:effect-effect-k}
\end{figure}

\tr{
\begin{figure}[tb]
    \begin{minipage}{\linewidth}
        \center
        \includegraphics[width=0.9\textwidth]{fig-script/fig/title-n.eps}
    \end{minipage}
	\centering
	\begin{subfigure}[b]{0.23\textwidth}
		\centering
		\includegraphics[width=\textwidth]{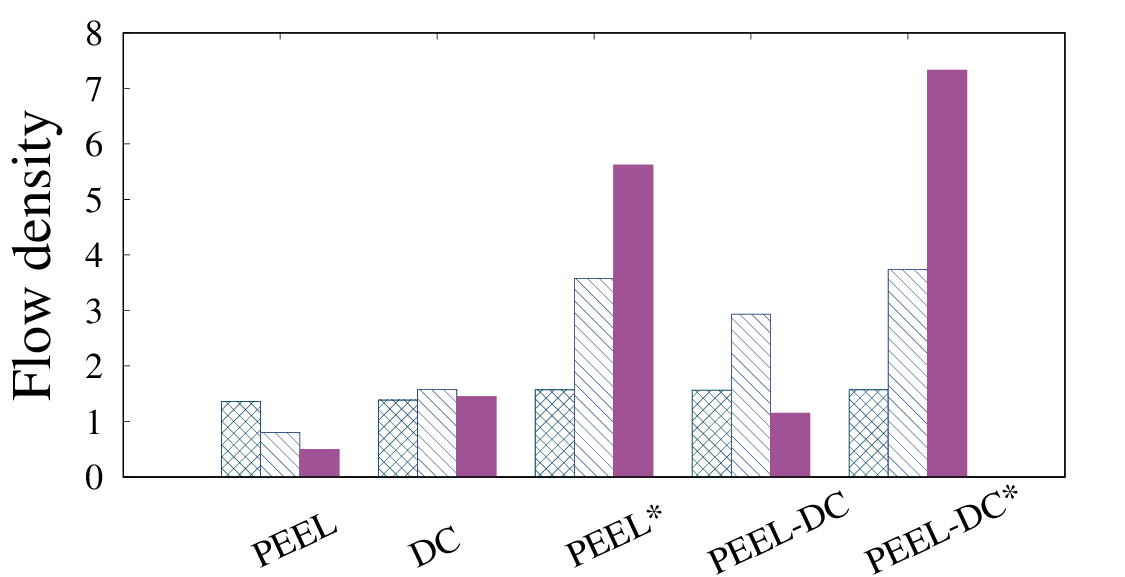}
		\caption{Btc2013}
		\label{fig:btc2013-effect-l}
	\end{subfigure}
	\begin{subfigure}[b]{0.23\textwidth}
		\centering
		\includegraphics[width=\textwidth]{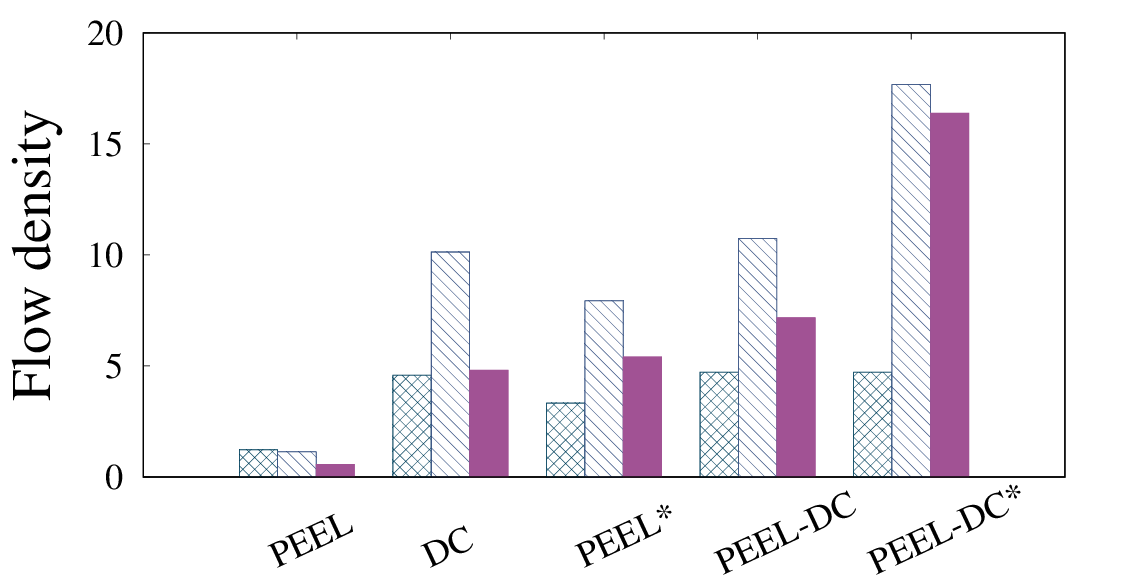}
		\caption{Eth2021}
		\label{fig:eth2021-effect-l}
	\end{subfigure}
    \captionsetup{width=1.1\linewidth}
    \vspace{-0.5em}
	\caption{Flow density of algorithms by varying $n$ (with $k=6$)}
	\label{fig:effect-l}
\end{figure}
}

\tr{
\stitle{Impact of $n$.} As we vary the parameter $n$ in our experiments, we observe that  \GDYWCC*{} significantly outperforms the other techniques. Specifically, when $n=16$, \GDYWCC*{} is able to find $1.13$ (resp. $1.15$, $1$, and $1.01$) times denser flow than \BWCC{} (resp. \Greedy{}, \Greedy*{} and \GDYWCC{}). As $n$ increases, the advantage of \GDYWCC*{} becomes even larger. When $n=64$, \GDYWCC*{} can find up to $5.05$ (resp. $14.79$, $1.30$, and $6.37$) times denser flow than \BWCC{} (resp. \Greedy{}, \Greedy*{} and \GDYWCC{}). Our results demonstrate that \GDYWCC*{} consistently outperforms the compared algorithms.
}

\tr{
In conclusion, our experiments show that \GDYWCC*{} is an effective technique for detecting dense flows in temporal transaction networks. It outperforms the other algorithms in terms of efficiency, effectiveness, and scalability under various settings, including different values of $k$ and $n$. \GDYWCC*{} is able to find the densest flow even in larger datasets such as Btc2013 and Eth2021, and its performance does not deteriorate as the value of $k$ or $n$ increases.
}

\eat{
\etitle{Size of transformed networks.} To obtain a temporal-free network, we transform a \TFNet{} to \RTFNet{} as introduced in Sec~\ref{sec:transform}. We report the size of the network after reduction and transformation in Tab.~\ref{table:Statistics}. For real-life networks, $\name$ reduces $XX\%$ of the networks.

\etitle{Accuracy.} We randomly selected two sets of vertices, $\Src$ and $\Dst$, from $|V|$. 1) We compared the $\MFavg(\Src, \Dst)$ between $\Src$ and $\Dst$ on \TFNet{} and \RTFNet{} to show the errors without graph transformation. 2) We also shown the effectiveness of $\Greedy$ in practice. The exact $\MFavg(\Src, \Dst)$ is computed by $\Baseline$. For consistence, we denoted answers returned by $\Greedy$ and $\Baseline$ on \TFNet{} by $\hat{\MFavg}(\Src, \Dst)$. The error is denoted by $\epsilon = \frac{|\hat{\MFavg}(\Src, \Dst) - \MFavg(\Src, \Dst)|}{\MFavg(\Src, \Dst)}$. We repeated the above procedure $1000$ times and got the average error. Since $\Baseline$ cannot scale to process large $\Src$ and $\Dst$ when $|\Src| + |\Dst| > \red{XX}$, we set $|\Src| + |\Dst| = \{\red{XX1, XX2, XX3}\}$ as we presented in Fig.~\ref{fig:accurcy-btc}. On BTC2009, $\bar{\epsilon}$ of $\Greedy$ (resp. $\Baseline$-$\mathsf{TF}$) is $\red{YY1}$ (resp. $\red{YY2}$). Similarly, $\Greedy$ also very accurate on BTC2010$\sim$BTC2013 in practice.

\stitle{Effectiveness.} We also evaluate the effectiveness of our approximation algorithm (detailed in Sec.~\ref{sec:appr}).

\begin{figure}[tb]
	\centering
	\begin{subfigure}[b]{0.22\textwidth}
		\centering
		\includegraphics[width=\textwidth]{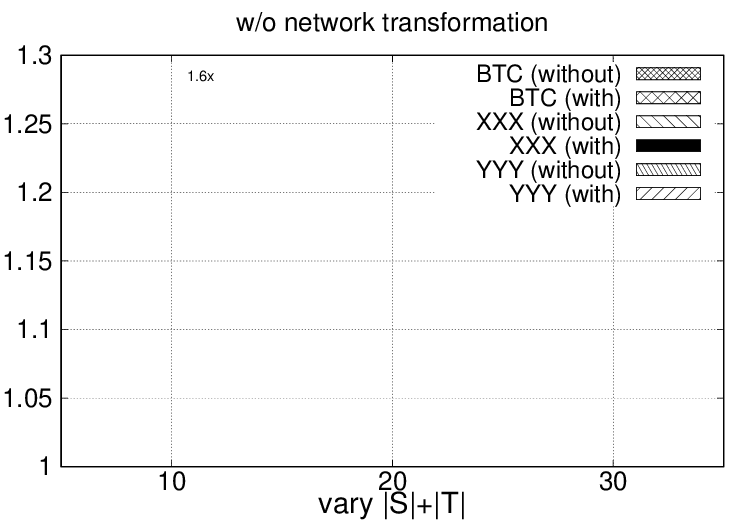}
		\caption{Accuracy on BTC}
		\label{fig:accurcy-btc}
	\end{subfigure}
	\hfill
	\begin{subfigure}[b]{0.22\textwidth}
		\centering
		\includegraphics[width=\textwidth]{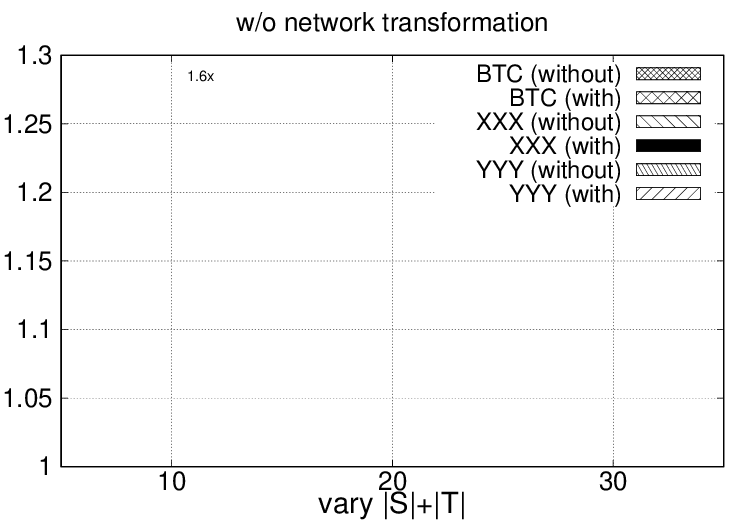}
		\caption{Accuracy on XXX}
		\label{fig:accuracy-xxx}
	\end{subfigure}
	\caption{Comparison between {\TFNet} and {\RTFNet}}
	\label{fig:sizeappr}
\end{figure}

}

\section{INDUSTRIAL APPLICATION}

\subsection{Case Study 1: Application of $\name{}$ within \Grab{}}\label{sec:real}

\stitle{Dataset Description.} To validate the practicality of $\name{}$, we tested its efficiency and effectiveness on an industry dataset, \gfg{}, provided by \Grab{}. \gfg{} comprises a transaction flow network with $|V| = 3.38$M nodes and $|E| = 28.64$M edges, where each node represents entities such as users, merchants, digital wallets, or card numbers. Each edge represents transactions or transfer records between these nodes (all data have been normalized to large random values for business privacy).

\stitle{Flow Density as a Fraud Indicator.} We examined the flow densities of different fraudulent and benign groups, using ground truth provided by Grab. For each group size, we collected 100 samples of fraudulent groups. If a node among the suspects had larger outgoing funds than incoming funds, we allocated it to $\Src$; if the incoming flow was greater, it was assigned to $\Dst$. As illustrated in Figure~\ref{fig:motivation}, fraudulent groups tend to exhibit higher flow density compared to benign ones, underscoring flow density as a strong indicator of fraudulent activity. The figure also reveals that smaller groups have higher flow density, highlighting the necessity of controlling group size in queries. Without a minimum size threshold \( k \), queries might only return small groups, potentially overlooking larger fraud networks. By enforcing the \( k \) parameter, \SDMF{} captures extensive fraud networks, increasing the likelihood of identifying a broader range of fraudulent actors.

\begin{figure}[tb]
\centering
    \includegraphics[width=0.75\linewidth]{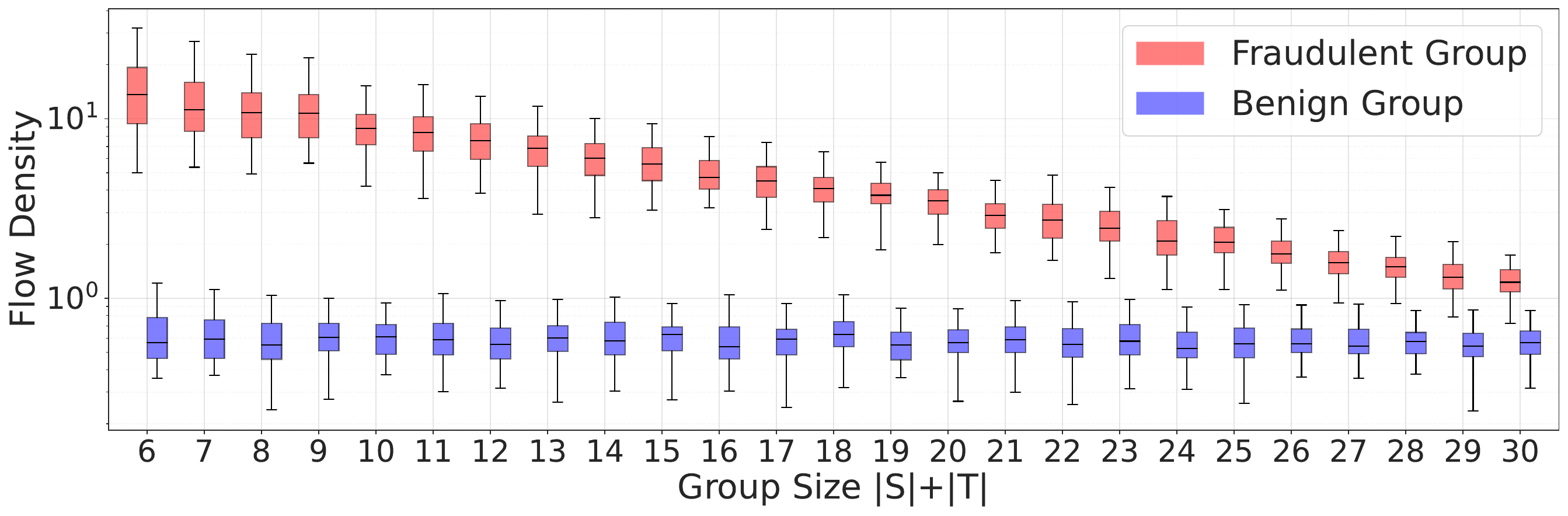}
    \vspace{-0.5em}
    \caption{Distribution of Flow Density by Group Size}\label{fig:motivation}
\end{figure}

\stitle{Query Formulation.} We used \Spade{}~\cite{jiang2023spade} to identify suspects. \jiaxin{\Spade{} supports three density metrics: \emph{dense subgraphs} (DG~\cite{charikar2000greedy}), \emph{dense weighted subgraphs} (DW~\cite{gudapati2021search}), and \emph{Fraudar}-based density (FD~\cite{hooi2016fraudar})}. We divided the detected suspects into query sets $\Src$ and $\Dst$ for each mode. If a node had larger outgoing funds than incoming funds among the suspects, we allocated it to $\Src$; if the incoming flow was greater, it was assigned to $\Dst$. We then utilized $\Src$ and $\Dst$ as queries, setting $k$ to $20\%$ of $|\Src| + |\Dst|$, i.e., $k = 470$.

\stitle{Impact of Temporal Dependency.} We compared the detection results with and without incorporating temporal dependencies. We observed that neglecting temporal dependencies led to detecting only $22.7\%$ of the flow volume. Specifically, precision merely increased from $64.93\%$ to $82.39\%$, a less substantial improvement compared to cases considering temporal dependencies. These results underscore the critical role of temporal context in improving the accuracy of fraud detection, emphasizing the need for methodologies that account for the timing of transactions.

\stitle{Efficiency and Effectiveness.} Since \BWCC{} could not finish within an hour, we limited our comparison to four algorithms of $\name{}$ and \Spade{}. The results are illustrated in Figure~\ref{fig:gfg}. \GDYWCC* took only $1.04\%$, $0.43\%$, and $0.14\%$ of the time that \Spade{} takes, respectively, yet detected $3.51\times$, $8.41\times$, and $3.70\times$ denser money flows under DG, DW, and FD modes. This significantly reduces the time needed for manual screening and verification at \Grab{}. It is worth noting that the detection pipeline in \Grab{}, \BWCC{}, could not complete the \SDMF{} queries within 24 hours, whereas \GDYWCC* typically completed the \SDMF{} queries in about one second.

\stitle{Detection of Credit Card Fraud.} Compared to the ground truth provided by \Grab{}, only $64.9\%$ of the cases reported by \Spade{} are fraudsters or fraudulent cards. The precision is enhanced from $64.9\%$ to $95.8\%$, and the flow density is increased from $0.46$ to $0.80$ (after normalizing the amounts to random values). This improvement is attributed to the flow peeling, which excludes users with lower participation.

\begin{figure}[tb]
\begin{minipage}{\linewidth}
        \centering
        \includegraphics[width=0.9\textwidth]{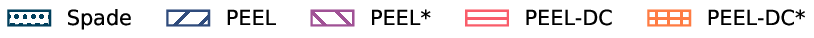}
    \end{minipage}
 	\begin{subfigure}[b]{0.235\textwidth}
		\centering
		\includegraphics[width=\textwidth]{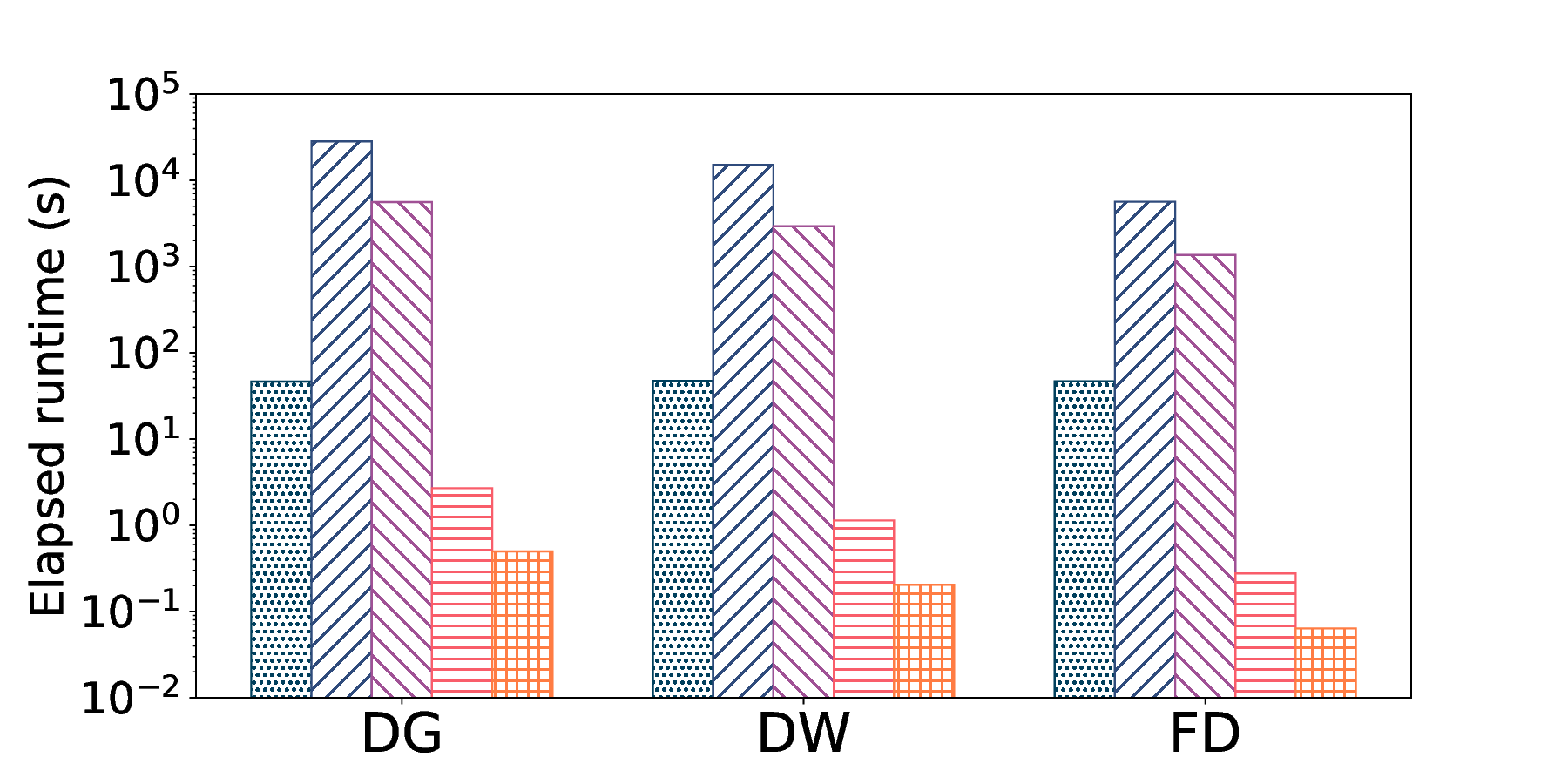}
 		\caption{Runtime}
		\label{fig:gfgruntime}
	\end{subfigure}
	\begin{subfigure}[b]{0.235\textwidth}
		\centering
		\includegraphics[width=\textwidth]{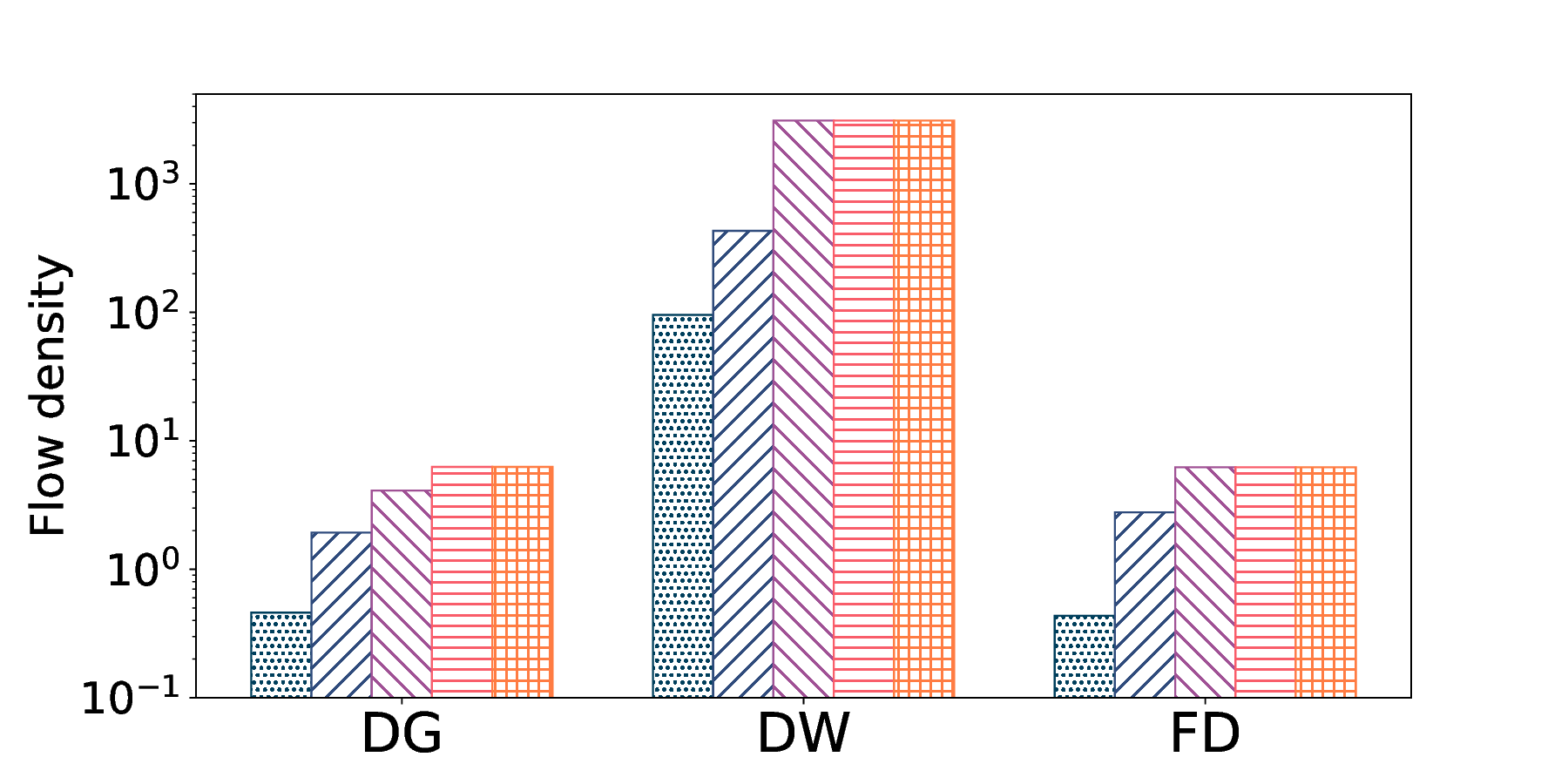}
 		\caption{Flow density}
		\label{fig:gfgdensity}
	\end{subfigure}
    \begin{center}
    \begin{subfigure}[b]{0.48\textwidth}
    \centering
    \includegraphics[width=0.8\linewidth]{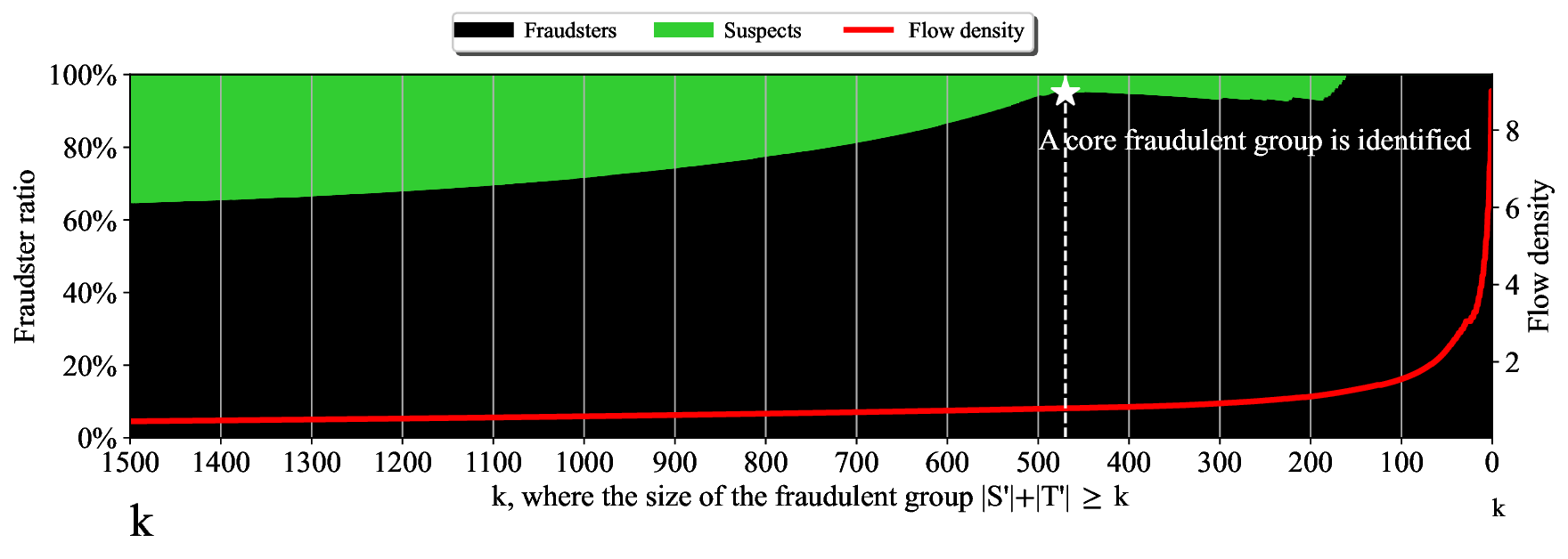}
    \vspace{-0.5em}
    \caption{Impact of $k$ to fraudster ratio and flow density}
    \label{fig:gfgvaryk}
    \end{subfigure}
    \end{center}
    \vspace{-1.2em}
    \caption{Case study on \Grab{} transaction networks}
    \label{fig:gfg}
\end{figure}

\subsection{Case Study 2: Application of $\name{}$ within NFT}\label{Exp-CaseStudy}

\begin{figure}[tb]
	\begin{center}
        \includegraphics[width=0.85\linewidth]{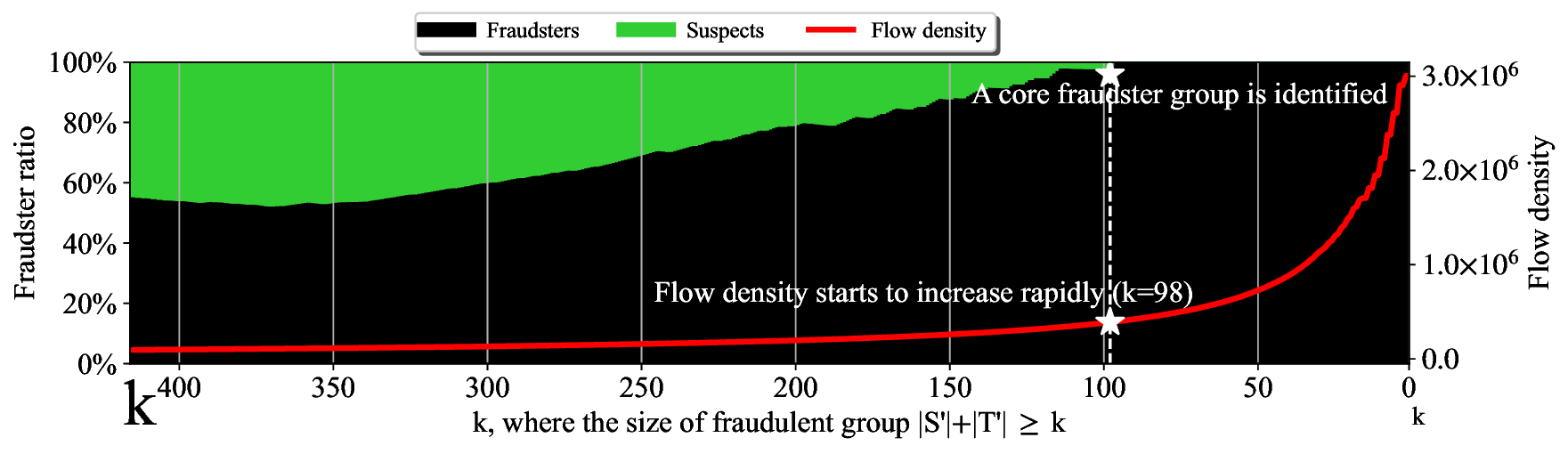}
    	\includegraphics[width=0.85\linewidth]{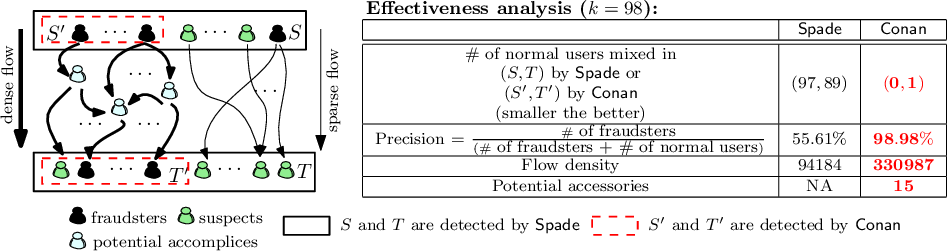}
    \end{center}
    \vspace{-1.2em}
	\caption{Case Study on NFT Networks}
	\label{fig:casestudy}
\end{figure}

We employed $\name{}$ to examine instances of wash trading fraud in the NFT network. In this study, we leveraged known wash trading fraud identified in NFT communities~\cite{dune} as ground truth. The query formulation is identical to that in Section~\ref{sec:real}. Previous research~\cite{jiang2023spade,gudapati2021search,hooi2016fraudar} has shown limited success in accurately identifying suspicious addresses; only $55.61\%$ of the suspects ($\Src$ and $\Dst$ detected by \Spade{}) are true fraudsters. $\name{}$ detects the densest flow from $\Src$ to $\Dst$ and obtains subsets $\Src'$ and $\Dst'$. We show the impact of the parameter $k$ at the top of Figure~\ref{fig:casestudy}. If $k$ is smaller, more vertices are peeled, leaving only the core fraudsters and resulting in denser flows, as indicated by the red line. When $k = 98$, the flow density increases rapidly and a core fraudulent group, $\Src'$ and $\Dst'$, is identified. The flow density increases from $94,\!184$ to $330,\!987$, and the precision increases from $55.61\%$ to $98.98\%$. We investigate the vertices in the flow between $\Src'$ and $\Dst'$ and spot 15 fraudsters not identified by \Spade{}.

\stitle{Summary.} Our case studies demonstrate the effectiveness of $\name{}$ in (1) identifying meaningful fraudulent communities in real-life transaction flow networks, (2) tracking money transfers among fraudsters, and (3) revealing more potential accomplices. By setting $k$ between $0.2 \times (|\Src| + |\Dst|)$ and $0.3 \times (|\Src| + |\Dst|)$, $\name{}$ returns a core fraudulent group.

\tr{This is because these studies ignore the chain transfer scheme and only consider one-hop transfers, which results in an unclear pattern of fraudulent activities.}

\tr{
\etitle{Query formulation.} We apply \Spade{}~\cite{jiang2023spade} to an NFT network~\cite{dune} to obtain a set of vertices of a dense subgraph (\aka suspects) denoted by $V^{s}$. To formulate the query, $\name{}$ first divides $V^{s}$ into two groups, $\Src$ and $\Dst$. If $\MFlow(u,V^{s}\setminus \{u\}) \geq \MFlow(V^{s}\setminus \{u\},u)$, $\name{}$ assigns $u$ to $\Src$. Otherwise, $u$ is assigned to $\Dst$. As shown in Figure~\ref{fig:casestudy}, only $55.61\%$ of the suspects are true fraudsters.}

\tr{The size of $\Src$ is $199$ and $102$ vertices among them are labeled fraudsters. The size of $\Dst$ is $220$ and $89$ vertices among them are labeled with fraudsters.}

\tr{We have also repeated this case study $100$ times by generating $\Src$ and $\Dst$ randomly. Most of them exhibit the similar pattern that is shown in Figure~\ref{fig:casestudy}.}

\tr{The size of $\Src'$ is $49$, and all the vertices are labeled with fraudsters. And the size of $\Dst'$ is $49$, and the $48$ vertices are labeled with fraudsters.}

\tr{Intuitively, for any $u\in V^{s}$, we compare the maximum flow coming from $u$ to other vertices $V^{s}\setminus \{u\}$ with the maximum flow coming from other vertices in $V^{s}\setminus \{u\}$ to $u$.}

%% file: exp-datasets.tex
\eat{
\begin{table*}[tb]
\caption{Statistics of real-world datasets}\label{table:Statistics}
\centering
\begin{scriptsize}
\begin{tabular}{|c|c|c|c|c|c|c|c|c|}
  \hline
  {\bf Datasets} & {\bf $|V|$} & {\bf $|E|$}  & avg. degree & $\widehat{V}$ & $\widehat{E}$ & $V_c$ & $E_c$ & avg. new edges per day \\
  \hline
  Btc2009 & 11,639 & 11,126 & 2.1  & 33,891  & 33,378  & 11,643  &  11,130 &  30.5   \\
  \hline
  Btc2010 & 179,145 & 211,172 & 2.4  & 601,489  & 633,516  & 179,588  & 211,615 & 578.6\\
  \hline
  Btc2011 & 1,987,113 & 3,909,286 & 3.9  & 9,805,685    & 11,727,858   & 1,996,947    & 3,919,120 & 10710.4\\
  \hline
  Btc2012 & 8,577,776 & 18,389,841 & 4.3  & 45,357,458  & 55,169,523  & 9,449,189  & 19,261,254 & 50383.1\\
  \hline
  Btc2013 & 19,658,304 & 43,570,830 & 4.4  & 106,799,964  & 130,712,490  & 21,039,323  & 44,951,849 & 119372.1\\
  \hline
\end{tabular}
\end{scriptsize}
\end{table*}
}

\begin{table*}[tb]
\caption{\jiaxin{Statistics of datasets and queries of $\Src$ and $\Dst$ with the default size $n=32$, where $\widehat{V}$ (\resp $\widehat{E}$) denotes the vertex set (\resp edge set) of the compressed \RTFNet{}, and $\Delta_d$ (resp. $\Delta_w$ and $\Delta_m$) denotes a time span of a day (resp. a week and a month)}}\label{table:Statistics}
\vspace{-0.5em}
\centering
\resizebox{\textwidth}{!}{
\begin{scriptsize}
\begin{tabular}{|c|c|c|c|c|c|c|c|c|m{1.1cm}<{\centering}|m{1.1cm}<{\centering}|c|m{1.05cm}<{\centering}|m{1.05cm}<{\centering}|c|}
  \hline
  \multirow{2}{*}{\bf Datasets} & \multirow{2}{*}{\bf $|V|$} & \multirow{2}{*}{\bf $|E|$} & \multicolumn{3}{c|}{avg. $|\widehat{V}|$ of the \RTFNet{}} & \multicolumn{3}{c|}{avg. $|\widehat{E}|$ of the \RTFNet{}}  & \multicolumn{3}{c|}{avg. out-degree of the vertices in $\Src$} &  \multicolumn{3}{c|}{avg. in-degree of the vertices in $\Dst$} \\ \cline{4-15}
  & & & $\Delta_d$ & $\Delta_w$ & $\Delta_m$ & $\Delta_d$ & $\Delta_w$ & $\Delta_m$ & $\Delta_d$ & $\Delta_w$ & $\Delta_m$ & $\Delta_d$ & $\Delta_w$ & $\Delta_m$ \\
  \hline
Btc2011
& 1,987,113 & 3,909,286 & 10,127 & 55,290 & 212,921 & 11,748 & 82,397 & 362,380 & 1.33 & 1.59 & 1.82 & 1.31 & 1.65 & 1.89 \\ \hline
Btc2012
& 8,577,776 & 18,389,841 & 36,651 & 220,058 & 992,282 & 52,554 & 369,522 & 1,804,755 & 1.63 & 1.89 & 2.17 & 1.57 & 1.86 & 2.10 \\ \hline
Btc2013
& 19,658,304 & 43,570,830 & 83,924 & 502,205 & 2,072,286 & 127,588 & 900,325 & 3,995,375 & 1.73 & 2.00 & 2.19 & 1.70 & 1.99 & 2.16 \\ \hline
Eth2016
& 670,559 & 13,654,676 & 19,720 & 92,361 & 368,071 & 43,209 & 311,628 & 1,394,052 & 5.69 & 6.25 & 7.15 & 3.63 & 5.67 & 4.93 \\ \hline
Eth2021
& 58,274,378 & 461,781,254 & 600,576 & 3,275,312 & 11,740,881 & 1,326,215 & 9,536,231 & 41,863,219 & 1.82 & 2.21 & 1.90 & 2.75 & 1.64 & 2.62 \\ \hline
       IBM &    2,116,168 &  179,702,229 & \multicolumn{3}{c|}{38,721,761} & \multicolumn{3}{c|}{216,307,822} & \multicolumn{3}{c|}{147.67} & \multicolumn{3}{c|}{105.57} \\ \hline
\end{tabular}
\end{scriptsize}
}
\end{table*}

%% file: exp-query-performance.tex
\begin{figure*}[tb]
        \begin{minipage}{\linewidth}
        \centering
        \includegraphics[width=1\textwidth]{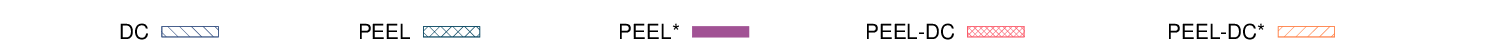}
        \end{minipage}
		\begin{minipage}{\linewidth}	
        \centering
			\begin{tabular}{cc}
				\begin{minipage}[t]{0.19\textwidth}
            \includegraphics[width=\textwidth]{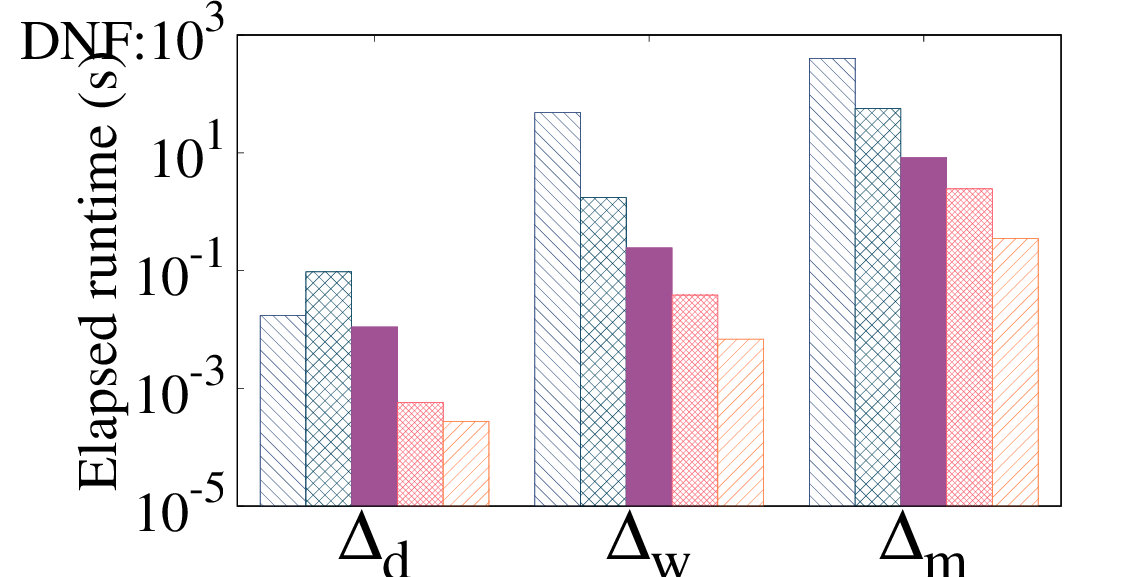}
            \vspace{-5mm}
				\subcaption{Btc2011}\label{fig:overall-2011-d}
				\end{minipage}
				\begin{minipage}[t]{0.19\textwidth}
                    \includegraphics[width=\textwidth]{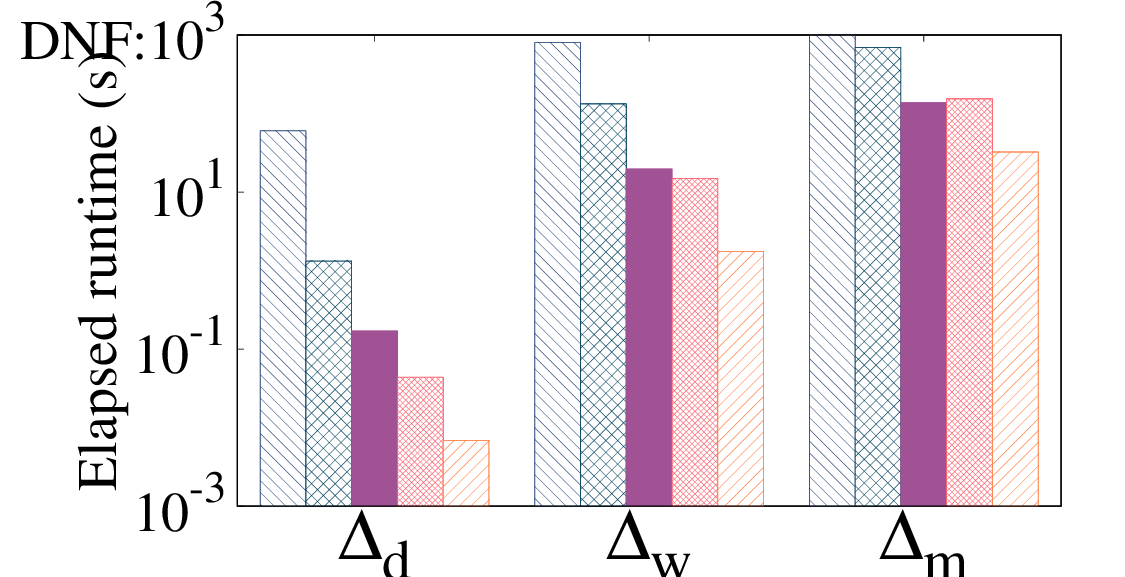}
                    \vspace{-5mm}
				\subcaption{Btc2012}\label{fig:overall-2012-d}
				\end{minipage}
				\begin{minipage}[t]{0.19\textwidth}
                    \includegraphics[width=\textwidth]{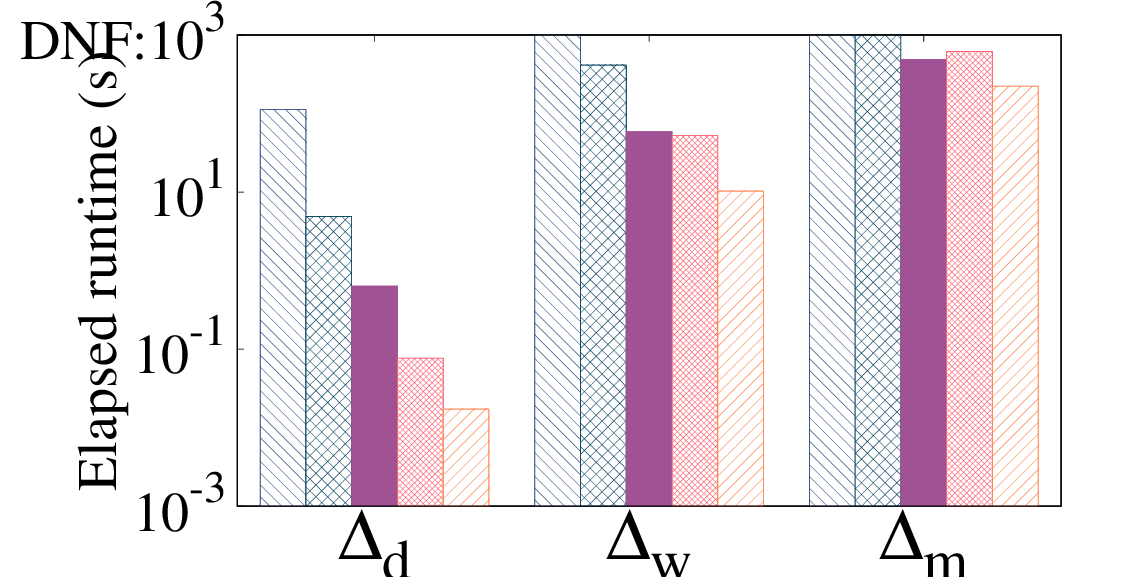}
                    \vspace{-5mm}
					\subcaption{Btc2013}\label{fig:overall-2013-d}
					\end{minipage}
					\begin{minipage}[t]{0.19\textwidth}
                    \includegraphics[width=\textwidth]{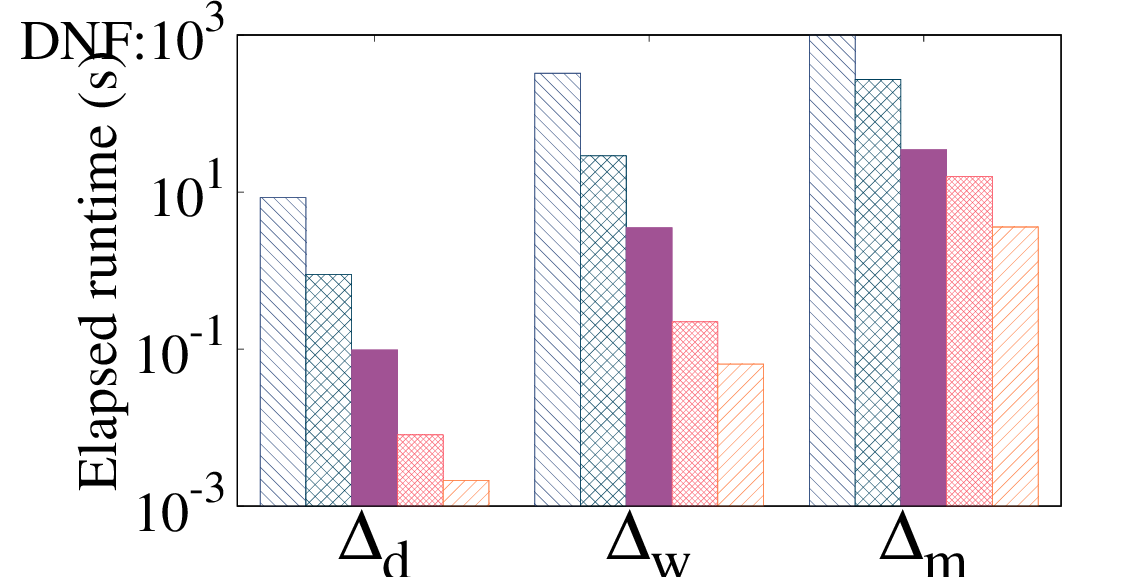}
                    \vspace{-5mm}
					\subcaption{Eth2016}\label{fig:overall-xxx-d}
					\end{minipage}
					\begin{minipage}[t]{0.19\textwidth}
                    \includegraphics[width=\textwidth]{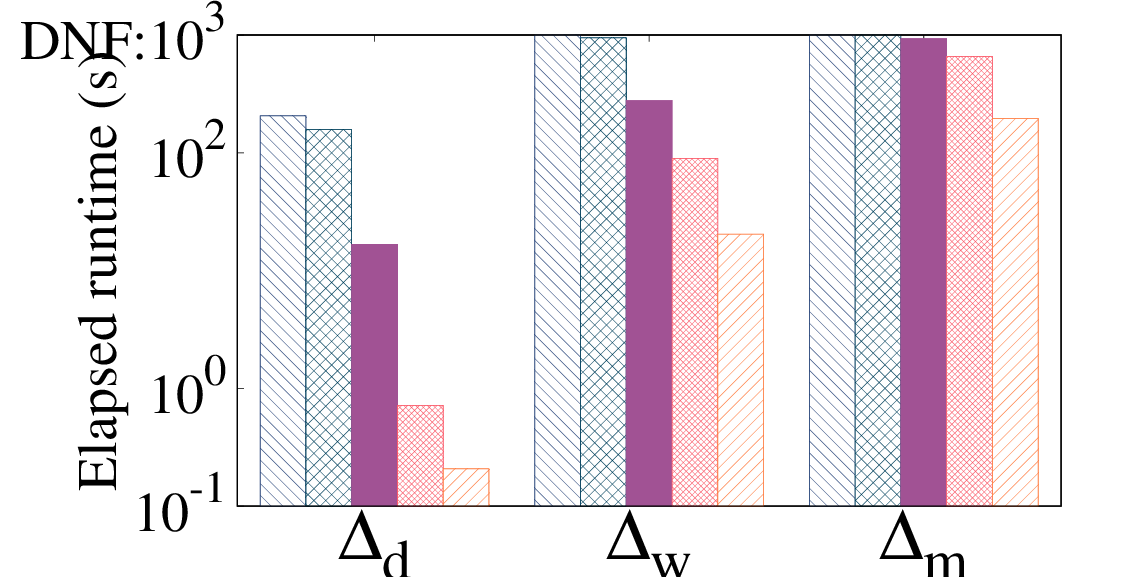}
					\subcaption{Eth2021}
					\label{fig:overall-yyy-d}
					\end{minipage}
				\end{tabular}
    \vspace{-0.5em}
	\caption{Elapsed runtime of algorithms on the five real-world datasets under the default settings}
	\label{fig:overall-performance}
	\end{minipage}
	\end{figure*}

%% file: exp-sizes.tex
\begin{table*}[tb]
\caption{Impact of optimization techniques on graph size reduction and maximum flow time efficiency ("avg. time" represents the average time taken for a random pair of source and sink.)}\label{table:Sizes}
\vspace{-0.5em}
\centering
\resizebox{0.85\textwidth}{!}{
\begin{tabular}{|c|c|c|c|c|c|c|c|c|c|}
  \hline
  \multirow{2}{*}{\bf Datasets} & \multicolumn{3}{c|}{\RTFNet{}} & \multicolumn{2}{c|}{with network reduction}  & \multicolumn{4}{c|}{with network reduction \& network compression} \\ \cline{2-10}
  & $|V|$ & $|E|$ & avg. time(s) & avg.$|V|$ & avg.$|E|$ & avg.$|V|$ & avg.$|E|$ & avg. time(s) & speedup \\
  \hline
Btc2011 & 9.8M  & 12M    & 12.62  & 0.5M   & 0.6M   & 0.1M     & 0.2M    & 0.17   & \textbf{72.7$\times$}  \\
Btc2012 & 45M   & 55M    & 67.54  & 6.8M   & 6.4M   & 1.2M     & 2.1M    & 4.47  & \textbf{15.1$\times$}  \\
Btc2013 & 107M  & 131M   & 246.16 & 18M    & 21M    & 3.8M     & 7.2M    & 16.78 & \textbf{14.7$\times$}  \\
Eth2016 & 28M   & 41M    & 57.04  & 1.9M   & 2.7M   & 0.3M     & 1.1M    & 1.23  & \textbf{46.3$\times$}  \\
Eth2021 & 982M  & 1.4B   & -      & 75M    & 100M   & 15M      & 41M     & 55.37 & -  \\

\hline



\end{tabular}
}
\end{table*}

%% file: 8-relatedworks.tex
\vspace{-1em}
\section{Related Work}\label{sec:related}

\tr{
In recent years, there is an increasing interest in query processing on temporal graphs~\cite{kosyfaki2021flow,wu2014path,wu2016reachability,zhang2019efficient,yuan2019constrained}. Under the topic of fraud detection, current researches are mainly focused on dense subgraph discovery techniques~\cite{jiang2023spade,hooi2016fraudar,shin2017densealert,gudapati2021search}. Under the topic of transaction flow networks, maximum flow is one of the fundamental problems in graph theory with many solutions~\cite{ford1956maximal,dinic1970algorithm,goldberg1988new,hochbaum2008pseudoflow,hochstein2007maximum,chen2022maximum}, and recently there are also a few works on flows under temporal constraints~\cite{hamacher2003earliest,akrida2019temporal}.
}

\tr{
To detect fraudulent activities, current research has predominantly focused on dense subgraph discovery techniques~\cite{jiang2023spade,hooi2016fraudar,shin2017densealert,gudapati2021search}. However, these methods often overlook the complexities of layering transfer schemes and temporal dependencies critical in fraud detection. While \cite{kosyfaki2021flow} addresses temporal patterns and traditional maximum flow algorithms~\cite{dinic1970algorithm,hochbaum2008pseudoflow,goldberg1988new,chen2022maximum} tackle layering transfers, distinguishing between benign and fraudulent actors remains a challenge, as legitimate users also add to the flow volume within communities. A pivotal shortcoming of these techniques is their failure to consider the flow density within transactions, a critical aspect for identifying the concentrated, abnormal activities indicative of fraud. 
}

\stitle{Query Processing on Temporal Graphs.} Recently, there has been a growing interest in query processing on temporal graphs from both the industry and research communities (\eg \cite{kosyfaki2021flow,wu2014path,wu2016reachability,zhang2019efficient,yuan2019constrained}). The most relevant work to this paper is Kosyfaki et al.~\cite{kosyfaki2021flow}. They defined two flow transfer models on temporal networks but only focused on the temporal flow problem between a single source $s$ and a single sink $t$. The authors’ solution for maximum temporal flow applies to graphs with a maximum of $10$K transactions, owing to the quadratic time complexity. In contrast, $\name$ aims to find the densest flow of two given sets $\Src$~and~$\Dst$.

\stitle{Dense Subgraph Detection.} Current research in detecting fraudulent activities has predominantly focused on dense subgraph discovery techniques, as detailed in studies like~\cite{jiang2023spade,hooi2016fraudar,shin2017densealert,gudapati2021search,liu2017holoscope,jiang2024spade+,chen2024rush,zhou2025efficient,chen2023densest,liang2026accelerated,Niu2025sansefficientdensest,Jiang2025communitydetectionin,Jiang2025dupinaparallel}. However, these methods often fail to account for the complexities of layering transfer schemes, crucial in fraud detection. Though the approach in~\cite{li2020flowscope} emphasizes detecting dense flows, it is confined to strictly k-partite data and lacks adaptability to more intricate network structures, as indicated by~\cite{tariq2023topology}. Furthermore, existing studies on dense flow detection tend to overlook the temporal dependencies in transaction networks, which can significantly impact detection effectiveness. \jiaxin{In addition to these algorithmic efforts, several labeled benchmarks have been released to support fraud and phishing detection in transaction graphs. Notably, the Ethereum phishing-label transaction dataset released by InPlusLab has been used for graph-based phishing and fraud detection in Ethereum networks~\cite{xblockEthereum}.} $\name$ sets itself apart as the first framework to focus on layering transfer schemes densest flow queries, incorporating temporal dependency for enhanced precision in fraud detection applications.

\stitle{Maximum Flow.} Over the past few decades, many solutions have been proposed to solve the maximum-flow problem. Ford et al.~\cite{ford1956maximal} introduced the first feasible-flow algorithm by iteratively finding the augmenting paths. Dinic~\cite{dinic1970algorithm} built a layered graph with a breadth-first search on the residual graph to return the maximum flow in a layered graph in $O(|V|^2|E|)$. Goldberg et al.~\cite{goldberg1988new} proposed the push-relabel method with a time complexity of $O(|V|^3)$. Hochbaum~\cite{hochbaum2008pseudoflow} used a normalized tree to organize all unsaturated arcs, then the process of finding augmenting paths can be accelerated to $O(|V||E|\log|V|)$ by incorporating the dynamic tree. Hochstein et al.~\cite{hochstein2007maximum} optimized the push-relabel algorithm and calculated the maximum flow of a network with $k$ crossings in $O(k^3|V|\log |V|)$ time. Chen et al.~\cite{chen2022maximum} introduced an algorithm capable of nearly linearly addressing maximum and minimum-cost flow problems by constructing flows through a sequence of approximate undirected minimum-ratio cycles. While this method represents a significant theoretical advancement, its practical implementation remains an open challenge. $\name$ stands out, differing in two significant ways: (a) $\name$ pioneers the incorporation of a densest flow query; and (b) $\name$ returns the maximum temporal flow in transaction networks with temporal dependencies. Hence, $\name$ operates orthogonally to the individual maximum flow algorithms.

\stitle{Temporal Flow.} There are a few studies on flows under temporal constraints. Horst et al.~\cite{hamacher2003earliest} investigate the effects of time-varying capacities on maximum flow in temporal graphs. Akrida et al.~\cite{akrida2019temporal} treat flow networks as ephemeral entities, with each edge available only during certain times, and analyze maximum flow within specific time intervals. The concept of earliest arrival flow, explored in studies~\cite{schmidt2014earliest,gale1959transient,wilkinson1971algorithm,skutella2009introduction}, focuses on determining the earliest time for a flow to travel from $s$ to $t$. $\name$ is the first framework to concentrate on flow density with a temporal flow dependency.


%% file: 9-conclusion.tex
\section{Conclusions}\label{sec:conclusions}

\revise{Grab practices advanced fraud detection techniques to combat sophisticated fraudulent activities in its transaction networks.} However, existing methods face challenges in scalability and effectiveness when dealing with large-scale data and complex fraud patterns. To address these issues, we propose a novel query called \SDMF{}, specifically designed to detect fraud in transaction networks. We introduce an efficient solution, $\name$, to solve \SDMF{}. $\name$ incorporates a network transformation technique that ensures the correct maximum temporal flow is returned, allowing seamless integration with existing maximum flow algorithms. Our experiments demonstrate that $\name$ is efficient and effective, with query evaluations up to three orders of magnitude faster than baseline algorithms. \revise{Deployment of $\name$ in operational environments has significantly improved fraud detection accuracy and processing speed, leading to reduced financial losses for Grab.}



%% file: bio.tex
\begin{IEEEbiography}[{\includegraphics[width=1in,height=1.2in,clip,keepaspectratio]{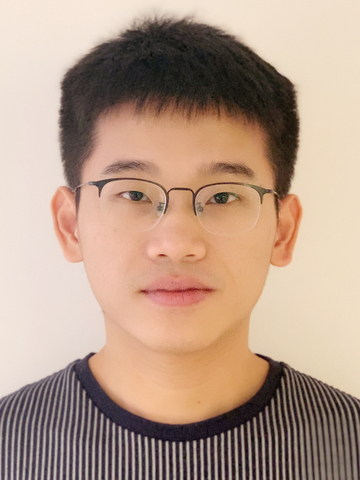}}]{Jiaxin~Jiang} is a senior research fellow in the School of Computing, National University of Singapore. He received his BEng degree in computer science and engineering from Shandong University in 2015 and PhD degree in computer science from Hong Kong Baptist University (HKBU) in 2020. His research interests include graph-structured databases, distributed graph computation and fraud detection.
\end{IEEEbiography}

\begin{IEEEbiography}[{\includegraphics[width=1in,height=1.2in,clip,keepaspectratio]{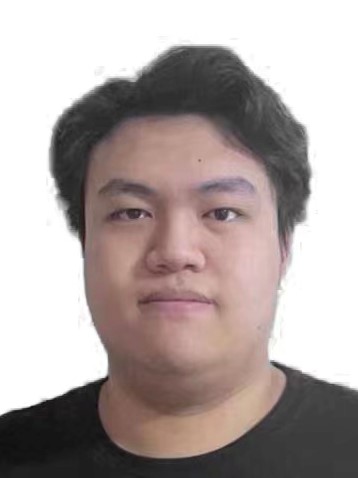}}]{Yunxiang~Zhao} recieved the BEng degree in computer science and engineering from East China Normal University in 2022. He is currently a PhD student in the Department of Computer Science, Hong Kong Baptist University. His research interests include graph data namagements and time series analysis.
\end{IEEEbiography}

\begin{IEEEbiography}[{\includegraphics[width=1in,height=1.2in,clip,keepaspectratio]{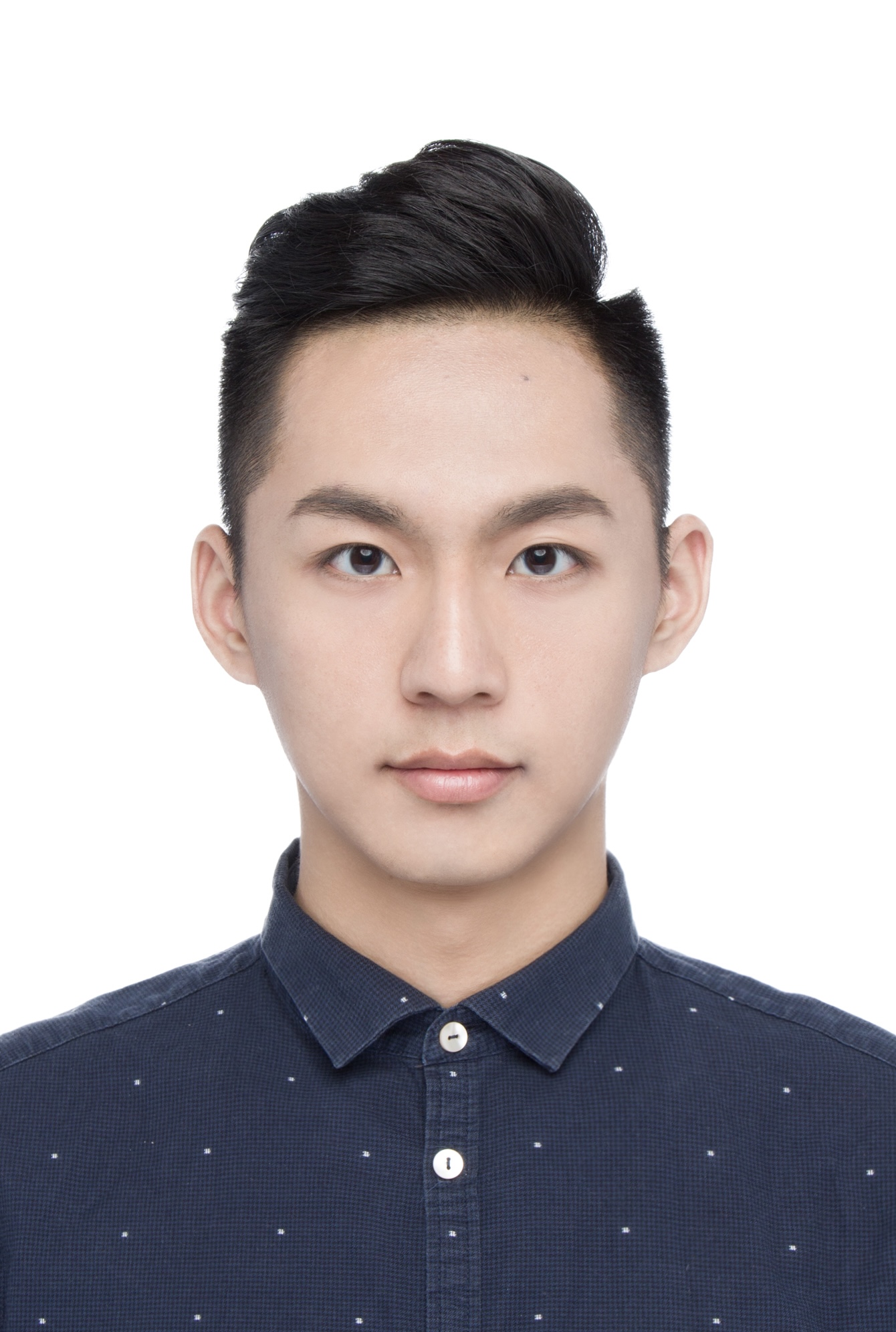}}]{Lyu~Xu} is currently a postdoctoral research fellow in the College of Computing and Data Science, Nanyang Technological University. He received his BEng degree in computer science and technology from Southeast University, MEng degree in computer science and technology from Sun Yat-sen University, and PhD degree in computer science from Hong Kong Baptist University (HKBU). His research interests include graph-structured databases and privacy preserving computation.
\end{IEEEbiography}

\begin{IEEEbiography}[{\includegraphics[width=1in,height=1.2in,clip,keepaspectratio]{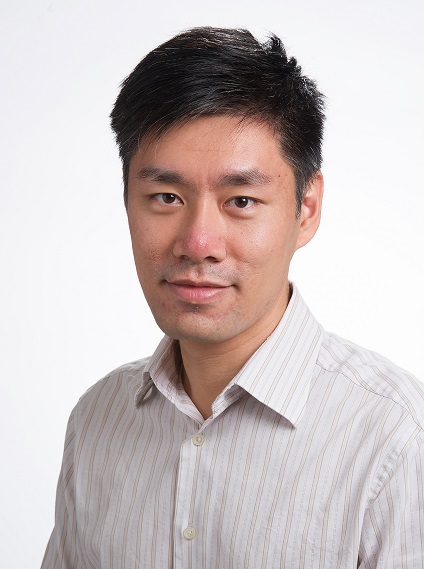}}]{Byron~Choi}
is a Professor in the Department of Computer Science at
the Hong Kong Baptist University. He received the bachelor of engineering degree
in computer engineering from the Hong Kong University of Science and Technology
(HKUST) in 1999 and the MSE and PhD degrees in computer and information science
from the University of Pennsylvania in 2002 and 2006, respectively. His research interests include graph data management and time series analysis.
\end{IEEEbiography}

\begin{IEEEbiography}[{\includegraphics[width=1in,height=1.2in,clip,keepaspectratio]{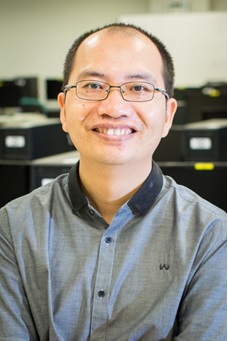}}]{Bingsheng~He} received the bachelor degree in computer science from Shanghai Jiao Tong University (1999-2003), and the PhD degree in computer science in Hong Kong University of Science and Technology (2003-2008). He is a professor in School of Computing, National University of Singapore. His research interests are high performance computing, distributed and parallel systems, and database systems. He is an ACM Distinguished member (class of 2020), and a fellow of IEEE (class of 2025). 
\end{IEEEbiography}

\begin{IEEEbiography}[{\includegraphics[width=1in,height=1.2in,clip,keepaspectratio]{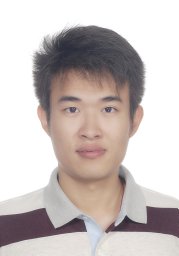}}]{Shixuan~Sun} is a Tenure-Track Associate Professor at the Department of Computer Science and Engineering, Shanghai Jiao Tong University. He received his Ph.D. in Computer Sciences from the Department of Computer Science and Engineering, Hong Kong University of Science and Technology (HKUST) in 2020. Prior to that, he got his M.S and B.S. in Computer Sciences from the School of Software Engineering, Tongji University in 2014 and 2011 respectively. His research interests are in database systems, graph algorithms, and parallel computing. Current focus is on building high-performance graph data management systems.
\end{IEEEbiography}

\begin{IEEEbiography}[{\includegraphics[width=1in,height=1.2in,clip,keepaspectratio]{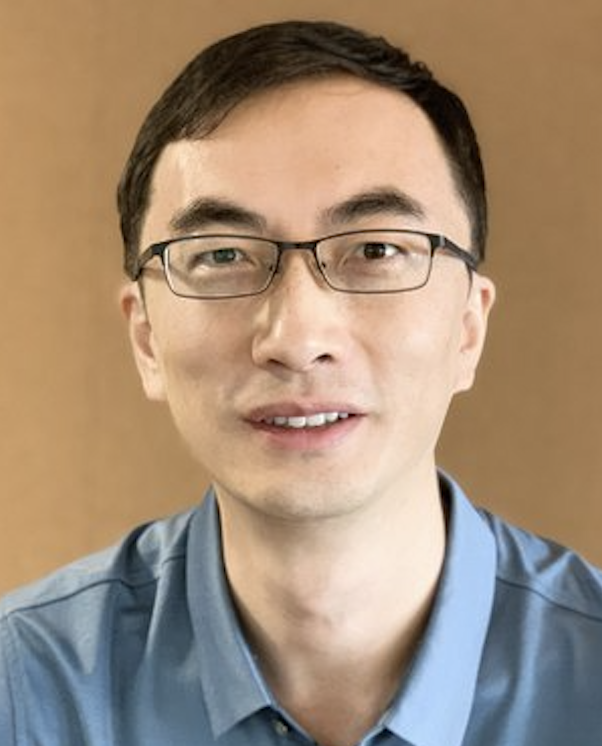}}]{Jia~Chen} received his Bachelor degree in Automation from Tsinghua University (1999-2003) and the PhD degree in Computer Science from Hong Kong University of Science and Technology (2005-2009). He is currently Head of Data Science, Integrity at Grab. His research interests include efficient machine learning for large scale, high dimensional and real time data, and generative models for image and text.
\end{IEEEbiography}

%% file: appendix.tex
\newpage
\clearpage
\appendix

\section{Proof}\label{sec:proof}

\subsection{NP-hardness}

\begin{figure}[h]
    \begin{center}
    \includegraphics[width=0.75\linewidth]{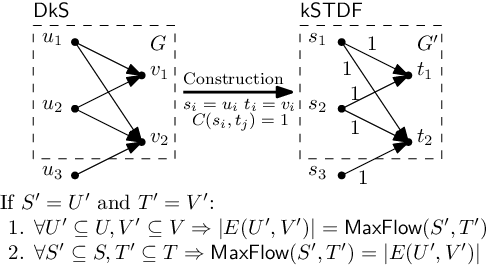}
    \end{center}
    \caption{Construction from $\DKS$ to \SDMF{}}
    \label{fig:construction}
\end{figure}

\begin{manuallemma}{\ref{lemma:nphard}}
	The decision problem of \SDMF{} is NP-complete.
\end{manuallemma}
\begin{proof}

The {\SDMF} problem is provably NP-complete, demonstrated through a reduction from the decision problem of the densest subgraph with at least $k$ vertices ($\DKS$ problem), as referred to in \cite{feige1997densest}. It's a known fact that the $\DKS$ problem is NP-complete even on bipartite graphs.

Given an instance of the $\DKS$ problem, denoted as $G=(U\cup V, E)$, we construct a corresponding instance of {\SDMF} as follows: We create a graph $G'=(S\cup T, E', C)$ to serve as the input for {\SDMF}, wherein $\Src = U$, $\Dst = V$, the edge set $(u,v)\in E'$ if and only if $(u,v)\in E$, and the capacity function $C(u,v) = 1$ for all $(u,v)\in E'$. This construction, which is graphically represented in Figure~\ref{fig:construction}, can be completed in polynomial time.

Let's define the induced subgraph of $U'\cup V'$ in $G$ as $E(U',V')$. It follows that $|E(U',V')| = \MFlow(\Src',\Dst')$, where $\Src' = U$ and $\Dst' = V$. Given that $|U'| + |V'| = |\Src'| + |\Dst'|$, the density of the induced subgraph by $U'\cup V'$ can be expressed as $\frac{|E(U',V')|}{|U'| + |V'|} = \frac{\MFlow(\Src', \Dst')}{|\Src'| + |\Dst'|} = g(\Src',\Dst')$. Therefore, if $|\Src'|+|\Dst'| \geq k$ and $g(\Src',\Dst')$ is maximized, then the density of the subgraph induced by $U'=\Src'$ and $V'=\Dst'$ is also maximized.

Assuming that $\Src'$ and $\Dst'$ can be determined in polynomial time, it then follows that $U'$ and $V'$ can also be found in polynomial time. However, this contradicts the NP-completeness of the $\DKS$ problem.

\end{proof}

\subsection{3-Approximation}\label{sec:proof:3appr}

\begin{manuallemma}{\ref{lemma:existence}}
    $\forall F \in [0, g(\Src, \Dst)]$, $\exists i \in [0,n]$, $(\Src_i, \Dst_i) = F\Core(\Src,\Dst)$.
\end{manuallemma}

\begin{proof}
We prove the lemma in contradiction by assuming that $F\Core(\Src,\Dst)$ does not exist, \ie $\delta_i < F$. We have the following contradiction.
\begin{align*}
\MFlow(\Src, \Dst) &= \MFlow(\Src_n, \Dst_n) = \MFlow(\Src_{n-1}, \Dst_{n-1}) + \delta_{n}  \\
&= \MFlow(\Src_{n-2}, \Dst_{n-2}) + \delta_{n} + \delta_{n-1} = \ldots  \\
&= \MFlow(\Src_{0}, \Dst_{0}) + \sum_{i=1}^{n} \delta_j =  \sum_{i=1}^{n} \delta_j \\
        &< nF \leq \frac{n\MFlow(\Src, \Dst)}{n} = \MFlow(\Src, \Dst) \\
\end{align*}
We can conclude that there exists $F\Core(\Src,\Dst)$ for $0 \leq F\leq g(\Src, \Dst)$.
\end{proof}

Given $F \in [0,g(\Src,\Dst)]$, there might exist several $F\Core$. In our following discussion, we only consider the $F\Core$ with the largest index $i$ for $\Src_i$ and $\Dst_i$, \ie $\delta_j < F$ for $j\in (i,n]$.

\begin{manuallemma}{\ref{lemma:apprcore}}
$\forall \alpha\in [0,1]$ and $F = \alpha g(\Src, \Dst)$, $\MFlow(\Src^F,\Dst^F) \ge (1-\alpha)\MFlow(\Src, \Dst)$, where $(\Src^F,\Dst^F)=F\Core(\Src,\Dst)$.
\end{manuallemma}

\begin{proof}
We consider the maximum flow between $\Src$ and $\Dst$.
\begin{align*}
        \MFlow(\Src, \Dst) &= \sum\limits_{j=1}^{n} \delta_j = \sum\limits_{j=1}^{i} \delta_j + \sum\limits_{j=i+1}^{n} \delta_j\\
        &\le \MFlow(\Src_i, \Dst_i) + (n-i)F  \\ 
        &\leq \MFlow(\Src^{F}, \Dst^{F}) + nF
\end{align*}
Therefore, we have the following conclusion. 
    \begin{align*}
        \MFlow(\Src^{F}, \Dst^{F}) & \ge \MFlow(\Src, \Dst) -  nF = \MFlow(\Src, \Dst) - n\alpha g(\Src,\Dst)   \\
        & = \MFlow(\Src, \Dst) - \frac{n\alpha\MFlow(\Src, \Dst)}{n} \\ 
        &  = (1-\alpha)\MFlow(\Src, \Dst) 
    \end{align*}
\end{proof}

\begin{manualtheorem}{\ref{theorem:appr}}
    Algorithm~\ref{algo:greedy} is a $3$-approximation algorithm.
\end{manualtheorem}
\begin{proof}
Consider the peeling sequence $\{(\Src_n,\Dst_n), \ldots, (\Src_0,\Dst_0)\}$. Let $(\Src',\Dst')$ represent the exact solution to the \SDMF{} problem. For any integer $k$, there exists $i \in [k, n]$ such that $g(\Src_i, \Dst_i) \geq \frac{g(S', T')}{3}$.

Define $\alpha = \frac{2}{3}$ and let $F = \frac{2g(\Src', \Dst')}{3}$. According to Lemma~\ref{lemma:apprcore}, we have:
\begin{equation}
\MFlow({\Src'}^F, {\Dst'}^F) \ge \frac{\MFlow(\Src', \Dst')}{3}    
\end{equation}

Since $\Src'$ and $\Dst'$ are subsets of $\Src$ and $\Dst$ respectively, we deduce that 
\begin{equation}\label{eq:fcore3appr}
    \MFlow(\Src^F, \Dst^F) \ge \frac{\MFlow(\Src', \Dst')}{3}    
\end{equation}

By definition, the peeling flow of any vertex $u$ is no less than $F$. Consequently, the total incoming flow to $\Dst^F$ equals the outgoing flow from $\Src^F$, leading to:
\begin{align*}
g(\Src^F,\Dst^F) &\ge \frac{(|\Src^F| + |\Dst^F|) \times F}{2\times(|\Src^F| + |\Dst^F|)} \\
&\ge \frac{g(\Src',\Dst')}{3}.
\end{align*}

From Lemma~\ref{lemma:existence}, there exists some $i$ such that $(\Src_i, \Dst_i) = F\Core(\Src, \Dst)$, \ie $\Src_i=\Src^F$ and $\Dst_i=\Dst^F$.

\stitle{Case 1:} If $i \geq k$, then $(\Src_i, \Dst_i)$ satisfies all requirements.

\stitle{Case 2:} If $i < k$, then $(\Src_k, \Dst_k)$ is a 3-approximation of the answer. Since $(\Src',\Dst')$ is the exact answer and due to the definition, we have
\begin{equation}\label{eq:ksize}
 |\Src'|+|\Dst'| \geq k   
\end{equation}
Then we have:
\begin{align*}
g(\Src_k, \Dst_k) &= \frac{\MFlow(\Src_k,\Dst_k)}{k}   \\
& \geq \frac{\MFlow(\Src_i,\Dst_i)}{k} & \because \Src_i\subseteq \Src_k, \Dst_i\subseteq \Dst_k\\
&\geq \frac{\MFlow(\Src',\Dst')/3}{k} & \because \textnormal{Equation}~\ref{eq:fcore3appr} \\
& \geq \frac{g(\Src',\Dst')}{3} & \because  \textnormal{Equation}~\ref{eq:ksize}
\end{align*}

In either case, we establish that Algorithm~\ref{algo:greedy} is a 3-approximation algorithm, concluding the proof.
\end{proof}

\tr{
\begin{property}\label{property:sameside1}
     $\forall s', s\in \Src$, $\mathsf{PF}(s',\Src\setminus\{s\},\Dst)\ge \mathsf{PF}(s',\Src,\Dst)$.
\end{property}
\begin{property}\label{property:sameside2}
     $\forall t', t\in \Dst$, $\mathsf{PF}(t',\Src,\Dst\setminus\{t\})\ge \mathsf{PF}(t',\Src,\Dst)$.
\end{property}
}

\begin{manualproperty}{\ref{property:otherside}}
    $\forall s\in \Src, t\in \Dst$, 1) $\mathsf{PF}(t,\Src\setminus\{s\},\Dst)\ge \mathsf{PF}(t,\Src,\Dst)-\mathsf{PF}(s,\Src,\Dst)$; and 2) $\mathsf{PF}(s,\Src,\Dst\setminus\{t\})\ge \mathsf{PF}(s,\Src,\Dst) - \mathsf{PF}(t,\Src,\Dst)$.
\end{manualproperty}

\begin{proof}
With the definition of peeling flow, we have
\begin{equation}\label{eq:tst}
\begin{split}
\PF(t,\Src,\Dst) & =  \MFlow(\Src,\Dst)-\MFlow(\Src,\Dst\setminus\{t\})\\
\end{split}
\end{equation}
\begin{equation}\label{eq:sst}
\begin{split}
\PF(s,\Src,\Dst) & =  \MFlow(\Src,\Dst)-\MFlow(\Src\setminus\{s\},\Dst)\\
\end{split}
\end{equation}
Combining Equation~\ref{eq:tst} and Equation~\ref{eq:sst}, we have the following.
\begin{equation}\label{eq:combine}
    \PF(t,\Src,\Dst) - \PF(s,\Src,\Dst) =\MFlow(\Src\setminus\{s\},\Dst) -  \MFlow(\Src,\Dst\setminus\{t\})
\end{equation}
Due to property~\ref{property:peelnotneg}, we have
\begin{equation}
    \PF(s,\Src,\Dst\setminus\{t\}) = \MFlow(\Src,\Dst\setminus\{t\}) - \MFlow(\Src\setminus\{s\},\Dst\setminus\{t\}) \ge 0
\end{equation}
Therefore, $\MFlow(\Src,\Dst\setminus\{t\}) \geq \MFlow(\Src\setminus\{s\},\Dst\setminus\{t\})$. Moreover, with Equation~\ref{eq:combine}, we have the following.
\begin{equation}
\begin{split}
\PF(t,\Src\setminus\{s\},\Dst) & =  \MFlow(\Src\setminus\{s\},\Dst)-\MFlow(\Src\setminus\{s\},\Dst\setminus\{t\})\\
& \geq \MFlow(\Src\setminus\{s\},\Dst)-\MFlow(\Src,\Dst\setminus\{t\})\\
& = \PF(t,\Src,\Dst) - \PF(s,\Src,\Dst)
\end{split}
\end{equation}
Similarly, we have $\mathsf{PF}(s,\Src,\Dst\setminus\{t\})\ge \mathsf{PF}(s,\Src,\Dst) - \mathsf{PF}(t,\Src,\Dst)$.
\end{proof}

\subsection{\jiaxin{Case Study: Densest Commute Flows in Grab-Posisi}}
\label{sec:posisi}

\jiaxin{\stitle{Dataset and graph construction.}
To demonstrate that \SDMF{} is not restricted to financial transaction graphs, we conduct a case study on Grab-Posisi, a GPS trajectory dataset collected from Grab’s ride-hailing operations in Southeast Asia~\cite{huang2019grab}. We focus on the morning peak period (7--10\,am) and discretize the city into coarse-grained urban zones corresponding to well-known residential and business areas. Each zone is represented as a vertex, and a directed edge $(u,v)$ aggregates commute trips from origin zone $u$ to destination zone $v$. The capacity $C(u,v)$ equals the number of trips during the time window, and the timestamp $\T(u,v)$ is derived from the departure times of the trips, yielding a temporal flow network.}

\noindent \jiaxin{\stitle{Query formulation.}
Residential zones with high outbound commute volume (e.g., Jurong East, Clementi, Holland Village, Commonwealth, Redhill) are selected as the source set $\Src$, while major business districts with high inbound volume (e.g., Outram Park, City Center, Raffles) form the sink set $\Dst$. We issue an \SDMF{} query $(\Src,\Dst,k)$ with a minimum group size constraint $k$ to identify subsets $\Src'\subseteq\Src$ and $\Dst'\subseteq\Dst$ that maximize commute flow density under temporal feasibility.}

\noindent \jiaxin{\stitle{Findings.}
Figure~\ref{fig:posisi} visualizes the densest commute flow returned by \SDMF{}. Rather than selecting all high-volume zones, \SDMF{} identifies a compact subset of residential sources (e.g., Redhill and Holland Village) and business sinks (e.g., City Center and Outram Park) connected by a small number of dominant commute paths. These paths form a coherent west-to-central commute structure, reflecting the primary morning travel direction into Singapore’s CBD. Compared with a volume-based baseline that ranks zones independently, the \SDMF{} result (i) concentrates substantially more flow per selected zone and (ii) yields a clearer, corridor-like spatial pattern instead of scattered origin--destination pairs. This case study illustrates that \SDMF{} captures structurally meaningful dense flows in urban mobility networks, demonstrating its applicability beyond financial transaction graphs.}

\begin{figure}[tb]
  \centering
  \includegraphics[width=0.85\linewidth]{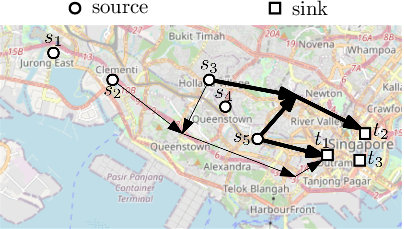}
  \caption{\jiaxin{Densest commute flow between residential and business zones on Grab-Posisi (morning peak) ($s_1$: Jurong East; $s_2$: Clementi; $s_3$: Holland Village; $s_4$: Commonwealth; $s_5$: Redhill; $t_1$: Outram Park; $t_2$: City Center; $t_3$: Raffles)}}
  \label{fig:posisi}
\end{figure}

\input{6.5-optimization}

%% file: 6.5-optimization.tex
\stab
\stab
\stab
\stab

\input{3.1-networdreduction}

\jiaxin{\subsection{Network Compression}}\label{sec:compression}

In this subsection, we propose a flow-preserving compression technique to reduce the size of the network. For ease of analysis, the vertices in \RTFNet{} are classified into three categories based on their in-degree and out-degree. We denote in-degree by \(\degi\) and out-degree by \(\dego\).

\begin{itemize}[leftmargin=*]
    \item \stitle{Flow-out Vertices}: $u$ is a flow-out vertex if 1) $\degi(u)$ $=0$, or 2) $\degi(u)$$=1$ and $C(v,u) = +\infty$, for $(v,u)\in E$.
    \item \stitle{Flow-in Vertices}: $u$ is a flow-in vertex if 1)~$\dego(u) = 0$, or 2) $\dego(u)=1$ and $C(u,v) = +\infty$, for $(u,v)\in E$.
    \item \stitle{Flow-crossing Vertices}: $u$ is a flow-crossing vertex if \red{$\exists (u,v_1), (v_2,u)\in E$, $C(u,v_1)\not=+\infty$ and $C(v_2,u)\not=+\infty$}.
\end{itemize}

\stitle{Vertex Compression ($\CR(G, u_1,u_2)$).} Given an \RTFNet{}, $G=(V,E,C)$, the compression of two vertices $u_1$ and $u_2$ has the following steps: 
\begin{enumerate}
    \item Add a super vertex $u$ to $V$.
    \item If $(v,u_1)\in E$ or $(v,u_2)\in E$, add $(v,u)$ to $E$; If $(u_1,v)\in E$ or $(u_2,v)\in E$, add $(u,v)$ to $E$.
    \item $u_1$, $u_2$, and the edges adjacent to them are removed from $E$. The compression process is denoted by $(G',u)=\CR(G,u_1,u_2)$, where $G'$ is the compressed graph and $u$ is the super vertex.
\end{enumerate}
\eat{
1) add a super vertex $u$ to $V$. 2) If $(v,u_1)\in E$ or $(v,u_2)\in E$, add $(v,u)$ to $E$; If $(u_1,v)\in E$ or $(u_2,v)\in E$, add $(u,v)$ to $E$. 3) $u_1$, $u_2$, and the edges adjacent to them are removed from $E$. The compression process is denoted by $(G',u)=\CR(G,u_1,u_2)$, where $G'$ is the compressed graph and $u$ is the super vertex.}

Given a vertex compression $(G',u)=\CR(G,u_1,u_2)$, if \red{$u_1,u_2\not \in \Src\cup \Dst$} and $\MFlow(s,t)$ on $G$ is equal to that on $G'$ for any source $s$ and sink $t$, then we say $\CR(G,u_1,u_2)$ is {\em flow-preserving}.

\eat{(where $s,t\neq u_1,u_2$)}

\stitle{Network Compression (Algorithm~\ref{algo:compression}, Figure~\ref{fig:compress}).} We have discovered that there are five types of flow-preserving vertex compression. To illustrate these types in the compression, we color the flow-out vertices \textit{white} and flow-in vertices \textit{black} in Figure~\ref{fig:compress}. The network compression of $\name$ scans each vertex $u_1\in V$ and checks if it can be compressed with any of its neighboring vertex $u_2$. If the compression of $u_1$ and $u_2$ fits any of the five flow-preserving cases, $\name$ compresses them using the function $\CR(G,u_1,u_2)$. The network compression process ends when no more vertices can be compressed. The time complexity of Algorithm~\ref{algo:compression} is $O(|V| + |E|)$.

\begin{lemma}\label{lemma:5cases}
    Given a graph $G$ and an edge $(u_1,u_2)\in E$, let $(G',u)=\CR(G,u_1,u_2)$. Then, $(G',u)$ is flow-preserving in the following cases:
    \begin{itemize}[leftmargin=*]
        \item \textbf{Case a.} $u_1$ and $u_2$ are both flow-in vertices; and $u$ is  a flow-in vertex.
        \item \textbf{Case b.} $u_1$ and $u_2$ are both flow-out vertices; and $u$ is a flow-out vertex.
        \item \textbf{Case c.} $u_1$ is a flow-in vertex; and $u_2$ is a flow-out vertex.
        \item \textbf{Case d.} $u_1$ is a flow-in vertex; and $u_2$ is a flow-crossing vertex.
        \item \textbf{Case e.} $u_1$ is a flow-crossing vertex; and $u_2$ is a flow-out vertex.
    \end{itemize}
\end{lemma}

\eat{
\begin{lemma}\label{lemma:5cases}
    Given a $G$ and an edge $(u_1,u_2)\in E$, the following hold.
    \begin{itemize}
        \item \textbf{Case a.} If $u_1,u_2$ are both flow-in vertices, $(G',u)=\CR(G,u_1,u_2)$ is flow-preserving and $u'$ is also a flow-in vertex.
        \item \textbf{Case b.} If $u_1,u_2$ are both flow-out vertices, $(G',u)=\CR(G,u_1,u_2)$ is flow-preserving and $u'$ is a flow-out vertex.
        \item \textbf{Case c.} If $u_1$ is a flow-in vertex and $u_2$ is a flow-out vertex, $(G',u)=\CR(G,u_1,u_2)$ is flow-preserving.
        \item \textbf{Case d.} If $u_1$ is a flow-in vertex and $u_2$ is a flow-crossing vertex, $(G',u)=\CR(G,u_1,u_2)$ is flow-preserving.
        \item \textbf{Case e.} If $u_1$ is a flow-crossing vertex or $u_2$ is a flow-out vertex, $(G',u)=\CR(G,u_1,u_2)$ is flow-preserving.
    \end{itemize}
\end{lemma}
}

\stitle{Proof Sketch.} The correctness of Lemma~\ref{lemma:5cases} can be established by the following two lemmas:
\begin{enumerate}
    \item If $u_1$ is flow-in vertex or $u_2$ is flow-out vertex, $\CR(G,u_1,u_2)$ is flow-preserving. It is because either $u_1$ has no other out-edge or $u_2$ has no other in-edge, $f(u_2,u_1)=0$ regardless of $C(u_2,u_1)$, and thus, $\CR(G,u_1,u_2)$ is flow-preserving. 
    \item If $u_1$ and $u_2$ are both flow-in (resp. flow-out) vertices, then for $(G',u)=\CR(G,u_1,u_2)$, $u$ is a flow-in (resp. flow-out) vertex. By the definition of vertex compression, $u$ has at most one outgoing edge and its capacity is $+\infty$, which follows the definition of flow-in (resp. flow-out) vertices.
\end{enumerate}

The detailed proof is shown in Appendix~\ref{sec:theorem1} of \cite{techreport}.

\begin{algorithm}[tb]
    \caption{Network Compression}\label{algo:compression}
    \footnotesize
    \SetKwProg{Fn}{Function}{}{}
    \KwIn{An \RTFNet{} $G=(V, E, C)$}
    \KwOut{A compressed network $G'$}
    \ForEach(){$u_1 \in V$}{
        \ForEach(){$u_2 \in N(u_1)$}{
            \If{compression is flow-preserving for $(u_1, u_2)$}{
                $(G',u) = \CR(G,u_1,u_2)$ \\
                \tcp{Update color based on conditions}
                \uIf{$u_1$ \textnormal{and} $u_2$ \textnormal{are both black}}{
                    color $u$ black \tcp*[h]{Case a}
                }
                \uElseIf{$u_1$ \textnormal{and} $u_2$ \textnormal{are both white}}{
                    color $u$ white \tcp*[h]{Case b}
                }
                \ElseIf{$u_1$ \textnormal{is black and} $u_2$ \textnormal{is white} \textnormal{or} $u_1$ \textnormal{or} $u_2$ \textnormal{is gray}}{
                    color $u$ gray \tcp*[h]{Cases c, d, e}
                }
            }
        }
    }
    \Return $G$
\end{algorithm}

\begin{figure}[tb]
	\begin{center}
	\includegraphics[width=0.8\linewidth]{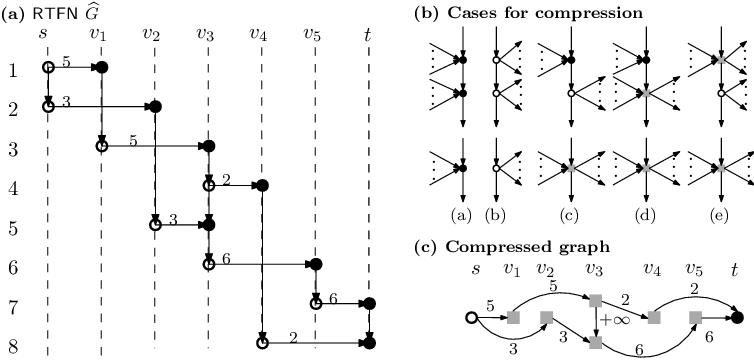}
		\end{center}
  \vspace{-3mm}
	\caption{The compressed graph of Figure~\ref{fig:transformed}(a)}
	\label{fig:compress}
\end{figure}

\begin{example}
    Consider the \RTFNet{} shown in Figure~\ref{fig:transformed}(a), $s^1$, $s^2$, $v_1^3$, $v_2^5$, $v_3^4$, $v_3^6$, $v_4^8$ and $v_5^7$ are flow-out vertices whereas $v_1^1$, $v_2^2$, $v_3^3$, $v_4^4$, $v_3^5$, $v_1^3$, $v_5^6$, $t^7$ and $t^8$ are flow-in vertices. We highlight the flow-out vertices white and flow-in vertices black as shown in Figure~\ref{fig:compress}(a). Five flow-preserving vertex compression cases are presented in Figure~\ref{fig:compress}(b). As the compression of $s^1$ and $s^2$ is an instance of \textbf{Case (b)}, they are compressed as illustrated in Figure~\ref{fig:compress}(c). Similarly, $(t^7,t^8)$ is an instance of \textbf{Case (a)} and $(v_1^1,v_1^3)$ is an instance of \textbf{Case (c)}, they are compressed accordingly. The compressed graph is shown in Figure~\ref{fig:compress}(c).
\end{example}
\eat{
\stitle{Time complexity.} The time complexity of Algorithm~\ref{algo:compression} is $O(|V| + |E|)$, as each edge is traversed at most once and the checking of the five cases at each edge takes constant time.
}

\subsection{The proof of network processing}\label{sec:theorem1}

\begin{manualtheorem}{\ref{theorem:identical}} 
    Given a \TFNet{} $G$ with a source $s$ and a sink $t$, the maximum flow value \MFlow$(s,t)$$=$\MFlow$(s^{\ts_1},t^{\ts_{\max}})$ in the \RTFNet{} $\Transform$.    
\end{manualtheorem}
\begin{proof}
    We prove this theorem by contradiction. Assume that the maximum flow from $s^{\ts_1}$ to $t^{\ts_{\max}}$ on \RTFNet{} is $\widehat{f}$ and the maximum flow from $s$ to $t$ on \TFNet{} is $f$. 1) If $|\widehat{f}|~<|f|$, there exists a flow $\widehat{f'}$ on \RTFNet{}, such that $|\widehat{f'}|~=|f|$ by Lemma~\ref{lemma:tfnet}. $|\widehat{f'}|~>|\widehat{f}|$ contradicts the assumption that $\widehat{f}$ is the maximum flow on \RTFNet{}. 2) If $|\widehat{f}|~>|f|$, $f$ is not the maximum flow on \TFNet{} using Lemma~\ref{lemma:rtfnet} which contradicts the assumption. Hence, we conclude that $|\widehat{f}|~=|f|$.
\end{proof}

\begin{lemma}\label{lemma:tfnet}
Given a \TFNet{} $G$, two vertex $s$ and $t$, and any temporal flow $f$ between $s$ and $t$, there exists a flow $\hat{f}$ with the same value between $s^{\tau_1}$ and $t^{\tau_{max}}$ in \RTFNet{} $\Transform$, following the transformation function $\Tran$, where $s^{\tau_1}$ is the first copy of $s$ and $t^{\tau_{max}}$ is the last copy of $t$.
\end{lemma}
\begin{proof}
    We prove this lemma using the construction method. 

The transformation function (resp. the transformation reverse function) denoted as $\Tran$ (resp. $\Tran^{-1}$), is used to transform edges from \TFNet{} to \RTFNet{} (resp. from \RTFNet{} to \TFNet{}).

    \stitle{Flow construction.} 1) Consider the edge $e=(u,v)$ and $\widehat{E} = \Tran(e)$. We let $\hat{f}(\widehat{E}) = f(e)$. 2) Consider a vertex $u\in V$ and its copies $u^{\ts_1},u^{\ts_2},\ldots ,u^{\ts_k}$ in $\Transform$. We let $\yunxiang{\hat{f}} (u^{\ts_i},u^{\ts_{i+1}})=$ $\sum\limits_{v\in V} \hat{f}(v,u)$ - $\sum\limits_{v\in V} \hat{f}(u,v)$, where $T(v,u) \leq \ts_i$ and $T(u,v) \leq \ts_i$. We show that the construction preserves capacity constraint and flow conservation. 
    
    \stitle{Capacity constraint.} Consider edge $e=(u,v)$ and its transformed edge $\widehat{E} = \Tran(e)$. \textbf{Case 1:} $\widehat{C}(\widehat{E}) = C(e)$ implies that $\yunxiang{\hat{f}}(\widehat{E}) = f(e) \leq C(e) = \widehat{C}(\widehat{E})$. \textbf{Case 2:} Consider each vertical edge $\widehat{E} = (u^{\ts_i},u^{\ts_{i+1}})$. Since the temporal flow constraint in Definition~\ref{def:tflow}, $\yunxiang{\hat{f}}(\widehat{E})\yunxiang{\geq}0$. Since the $\widehat{C}(\widehat{E}) = +\infty$,  $0 \yunxiang{\leq\hat{f}} (\widehat{E}) < \widehat{C}(\widehat{E})$. Hence, the capacity constraint is preserved.
    
    \stitle{Flow conservation.} Consider a vertex $u\in V\setminus\{s,t\}$ and its copies $u^{\ts_1},u^{\ts_2},\ldots ,u^{\ts_k}$ in $\Transform$. For a copy $u^{\ts_{i+1}}$, the incoming flows are either from 1) $u^{\ts_i}$ or 2) the copies of other nodes. The outgoing flows are either to 1) $u^{\ts_{i+2}}$ or 2) the copies to other nodes.
    
    With the flow construction, the flow from $u^{\ts_i}$ is 
    \begin{equation}\label{eq:in1}
        \hat{f}(u^{\ts_i},u^{\ts_{i+1}})=\sum\limits_{v\in V, \ts \in (0, \ts_{i}]} f(v,u) - \sum\limits_{v\in V, \ts \in (0, t_{i}]} f(u,v)
    \end{equation}

    The flow from the copies of other nodes is 
    \begin{equation}
        \hat{f}(v^{\ts_{i+1}}, u^{\ts_{i+1}}) = \sum\limits_{v\in V, \ts =\ts_{i+1}} f(v,u)
    \end{equation}
        
    The flow to $u^{\ts_{i+2}}$ is 
    \begin{equation}
        \hat{f}(u^{\ts_{i+1}},u^{\ts_{i+2}})=\sum\limits_{v\in V, \ts \in (0, \ts_{i+1}]} f(v,u) - \sum\limits_{v\in V, \ts \in (0, \ts_{i+1}]} f(u,v)
    \end{equation}

    The flow to the copies of other nodes is 
    \begin{equation}\label{eq:out2}
        \hat{f}(u^{\ts_{i+1}},v^{\ts_{i+1}}) = \sum\limits_{v\in V, \ts =\ts_{i+1}} f(u,v)
    \end{equation}

    By the Equation~\ref{eq:in1}-\ref{eq:out2}, we have 
    \begin{equation}
        \hat{f}(u^{\ts_i},u^{\ts_{i+1}}) + \hat{f}(v^{\ts_{i+1}}, u^{\ts_{i+1}}) = \hat{f}(u^{\ts_{i+1}},u^{\ts_{i+2}}) + \hat{f}(u^{\ts_{i+1}},v^{\ts_{i+1}})
    \end{equation}
      The sum of the incoming flow is equal to the outgoing flow. Hence, the flow conservation is preserved. 
\end{proof}

\begin{lemma}\label{lemma:rtfnet}
    Given a \RTFNet{} $\hat{G}$, two vertex $s$ and $t$, and a flow $\hat{f}$ between $s$ and $t$. We can construct a flow $f$ with the same value in the corresponding \TFNet.
\end{lemma}
\begin{proof}
    We prove this lemma by construction method.

    \stitle{Construction.} Consider edge $\widehat{E} = (u^{\ts_i}, v^{\ts_i}) \in \widehat{E}$ and the original temporal edge $e=\Tran^{-1}(\widehat{E})$. If $\hat{f}(\widehat{E})\not=0$, we let $f(e)=\hat{f}(\widehat{E})$. We show that the construction preserves capacity constraint, temporal flow constraint, and flow conservation.

    \stitle{Capacity constraint.} $\widehat{C}(\widehat{E}) = C(e)$ implies that $f(e) = \hat{f}(\widehat{E}) \leq \widehat{C}(\widehat{E}) = C(e)$. Therefore, the capacity constraint is preserved.

    \stitle{Temporal flow constraint.} Given a vertex $u$, we collapse all the copies $u^{\ts_1},u^{\ts_2},\ldots, u^{\ts_k}$. Consider the timestamp $\ts_i$: 
    
    \begin{equation}
    \hat{f}(u^{\ts_i}, u^{\ts_{i+1}}) =\sum\limits_{v\in V,\ts \in (0,\ts_i]} f(v,u) - \sum\limits_{v\in V, \ts \in (0,\ts_i]} f(u,v)    
    \end{equation}
    
    $\hat{f}(u^{\ts_i}, u^{\ts_{i+1}})$ is non-negative, we have 
    
    \begin{equation}
    \sum\limits_{v\in V,\ts \in (0,\ts_i]} f(v,u) \geq \sum\limits_{v\in V, \ts \in (0,\ts_i]} f(u,v)    
    \end{equation}
    Therefore, the temporal flow constraint is preserved.

    \stitle{Flow conservation.} The flow conservation is preserved that there is not outgoing from of $u$ at timestamp $\ts_k$:
    
    \begin{equation}
    \sum\limits_{v\in V,\ts \in (0,\ts_i]} f(v,u) = \sum\limits_{v\in V, \ts \in (0,\ts_i]} f(u,v)    
    \end{equation}
\end{proof}

\begin{lemma}
    Given a \TFNet{} $G=(V,E,C,\T)$ and its \RTFNet{} $\Transform=(\widehat{V}, \widehat{E}, \widehat{C})$, the size of $\widehat{V}$ is bounded by $2|E|$ and the size of $\widehat{E}$ is bounded by $3|E| - |V|$.
\end{lemma}
\eat{
\begin{manuallemma}{\ref{lemma:space}}
    Given a \TFNet{} $G=(V,E,C,\T)$ and its \RTFNet{} $\Transform=(\widehat{V}, \widehat{E}, \widehat{C})$, the size of $\widehat{V}$ is bounded by $2|E|$ and the size of $\widehat{E}$ is bounded by $3|E| - |V|$.
\end{manuallemma}
}

\begin{proof}
    Consider a vertex $v\in V$, there are $\deg(v)$ copies in $\hat{G}$. Hence, the number of the vertices in $\hat{G}$ is $\sum_{v\in V}\deg(v) = 2|E|$. The number of the vertical edges of a vertex $v\in V$ is $\mathsf{deg}(v)-1$, where $\mathsf{deg}(v)$ is the degree of $v$. Hence, there are $\sum_{v\in V}\deg(v)-1 = 2|E| - |V|$. And the number of horizontal edges is equal to $|E|$. Therefore, the number of $\widehat{E}$ is $3|E| - |V|$.
\end{proof}

\begin{lemma}\label{lemma:dag}
    \RTFNet{} $\Transform=(\widehat{V}, \widehat{E}, \widehat{C})$, $\Transform$ is a DAG.
\end{lemma}

\begin{proof}
    We prove this by contradiction. We assume there is a cycle in \RTFNet{}, denoted by $(v_1,\ldots, v_n)$, where $v_n = v_1$. Based on the graph transformation, consider an edge $e_i= (v_i, v_{i+1})$. The $y$-axis (timestamp) of $v_i$ is smaller than that of $v_{i+1}$ if $e$ is a vertical edge. The $y$-axis (timestamp) of $v_i$ is equal to that of $v_{i+1}$ if $e$ is a horizontal edge. Therefore, the timestamp of $v_1$ is smaller than or equal to that of $v_n (v_1)$ by induction. Since there is at least one vertical edge, the timestamp of $v_1$ is smaller than that of $v_n (v_1)$, which contradicts that the timestamp of each vertex is unique. Hence, \RTFNet{} is a DAG.
\end{proof}



\eat{
\begin{proof}
1) Since $\delta_i \geq f$, minimum flow coming out or going to $u$ is more than $f$, where $u\in \Src_i$ or $u\in \Dst_i$. Hence, $(\Src_i,\Dst_i) \subseteq \Core_f(S,T)$; 2) Consider vertices $\{v_n,\ldots, v_{i+1}\}$. The flow coming out or going to $v_n$ is smaller than $f$ based on the assumption. Hence $v_n$ is not included in $f\Core(S,T)$.  After $v_n$ is removed from $\Src$ and $\Dst$, $v_{n-1}$ will be removed too since $\delta_{n-1} < f$. Specifically, $\{v_n,\ldots, v_{i+1}\}$ are all removed by induction. $f\Core(S,T) \subseteq (\Src_i,\Dst_i)$. Hence, $(\Src_i,\Dst_i) = f\Core(S,T)$.
\end{proof}
}

\begin{lemma}\label{lemma:gray}
  Given a flow network $G$ and an edge $(u_1,u_2)\in E$. If $u_1$ is a flow-in  vertex or $u_2$ is a flow-out vertex, $(G',u')=\CR(G,u_1,u_2)$ is flow-preserving.
\end{lemma}

\begin{proof}

    \yunxiang{Due to the definition of flow-in vertex and flow-out vertex, we have $C(u,v)=+\infty$.} We denote the maximum-flow $\MFlow(s,t)$ on $G$ (resp.$G'$) by $f$ (resp. $f'$). First, we prove that $|f|~\leq|f'|$.
  
  \yunxiang{We assume that the minimum cut on $G'$ is $(S',T')$, and let $S=S'\setminus\{u'\}$, $T=T'\setminus\{u'\}$, we consider two cases.}

\begin{enumerate}[wide, labelwidth=!, labelindent=0pt]
	\itemsep0em 
    \item If $u'\in S'$, $(S\cup\{u,v\},T)$ is a cut with the same value on $G$.
    \item If $u'\in T'$, $(S,T\cup\{u,v\})$ is a cut with the same value on $G'$.
\end{enumerate}
  
  Then, we prove that $|f|\ge |f'|$.
  
  We assume that the minimum cut on $G$ is $(S,T)$, and let $S'=S\setminus\{u,v\}$, $T'=T\setminus\{u,v\}$, we consider $4$ situations.

\begin{enumerate}[wide, labelwidth=!, labelindent=0pt]
	\itemsep0em 
    \item If $u,v\in S$, $(S'\cup u',T')$ is a cut with the same value on $G'$.
    \item If $u,v\in T$, $(S',T'\cup u')$ is a cut with the same value on $G'$.
    \item If $u\in S, v\in T$, the minimum cut contains the edge between $u$ and $v$, the capacity of which is $+\infty$, therefore $\MFlow(G,s,t)=+\infty$, and thus $\MFlow(G,s,t)\ge \MFlow(G',s,t)$.
    \item If $u\in T, v\in S$, because $u$ only have one outgoing edge and it points to $v$, moving $u$ from $T$ to $S$ will not increase the cut, therefore $(S'\cup u',T')$ is a cut with the same value on $G'$.
\end{enumerate}
\end{proof}


\begin{lemma}\label{lemma:colorblack}
  Given a flow network $G$ and an edge $(u_1,u_2)\in E$. If $u_1,u_2$ are\eat{both} flow-in vertices, then for $(G',u')=\CR(G,u_1,u_2)$, $u'$ is also an flow-in vertex.
\end{lemma}

\begin{proof}
  \yunxiang{All edges going out from $u'$ is either from $u$ or $v$ from step 3 of combination. However, the only outgoing edge from $u$ points to $v$ is removed from step 1 (and if the outgoing edge of $v$ points to $u$, it's also removed), therefore $u'$ can either have no outgoing edge or have one outgoing edge with a capacity of $+\infty$ which originally points from $v$, satisfying the condition of a flow-in vertex.} 
\end{proof}

\begin{lemma}\label{lemma:colorwhite}
  \yunxiang{Given a flow network $G$ and an edge $(u_1,u_2)\in E$. If $u_1,u_2$ are\eat{both} flow-out vertices, then for $(G',u')=\CR(G,u_1,u_2)$, $u'$ is also a flow-out vertex.}
\end{lemma}

\begin{proof}
  \yunxiang{This lemma can also be proven by the reverse graph. In the reverse graph, all flow-in vertices change into flow-out vertices, and vice versa. We can prove that $u'$ is a flow-in vertex on the reverse graph ${G'}^T$ by lemma \ref{lemma:colorblack}, therefore $u'$ on $G'$ is a flow-out vertex.}
\end{proof}

\begin{theorem}
    \yunxiang{Given a \TFNet{} $G^T=(V,E,C,T)$ and its \RTFNet{} after compression $\Transform=(\widehat{V}, \widehat{E}, \widehat{C})$, $|\widehat{V}|\le |E| + |V|$, $|\widehat{E}|\le 2|E|$.}

\end{theorem}

\begin{proof}
    \yunxiang{By construction of \RTFNet{}, for a vertex $u\in V$, its copies form a sequence $\Seq=\{u^{\tau_1},u^{\tau_2},\dots,u^{\tau_{\deg(u)}}\}$ and each vertex is painted either black and white. The size of the sequence can be reduced from $k$ to $\lfloor\frac{k}{2}\rfloor+1$.}
    
    \yunxiang{We prove this by construction. First, combine all continuous vertices with the same color. Let $\Seq'$ denotes the sequence after this step, all adjacent vertices in $\Seq'$ then have different colors, and we have $|\Seq'|\le |\Seq|$. Next, we combine all black vertices with their succeeding white vertices in $\Seq'$. All vertices in $\Seq'$ except the possible leading white vertex and trailing black vertex combine with another vertex in this step. Let $\Seq''$ denote the sequence after this step. If $|\Seq'|$ is odd, there will be either a leading white vertex or a trailing black vertex, hence $|\Seq''|=\lfloor\frac{|\Seq'|}{2}\rfloor+1$. If $|\Seq'|$ is even, there will be either no nodes cannot be combined or a leading white vertex and a trailing black vertex at the same time, hence $|\Seq''|=\lfloor\frac{|\Seq'|}{2}\rfloor$ or $|\Seq''|=\lfloor\frac{|\Seq'|}{2}\rfloor+1$. Therefore, we have $|\Seq''|\le \lfloor\frac{|\Seq'|}{2}\rfloor+1\le \lfloor\frac{k}{2}\rfloor+1$.}
    
    \yunxiang{With this proven, there will be $\sum_{u\in V}\left(\lfloor\frac{deg(u)}{2}\rfloor+1\right)\le |E|+|V|$ vertices after compression. Also, for each pair of nodes combined, an edge with an infinite flow between the two combined nodes will also be removed, hence there will be at least $|E|-|V|$ edges removed, and $|\widehat{E}|\le 2|E|$.}
\end{proof}

\begin{manualproperty}{\ref{property:reachability}}
    Given two pairs of sources and sinks $(s_1,t_1)$ and $(s_2,t_2)$, and the corresponding maximum flows $f_1=$\textnormal{\MFlow}$(s_1, t_1)$ and $f_2=$\textnormal{\MFlow}$(s_2,t_2)$. If $f_1$ and $f_2$ are overlap-free, then $f_1+f_2=$\textnormal{\MFlow}$(\{s_1,s_2\}, \{t_1, t_2\})$.
\end{manualproperty}

\begin{proof}
\yunxiang{If $f_1+f_2  >  \MFlow(\{s_1,s_2\}, \{t_1, t_2\})$, we construct a flow $f'$ that $f'(u,v)=f_1(u,v)+f_2(u,v)$ with $|f'|=|f_1|+|f_2|>\MFlow(\{s_1,s_2\}, \{t_1, t_2\})$, and $f'$ follows all conditions in \ref{def:flow} because $\reach(s_1,t_2) = \mathsf{False}$ and $\reach(s_2,t_1) = \mathsf{False}$.} 
 \yunxiang{And if $\MFlow(s_1, t_1) + \MFlow(s_2,t_2)  <  \MFlow(\{s_1,s_2\}, \{t_1, t_2\})$, because $\reach(s_1,t_2) = \mathsf{False}$ and $\reach(s_2,t_1) = \mathsf{False}$, flow goes out from $s_1$ is equal to flow goes in to $t_1$, resp. $s_2$, $t_2$. Therefore either $f_{out}(s_1)>\MFlow(s_1,t_1)$ or $f_{out}(s_2)>\MFlow(s_2,t_2)$}
\end{proof}

%% file: 3.1-networdreduction.tex
\subsection{Network Reduction}\label{sec:reduction}
\jiaxin{
\begin{definition}[Temporal Flow]\label{def:tflow}
In a \TFNet{} $G = (V,E,C,\T)$, a temporal flow from a source $s \in V$ to a sink $t \in V$ is a pair of mapping functions $\left\langle f, \T \right\rangle$ satisfying the following conditions:
\begin{enumerate}[leftmargin=*]
    \item \stitle{Capacity Constraint:} \textnormal{Identical to Property~(1) in Def.~\ref{def:flow}.}
    \item \stitle{Flow Conservation:} \textnormal{At the final time $\ts_{\max}$, which marks the end of the period under consideration, the flow conservation property mirrors Property~(2) in Def.~\ref{def:flow}}.
    \item \stitle{\em Temporal Flow Constraint:} \textnormal{$\forall u \in V \setminus \{s,t\}$ and $\forall \ts \in [1,\ts_{\max}]$, the cumulative flow entering $u$ up to time $\ts$ is at least as much as the flow leaving $u$: \newline $\sum_{e=(v,u)\in E, \T(e) \leq \ts} f(v,u)$$\geq$$\sum_{e=(u,v)\in E, \T(e) \leq \ts} f(u,v)$.}
\end{enumerate}
\end{definition}
}

\jiaxin{
We propose a practical method to reduce the size of a \TFNet{}, $G$, to a smaller \TFNet{}, $G'$. We remark that the network reduction processing is \eat{assumed to have been }done before evaluating the \SDMF{} queries.
}

\jiaxin{
\stitle{Edge Reduction.} We note that no temporal flow can occur through an outdated outgoing (resp incoming) edge, as this would violate the temporal flow constraint (resp. flow conservation) when $\ts=\T(u,v)$ (resp. $\ts=\T(v,u)$). There are two kinds of outdated edges: $\forall u\in V$, a) its outgoing edge $e=(u,v)$ ($u\not \in \Src$ and $v\not \in \Dst$) is outdated if $\forall (v',u)\in E$, $\T(u,v) < \T(v',u)$; and b) its incoming edge $e=(v,u)$ ($v\not \in \Src$ and $u\not \in \Dst$) is outdated if $\forall (u,v')\in E$, $\T(v,u) > \T(u,v')$. All outdated edges are removed without affecting the answers to the \SDMF{} queries.
}

\jiaxin{
\stitle{Vertex Reduction.} A vertex $u$ is a \textit{flow inlet vertex} if its in-degree is $0$ and $u\not\in \Src$. $u$ is a \textit{flow outlet vertex} if its out-degree is $0$ and $u\not\in \Dst$. Such vertices do not participate in any temporal flow and are hence removed.
}

\begin{figure}[tb]
	\begin{center}
	\includegraphics[width=\linewidth]{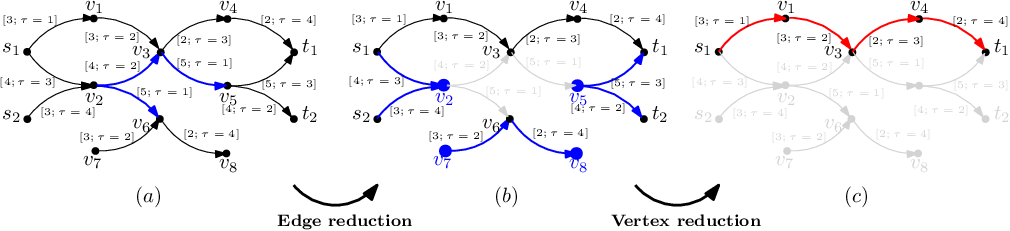}
		\end{center}
	\caption{A running example of network reduction}
	\label{fig:reduction}
\end{figure}

\jiaxin{
\stitle{Reduction Process.} The network reduction process in a \TFNet{}, $G$, can be summarized by two key steps. The first step identifies all outdated edges and flow inlet and outlet vertices with a time complexity of $O(|V| + |E|)$. Then, the second step removes these outdated edges and vertices and, subsequently, any isolated vertices. This process effectively reduces the network's size, making subsequent analysis more efficient.
}

\jiaxin{
\begin{example}
Consider a \TFNet{} with two sources, $\Src=\{s_1, s_2\}$, and two sinks, $\Dst=\{t_1,t_2\}$, as depicted in Figure~\ref{fig:reduction}(a). To reduce the network, we first identify outdated edges by examining the timestamps of each edge. Noting that $\T(v_2,v_6) = 1$ is less than both $\T(s_1,v_2) = 3$ and $\T(s_2,v_2) = 4$, edge $(v_2, v_6)$ is classified as outdated and is thus removed. Similarly, edges $(v_2,v_3)$ and $(v_3,v_5)$ are removed based on their timestamp comparisons as illustrated in Figure~\ref{fig:reduction}(b). In this process, $v_5$ and $v_7$ are identified as flow inlet vertices, and $v_2$ and $v_8$ as flow outlet vertices, leading to their removal. The reduced network is shown in Figure~\ref{fig:reduction}(c), on which the temporal flow \textnormal{\MFlow}$(\Src,\Dst) = 2$ is efficiently computed.
\end{example}
}